\documentclass[12pt]{article}
\usepackage{graphicx,amsfonts,amsmath,amssymb,amsthm,mathrsfs,url,hyperref}
\usepackage[usenames,dvipsnames]{xcolor}
\usepackage[toc,page]{appendix}

\title{Extended State Space for Describing Renormalized Fock Spaces in QFT}
\author{
Sascha Lill\footnote{Universit\`a degli Studi di Milano, Dipartimento di Matematica, Via Cesare Saldini 50, 20133 Milano, Italy}\ \footnote{ORCID: \href{https://orcid.org/0000-0002-9474-9914}{0000-0002-9474-9914} E-Mail: sascha.lill@unimi.it}~
} 
\date{\today}

\usepackage[utf8]{inputenc}
\usepackage{tikz}
\usetikzlibrary{decorations.pathreplacing}
\usepackage[compat=1.1.0]{tikz-feynhand}

\usepackage[section]{placeins}
\usepackage{siunitx}
\usepackage{capt-of, caption}
\usepackage{paralist}
\usepackage{booktabs}
\usepackage{hhline}
\usepackage{color}	

\newtheorem{lemma}{Lemma}
\numberwithin{lemma}{section}
\newtheorem{corollary}{Corollary}
\numberwithin{corollary}{section}
\newtheorem{theorem}{Theorem}
\numberwithin{theorem}{section}
\newtheorem{proposition}{Proposition}
\numberwithin{proposition}{section}

\numberwithin{remark}{section}
\theoremstyle{definition}\newtheorem{definition}{Definition}
\numberwithin{definition}{section}

\newcommand{\bk}{\boldsymbol{k}}

\newcommand{\bp}{\boldsymbol{p}}

\newcommand{\bx}{\boldsymbol{x}}
\newcommand{\by}{\boldsymbol{y}}

\newcommand{\bK}{\boldsymbol{K}}

\newcommand{\bP}{\boldsymbol{P}}

\newcommand{\bX}{\boldsymbol{X}}
\newcommand{\bY}{\boldsymbol{Y}}

\newcommand{\bLambda}{\boldsymbol{\Lambda}}

\newcommand{\cC}{\mathcal{C}}
\newcommand{\cD}{\mathcal{D}}
\newcommand{\cF}{\mathcal{F}}
\newcommand{\cI}{\mathcal{I}}
\newcommand{\cK}{\mathcal{K}}
\newcommand{\cN}{\mathcal{N}}

\newcommand{\cQ}{\mathcal{Q}}
\newcommand{\cS}{\mathcal{S}}
\newcommand{\cW}{\mathcal{W}}
\newcommand{\cX}{\mathcal{X}}

\newcommand{\sF}{\mathscr{F}}

\newcommand{\CCC}{\mathbb{C}}
\newcommand{\NNN}{\mathbb{N}}
\newcommand{\RRR}{\mathbb{R}}

\newcommand{\fc}{\mathfrak{c}}
\newcommand{\fh}{\mathfrak{h}}
\newcommand{\fr}{\mathfrak{r}}
\newcommand{\fR}{\mathfrak{R}}
\renewcommand{\Re}{\mathrm{Re}}
\renewcommand{\Im}{\mathrm{Im}}
\newcommand{\FB}{\overline{\sF}}
\newcommand{\FR}{\sF_{\ren}}
\newcommand{\WB}{\overline{\cW}}

\newcommand{\ad}{\mathrm{ad}}
\newcommand{\col}{\mathrm{col}}
\newcommand{\dom}{\mathrm{dom}}
\newcommand{\fin}{\mathrm{fin}}
\newcommand{\IBC}{\mathrm{IBC}}
\newcommand{\id}{I}
\newcommand{\Ker}{\mathrm{Ker}}
\newcommand{\loc}{\mathrm{loc}}
\newcommand{\Sdot}{\dot{\mathcal{S}}^\infty}

\newcommand{\Qdot}{\dot{\mathcal{Q}}}
\newcommand{\res}{\mathrm{res}}
\newcommand{\sgn}{\mathrm{sgn}}
\renewcommand{\span}{\mathrm{span}}
\newcommand{\supp}{\mathrm{supp}}

\newcommand{\eRen}{\mathrm{eRen}}
\newcommand{\ex}{\mathrm{ex}}
\newcommand{\Pol}{\mathrm{Pol}}
\newcommand{\ren}{\mathrm{ren}}
\newcommand{\Ren}{\mathrm{Ren}}
\newcommand{\renI}{\mathrm{renI}}
\newcommand{\F}{\mathrm{F}}
\newcommand{\Fex}{\mathrm{Fex}}
\newcommand{\Clas}{\mathrm{Clas}}

\DeclareMathOperator*{\esssup}{ess\,sup}

\newcounter{remarks}
\begin{document}

\maketitle
\begin{abstract}

We revisit the non-perturbative renormalization of a class of simple polaron models with resting fermions. The considered dispersion relations and form factors are allowed to be highly singular, such that infinite self-energies and wave function renormalizations may occur and the diagonalizing dressing transformations might not be implementable on Fock space. Instead of taking cutoffs, we rigorously interpret the self-energies and wave function renormalizations as elements of suitable vector spaces, as well as a field extension of the complex numbers. Moreover, we define two extended state vector spaces (ESSs) over this field extension, which contain a dense subspace of Fock space, but also incorporate non-square-integrable wave functions. Elements of these ESSs can be seen as states of virtual particles, described in a mathematically rigorous way.\\
The Hamiltonian without cutoffs, formally infinite counterterms, and the dressing transformation can then be defined as linear operators between certain subspaces of the two Fock space extensions. This way, we obtain a renormalized Hamiltonian which can be realized as a densely defined self-adjoint operator on Fock space.\\

\bigskip


\end{abstract}

\newpage
\tableofcontents

\section{Introduction}	
\label{sec:intro}

In this work, we revisit the renormalization of a certain class of formal quantum field theory\footnote{Here, the term ``QFT'' is to be understood as any quantum mechanical model whose description involves creation and annihilation operators. There is also a narrower sense in which a QFT is a model satisfying the Wightman axioms \cite[Ch.~3-1]{StreaterWightman1964} or the Haag--Kastler Axioms \cite{HaagKastler1964}. Our class of models in kept simple, as it serves for demonstrative purposes, and does not satisfy either of both axiom sets.} (QFT) Hamiltonians, which is done by employing a new mathematical framework that permits rigorous direct manipulations of formally infinite quantities. These quantities include divergent self-energies and wave function renormalizations, which are ubiquitous in QFT. The renormalization process is non-perturbative and does not involve any cutoffs.\\

The considered QFT Hamiltonians read as
\begin{equation}
	H = H_{0, y} + A^\dagger(v) + A(v) - E_\infty \;,
\label{eq:formalhamiltonian2}
\end{equation}
which means the following: We consider a species of $ M \in \NNN_0 $ fermions (associated with an index $ x $) and a species of $ N \in \NNN_0 $ bosons (associated with an index $ y $). The operators $ A^\dagger(v), A(v) $ make each fermion create/annihilate a boson with some form factor $ v $ (in momentum space), $ H_{0, y} $ describes the free evolution of the bosons and $ E_\infty $ is an \textbf{infinite self-energy counterterm}.\\
The Hamiltonian \eqref{eq:formalhamiltonian2} thus corresponds to a direct product of many van Hove Hamiltonians \cite{FewsterRejzner2020, Derezinski2003}, for which non-perturbative renormalization procedures are well-known to exist. But in contrast to preceding works, \eqref{eq:formalhamiltonian2} will not just be a formal expression, but a rigorous relation between linear operators that map $ \FB \to \FB_{\ex} $, where $ \FB \subset \FB_{\ex} $ are extensions of a dense subspace of Fock space. In particular, the multiplication by the (possibly infinite) self-energy becomes a well-defined operator $ E_\infty: \FB \to \FB_{\ex} $ .\\
In Section \ref{sec:furtherdress}, we will also encounter more general Hamiltonians of the form
\begin{equation}
	H_0 + A^\dagger(v) + A(v) - E_\infty \;,
\label{eq:formalhamiltonian3}
\end{equation}
with $ H_0 = H_{0, x} + H_{0, y} $ including both a fermionic and a bosonic dispersion relation.\\
Apart from the flat fermionic dispersion relation, our class of Hamiltonians \eqref{eq:formalhamiltonian2} allows to include the same dispersion relations and form factors as are used in many non-relativistic QFT models, including the Nelson model \cite{Nelson1964, Froehlich1973, Froehlich1974, U02}, the Fröhlich polaron \cite{U01}, as well as relativistic polarons \cite{Gross1973, Sloan1974}.\\

Our novel framework includes several mathematical tools that allow for rigorous manipulations with infinite quantities in QFT. These tools include:
\begin{itemize}
\item A \textbf{vector space of divergent integrals} $ \fr \in \Ren_1 $.
\item A \textbf{field of (possibly infinite) wave function renormalizations} $ \eRen $, containing exponentials $ e^\fr $ and fractions of linear combinations thereof. In particular, $ \eRen $ is a field extension of $ \CCC $.
\item An \textbf{extended state space (ESS)} $ \FB $, which is an $ \eRen $-vector space over smooth, complex-valued functions on configuration space, which are not necessarily $ L^2 $-integrable. $ \FB $ extends a dense subspace of the Fock space, $ \sF $.
\item A \textbf{second ESS} $ \FB_{\ex} $, which allows for multiplying $ \Psi \in \FB $ by elements $ \fr \in \Ren_1 $ (and not only by elements of $ \eRen $).
\end{itemize}

Neither the space $ \FB $ nor $ \FB_{\ex} $ carries a scalar product or even a topology. Instead, we assign physical meaning to subspaces of them by a \textbf{dressing transformation} $ W(s) $, which is a linear operator $ W(s): \FB_{\ex} \supset \cD_W \to \FB_{\ex} $ in form of a Gross transformation. The form factor $ s(\bk) $ does only depend on the boson momentum $ \bk $ and not on the fermion momentum $ \bp $, so one can fiber-decompose $ W(s) $ into several Weyl transformations, one for each fermion configuration $ \bX \in \RRR^{Md} $. The renormalized Hamiltonian $ \widetilde{H} $ will be defined without cutoffs by
\begin{equation*}
	W(s) \widetilde{H} = H W(s) \;,
\end{equation*}
i.e., we require that the following diagram commutes:

\begin{center}
	\begin{tikzpicture}

\node[anchor = east] at (0,0)  {$ \sF \supset \cD_\sF \supset \widetilde{\cD}_\sF $};
\node[anchor = east] at (0.4,1.5) {$ \FB_{\ex} \supset W(s)[\cD_\sF] \supset W(s)[\widetilde{\cD}_\sF] $};
\draw[thick, ->] (-0.3,0.3) --node[anchor = west] {$W(s)$} ++(0,0.9);
\draw[thick, ->] (0.2,0) -- node[anchor = south] {$\widetilde{H}$}  (2.8,0);

\node[anchor = west] at (3,0) {$ \cD_\sF \subset \sF $};
\node[anchor = west] at (2.6,1.5) {$ W(s)[\cD_\sF] \subset \FB_{\ex} $};
\draw[thick,->] (3.3,0.3) --node[anchor = west] {$W(s)$} ++(0,0.9);
\draw[thick, ->] (0.4,1.5) -- node[anchor = south] {$H$}  (2.6,1.5);

\end{tikzpicture}
\end{center}

The large domain $ \cD_\sF $, with the fermionic wave function being a Schwartz function, is defined in \eqref{eq:cDsF2}. It contains the small domain $ \widetilde{\cD}_\sF $ defined in \eqref{eq:cDsFtilde} via \eqref{eq:cCW0x}, with fermionic wave function in $ C_c^\infty $, and its support avoiding the set of collision configurations. It is necessary to introduce two domains, since $ \widetilde{H} $ does neither map $ \cD_\sF $ nor $ \widetilde{\cD}_\sF $ into itself. Nevertheless, $ \widetilde{H}: \widetilde{\cD}_\sF \to \cD_\sF $ can be defined. By Lemmas \ref{lem:cDsFdense} and \ref{lem:denselydefined}, both $ \cD_\sF $ and $ \widetilde{\cD}_\sF $ are dense in $ \sF $.\\

Our \textbf{main result} is Theorem \ref{thm:pullback}, which states that the renormalized Hamiltonian $ \widetilde{H} $ is indeed defined as a linear operator $ \widetilde{H}: \sF \supset \cD_{WS} \to \FB_{\ex} $ (with definition of the domain $ \cD_{WS} $ in \eqref{eq:cDWS}). By means of Lemma \ref{lem:denselydefined}, under certain conditions, this $ \widetilde{H} $ can indeed be interpreted as a Fock space operator $ \widetilde{H}: \widetilde{\cD}_\sF \to \cD_\sF $, which allows for self-adjoint extensions by Corollary \ref{cor:selfadjointext}. More precisely, we have
\begin{equation}
	\widetilde{H} = H_{0, y} + V \;,
\label{eq:HtildeESS}
\end{equation}
with $ V $ being a pair potential interaction between fermions.\\

The result \eqref{eq:HtildeESS} after a dressing is not too surprising, as \eqref{eq:formalhamiltonian2} just corresponds to a direct integral of van Hove Hamiltonians. It is well-known \cite{FewsterRejzner2020} that van Hove Hamiltonians can be diagonalized by an algebraic Weyl transformation. That means, after replacing creation and annihilation operators like $ a(v) \mapsto a(v) + \langle v, s \rangle $ and $ a^\dagger(v) \mapsto a^\dagger(v) + \langle s, v \rangle $, as well as dropping an infinite self-energy constant, one arrives at a non-interacting Hamiltonian, which equals $ H_{0, y} + V $ restricted to a fiber.\\
The main novelty in this paper is that we are able to define $ W(s), s \notin L^2 $ and certain products thereof, creation operators $ A^\dagger(v), v \notin L^2 $ and even divergent integrals $ \fr \in \Ren_1 $, as linear operators on suitable domains. A similar ESS framework was recently successfully applied by the author to Bogoliubov transformations in place of $ W(s) $ \cite{BBS, BBS2}. The rigorous treatment of divergent integrals may become useful also in a purely algebraic approach, which does not refer to any Fock space extensions, but nevertheless produces divergent integrals when evaluating commutation relations.\\

Let us add some remarks about the dressing operator $ W(s) $: Gross transformations are usually of the form $ W(s) = e^{A^\dagger(s) - A(s)} $, where $ A^\dagger(s) = \sum_{j = 1}^M A_j^\dagger(s) $ describes a boson creation induced by all fermions $ j \in \{1, \ldots, M\} $. We will choose a slightly different dressing operator, which formally reads
\begin{equation}
	W(s) = W_M(s) \ldots W_1(s) \;, \qquad W_j(s) = e^{A_j^\dagger(s) - A_j(s)} \;.
\end{equation}
The reason is that, when applying the Baker--Campbell--Hausdorff formula (as in \eqref{eq:WjWjprime}), formal calculations with $ e^{A^\dagger(s) - A(s)} $ would produce exponential operators of the kind $ e^{V_{jj'}} $, with $ V_{jj'} $ being a potential interaction between fermions $ j $ and $ j' $. Those expressions $ e^{V_{jj'}} $ cannot be defined as operators on the ESS $ \FB_{\ex} $, so we need to avoid their occurrence.\\
In general, we will define the extended dressing operators $ W_j(s) $ such that their action for $ s \in \fh $ slightly differs from that of the usual Fock space dressing operators $ W_{\sF, j}(s): \sF \to \sF $. That is, we drop any terms of the kind $ e^{V_{jj'}} $ in our definition. This ad-hoc modification is justified by the fact that, in a sufficiently regular case, the $ V_{jj'} $-operators commute with $ W(s) $ (Lemma \ref{lem:Xcommute}). So even if we could define them as operators, they would just act as an independent factor that can be pulled to the left. So dropping all $ e^{V_{jj'}} $-terms can essentially be seen as a simplification of the bookkeeping.\\
Nevertheless, it would be interesting for future works to define $ e^{V_{jj'}} $ as an operator, mapping to some ESS $ \FB_{\ex}' $, which is constructed differently than $ \FB_{\ex} $. The current ESS $ \FB_{\ex} $ also does not allow for defining general products of annihilation operators $ A_{j_1}(v_1) \ldots A_{j_n}(v_n) $ with $ n \ge 2 $. It is a further interesting question, how such products can be defined in the future by an alternative construction of some $ \FB_{\ex}' $.\\
The ESS construction is currently at an early stage of development. We expect that the methods developed here may facilitate the renormalization of more involved formal Hamiltonians.\\

Although the spaces $ \FB, \FB_{\ex} $ do not have a topological structure, we may define a \textbf{renormalized scalar product} on $ W(s)[\cD_\sF] $ via:
\begin{equation}
	\langle W(s) \Psi, W(s) \Phi \rangle_{\ren} := \langle \Psi, \Phi \rangle \qquad \forall \Psi, \Phi \in \cD_\sF \;.
\end{equation}
The completion of $ W(s)[\cD_\sF] $ with respect to $ \langle \cdot, \cdot \rangle_{\ren} $ defines a Hilbert space $ \FR $, the \textbf{renormalized Fock space}. The map $ W(s) $ then uniquely extends to an isometric isomorphism between $ \sF $ and $ \FR $.\\

We remark that our result \eqref{eq:HtildeESS} is actually just the lowest-order approximation in a perturbation expansion within the weak-coupling regime, and completely decouples the fermions from the bosonic radiation field. The reason is that we both consider resting fermions (by excluding $ H_{0, x} $ from $ H $) and restrict to form factors $ v(\bk) $, that only depend on the momentum $ \bk $ of the emitted boson, and not on the momentum $ \bp $ of the fermion emitting it. It is physically expected and confirmed for the Nelson model with UV-cutoff \cite{Teufel2002, TenutaTeufel2008}, that $ W_\Lambda^* H_\Lambda W_\Lambda $ contains interactions between fermions and the radiation field. When changing to $ v(\bp,\bk) $, one may formally use a so-called ``Lie--Schwinger series' \cite{Froehlich1977}:
\begin{equation*}
	\tilde{H} = e^A B a^{-A} = \sum_{n = 0}^\infty \frac{\ad^n(A)B}{n!} \;,
\end{equation*}
where $ A := - A^\dagger(s) + A(s), \; B := H $ and with the $ n $-fold commutator
\begin{equation*}
	\ad^n(A)B := [A,[A, \ldots ,[A,B] \ldots]] \;.
\end{equation*}
However, establishing well-definedness of $ \ad^n(A)B $ for $ H + H_{0, x} $ with $ v $ depending on $ \bp $ is a rather involved task, so we postpone it to future investigations.\\

Another interesting objective for future works would be to introduce a mass renormalization term $ \delta m $. In constructive QFT (CQFT), where the Hamiltonian is of a form different from \eqref{eq:formalhamiltonian2}, mass renormalization terms can also be found, e.g., within the renormalization of $ \Phi^4 $-theory in $ 2+1 $ dimensions \cite{Glimm1968}, the Yukawa model \cite{GlimmJaffeY2I, GlimmJaffeY2II, SchraderY2} or the Lee model \cite[Ch. III]{Hepp}. Other CQFT renormalization procedures work without a mass renormalization, such as those for $ \Phi^4 $- and $ P(\Phi) $-theory in $ 1+1 $ dimensions \cite{GlimmJaffeI, GlimmJaffeII, GlimmJaffeIII, GlimmJaffeP2, GlimmJaffeIV}, for related models with exponential interactions \cite{HoeghKrohnE2}, the Thirring model \cite{Thirring1958} or the Federbush model \cite{Ruijsenaars1982}. For a more extensive review on CQFT literature, see \cite{Summers2016, Dedushenko2022}\\

Let us also mention that there has recently been some considerable interest in another cutoff-free renormalization technique via so-called ``interior--boundary conditions'' (IBC) \cite{I00, I00b, I02, I03, I04, I06, I07, I10, I11, I09, I05}. In fact, IBC renormalization has been an important source of inspiration for the ESS construction: As we explain in Section \ref{sec:furtherdress}, IBC renormalization provides an explicit formula for the domain of the renormalized Hamiltonian \eqref{eq:IBCdomain}, which may formally contain vectors outside $ L^2 $, if the interaction becomes too singular. In fact, IBC renormalization is limited to models, where the ``interacting vacuum is Fock'', that is, one may use the same representation of creation- and annihilation operators for both the free and the interacting Hamiltonian. A famous theorem of Haag \cite{Haag1955}, \cite[Sect.~II.1]{Haag1996} states that this condition can never be fulfilled within a relativistic QFT satisfying the Haag--Kastler axioms. When passing over to relativistic dispersion relations and form factors, IBC renormalization therefore reaches a certain limit, which the ESS construction is designed to overcome. Even in cases where IBC renormalization is successful, ESSs allow for presenting the ``cancellation of infinities'' in a rigorous and particularly convenient way, as shown in Section \ref{sec:furtherdress}.

The rest of this paper is structured as follows: After specifying some mathematical notation in Section \ref{sec:Ham}, we conduct the ESS construction in Section \ref{sec:ext}, finally resulting in spaces $ \FB, \FB_{\ex} $.\\
	We then establish $ H_0, A^\dagger, A $ and $ E_\infty $ as linear operators on $ \FB $ or $ \FB_{\ex} $ in Section \ref{sec:opext}.\\
	In the following Section \ref{sec:Fren}, we construct the dressing transformation $ W(s) $. Further, we define an \textbf{extended Weyl algebra} $ \WB $, which includes linear combinations of the operators $ W_j(s) = e^{A_j^\dagger(s) - A_j(s)} $. $ \WB $ allows for a multiplication with infinite wave function renormalization factors $ \fc \in \eRen $, and the Weyl relations
\begin{equation}
\begin{aligned}
	W_j(s)^{-1} &= W_j(-s) \;,\\
	W_j(s_1) W_j(s_2) &= e^{\Im\langle s_1,s_2 \rangle} W_j(s_1+s_2) \;,\\
\end{aligned}
\end{equation}
hold with $ e^{\Im\langle s_1,s_2 \rangle} \in \eRen $. However, note that $ \WB $ is generated by $ W_j(s) $ for a fixed $ j $, instead of $ W(s) $ and the dressing operator $ W(s) $ is not contained in $ \WB $.\\
	In Section \ref{sec:pullback} we compute $ \widetilde{H} $ such that $ W(s) \widetilde{H} = H W(s) $. This section contains the main result, Theorem \ref{thm:pullback}.\\
	The proof that under certain conditions $ \widetilde{H} $ is a Fock space operator and allows for self-adjoint extensions follows in Section \ref{sec:selfadjoint}. Thus, $ \widetilde{H} $ generates well-defined quantum dynamics, although they are not necessarily unique.\\
	Finally, in Section \ref{sec:furtherdress} we briefly review, using ESSs, how the cancellation of divergent quantities within IBC renormalization works. Essentially, IBC renormalization uses a non-unitary dressing transformation $ W_{\IBC} = (1 + H_0^{-1} A^\dagger)^{-1} $, which we prove to be a bijective map $ \FB \to \FB $ under relatively mild assumptions (Proposition \ref{prop:IBC}). This includes cases, where $ W_{\IBC} $ maps out of the Fock space. We also consider the non-unitary dressing transformation $ T = e^{- \bLambda/2} e^{- H_0^{-1} A^\dagger} $ with $ e^{- \bLambda/2} $ being an infinite wave function renormalization, which is a simplified version of CQFT dressings in \cite{Friedrichs1965, Glimm1968, Hepp}. Here, we define $ T: \FB \to \FB $.\\

\section{The Mathematical Model}
\label{sec:Ham}

\subsection{Formal Hamiltonian}
\label{subsec:definition}

In this paper, we consider a class of non-perturbative QFT models. These involve two species of particles:
\begin{itemize}
\item There is one species of spinless fermionic particles ($ x $).
\item These fermions interact by exchange of spinless bosons ($ y $).
\end{itemize}

By $ M $ and $ N $, we denote the number of $ x $- and $ y $-particles, respectively. In the particle--position representation, a single particle is described by a Hilbert space vector
\begin{equation}
	\varphi \in \fh = L^2(\RRR^d,\CCC) \;,
\end{equation}
with some coordinates denoted by a boldface symbol $ \bx \in \RRR^d $. So the system has $ d $ space dimensions and $ d+1 $ spacetime dimensions. The configuration of the entire system is given by an element $ (\bX,\bY) $ of the configuration space $ \cQ $,
\begin{equation}
	(\bX,\bY)=(\bx_1,\ldots ,\bx_M,\by_1,\ldots ,\by_N) \in \cQ_x \times \cQ_y =: \cQ \;,
\end{equation}
where the $ x $-, and $ y $-configuration spaces $ \cQ_x, \cQ_y $ and its sectors $ \cQ_x^{(M)}, \cQ_y^{(N)}, \cQ^{(M,N)} $ are defined as
\begin{equation}
	\bigsqcup_{M=0}^\infty \RRR^{Md} =: \bigsqcup_{M=0}^\infty \cQ_x^{(M)} =: \cQ_x \;, \quad
	\bigsqcup_{N=0}^\infty \RRR^{Nd} =: \bigsqcup_{N=0}^\infty \cQ_y^{(N)} =: \cQ_y \;, \quad
	\cQ^{(M,N)} := \cQ_x^{(M)} \times \cQ_y^{(N)} \;.
\label{eq:cQsplit}
\end{equation}
The state of the system at time $ t $ is described by a Fock space vector
\begin{equation}
	\Psi \in L^2(\cQ, \CCC) \;.
\end{equation}
Note that the term ``state'' is also used to describe a complex-valued map $ \omega $ on the bounded linear operators on $ \sF $ (possibly a mixed state). We will only refer to state vectors when using the term ``state ''.

For states, we assume the function $ \Psi \in L^2(\cQ, \CCC) $ to be anti-symmetric under coordinate exchange for fermions $ x $ and symmetric under coordinate exchange for bosons $ y $. The corresponding symmetrization and antisymmetrization operators $ S_- , S_+ $ are defined on $ L^2(\cQ_x, \CCC) $ and $ L^2(\cQ_y, \CCC) $ as
\begin{equation}
\begin{aligned}
	(S_- \Psi)(\bx_1, \ldots, \bx_M) &= \frac{1}{M!} \sum_{\sigma \in S_M} \sgn(\sigma) \Psi(\bx_{\sigma(1)}, \ldots, \bx_{\sigma(M)}) \;\\
	(S_+ \Psi)(\by_1, \ldots, \by_N) &= \frac{1}{N!} \sum_{\sigma \in S_N} \Psi(\bx_{\sigma(1)}, \ldots, \bx_{\sigma(N)}) \;,\\
\end{aligned}
\label{eq:S+S-}
\end{equation} 

where $ S_N $ is the permutation group of $ \{1, \ldots, N\} $. The \textbf{fermionic/bosonic/total Fock space} is then
\begin{equation}
	\sF_x := S_-[L^2(\cQ_x, \CCC)] \;, \qquad
	\sF_y := S_+[L^2(\cQ_y, \CCC)] \;,	\qquad
	\sF := \sF_x \otimes \sF_y \;.
\end{equation}
Each of these spaces decays into sectors according to \eqref{eq:cQsplit}
\begin{equation}
	\sF_x = \bigoplus_{M = 0}^\infty \sF_x^{(M)} \;, \qquad
	\sF_y = \bigoplus_{N = 0}^\infty \sF_y^{(N)} \;, \qquad
	\sF = \bigoplus_{M, N = 0}^\infty \sF^{(M, N)} \;.
\label{eq:sectordecay}
\end{equation}

It is convenient to describe the action of the formal Hamiltonian $ H $ in the particle-momentum representation: For any $ \Psi \in \sF $ with particle--position representation $ \Psi(\bX,\bY) $, we define the Fourier transform $ \Psi(\bP,\bK) $ with momentum configuration
\begin{equation}
	(\bP,\bK) = (\bp_1, \ldots, \bp_M,\bk_1, \ldots, \bk_N) \in \RRR^{Md+Nd} \;,
\end{equation}
via
\begin{equation}
	\Psi(\bP,\bK) := (2 \pi)^{-\frac{Md + Nd}{2}} \int_{\RRR^N} \int_{\RRR^M} \Psi(\bX,\bY) \; e^{-i \bP \cdot \bX - i \bK \cdot \bY} \; d^{Md}\bX d^{Nd}\bY \;.
\end{equation}
In the following, the set of variables plugged into $ \Psi(\cdot,\cdot) $, i.e., $ (\bX,\bY) $ or $ (\bP,\bK) $ will specify, which representation is meant.\\

As a formal Hamiltonian, we consider the following expression with zero fermion dispersion relation:
\begin{equation}
	H = H_{0, y} + A^{\dagger}(v) + A(v) - E_\infty \;.
\label{eq:hamfix}
\end{equation}
In Section \ref{sec:furtherdress}, we will also consider similar Hamiltonians of this kind that feature a nonzero fermion dispersion relation, i.e., where $ H_{0, y} $ is replaced by $ H_0 = H_{0, x} + H_{0, y} $. The formal definitions of the relevant expressions read as follows:

\begin{itemize}
\item The kinetic term $ H_0 $ is characterized by two dispersion relations, i.e., by real-valued functions $ \theta(\bp), \omega(\bk) \in C^\infty(\RRR^d \setminus \{0\}) $ for fermions and bosons, respectively. In Theorem \ref{thm:pullback} we assume $ \theta \equiv 0 $, but it can be nonzero in Proposition \ref{prop:operatorextension} and in Section \ref{sec:furtherdress}. We decompose
\begin{equation}
\begin{aligned}
	H_0 = &H_{0, x} + H_{0, y} \;,\\
	(H_{0, x} \Psi)(\bP,\bK) = &\sum_{j=1}^M \theta(\bp_j) \Psi(\bP, \bK) \;,\\
	(H_{0, y} \Psi)(\bP,\bK) = &\sum_{\ell = 1}^N \omega(\bk_{\ell}) \Psi(\bP, \bK) \;.
\end{aligned}
\label{eq:H0}
\end{equation}
A shorthand notation is to use the symbols $ d \Gamma_x(\cdot) $ and $ d \Gamma_y(\cdot) $ for second quantization:
\begin{equation}
	H_{0, x} = d \Gamma_x(\theta) \;, \qquad H_{0, y} = d \Gamma_y(\omega) \;,
\end{equation}
with multiplication operators
\begin{equation}
\begin{aligned}
	&\theta: \fh \supseteq \dom(\theta) \to \fh \;, \qquad &&\omega: \fh \supseteq \dom(\omega) \to \fh \;,\\
	&\phi(\bp) \mapsto \theta(\bp) \phi(\bp) \;, && \phi(\bk) \mapsto \omega(\bk) \phi(\bk) \;.\\
\end{aligned}
\label{eq:thetaomegapseudo}
\end{equation}

\item The creation part $ A^{\dagger}(v) $ makes each fermion create a boson. It is specified by a form factor $ v $. In this paper, we will restrict to $ v \in C^\infty(\RRR^d \setminus \{0\}) $ and put some further assumptions on the scalings of $ v $ for $ |\bk| \to 0 $ (IR-regime) and $ |\bk| \to \infty $ (UV-regime). These assumptions are described in Section \ref{subsec:scaling}.\\

We may write $ A^{\dagger}(v) $ as a sum over operators $ A^{\dagger}_j(v), j \in \{1, \ldots, M\} $, which only make fermion $ j $ create a boson. The formal definition in particle--momentum representation reads
\begin{equation}
\begin{aligned}
	(A^{\dagger}(v) \Psi)(\bP,\bK) &= \left( \sum_{j=1}^M A_j^{\dagger} \Psi \right)(\bP,\bK)\\
	&= \sum_{j=1}^M \frac{1}{\sqrt{N}} \sum_{\ell = 1}^N v(\bk_{\ell}) \Psi(\bP + e_j \bk_{\ell},\bK \setminus \bk_{\ell}) \;,
\end{aligned}
\label{eq:adaggerp}
\end{equation}
where
\begin{itemize}[$ \bullet $]
\item $\bK \setminus \bk_{\ell} = (\bk_1, \ldots, \bk_{\ell - 1},\bk_{\ell + 1}, \ldots, \bk_N) $ denotes $ \bK $ without $ \bk_{\ell} $
\item $ \bP + e_j \bk_{\ell} = (\bp_1, \ldots, \bp_{j-1},\bp_j+\bk_{\ell},\bp_{j+1}, \ldots, \bp_M) $ is the shifted fermion momentum.
\end{itemize} 
Here, $ e_j $ is meant to denote the $ j $-th unit vector and, by slight abuse of notation, $ e_j \bk_{\ell} = (0, \ldots, \bk_{\ell}, \ldots, 0) \in \RRR^{Md} $ is used to denote the assignment of an additional momentum $ \bk_{\ell} $ to fermion $ j $, before boson $ \ell $ is emitted.\\
 
The corresponding definition in particle--position representation uses the Fourier inverse $ \check{v} $ of the form factor $ v $ and reads\footnote{We use the convention that $ A^\dagger_j $ is defined on all $ (M) $-sectors, where $ (A^\dagger_j(v) \Psi)^{(M, N)} = 0 $ whenever $ j > M $.}
\begin{equation}
	(A^{\dagger}(v) \Psi)(\bX,\bY) = \left( \sum_{j=1}^M A_j^{\dagger}(v) \Psi \right)(\bX,\bY) = \sum_{j=1}^M \frac{1}{\sqrt{N}} \sum_{\ell = 1}^N \check{v}(\by_{\ell}-\bx_j) \Psi(\bX,\bY \setminus \by_{\ell}) \;.
\label{eq:adaggerx}
\end{equation}
Without the cutoff, the form factor is $ v \notin L^2 $ in many physically relevant models. Its Fourier inverse $ \check{v} $ is therefore $ \notin L^2 $, as well, if it is even defined.\\

Note that there is a more general class of physically interesting models, which we do not treat here: The form factor may depend on the fermion momentum $ \bp $, i.e., it may read $ v(\bp,\bk) $ with $ v: \RRR^{2d} \to \CCC $, such as in \cite{Eckmann1970, BachFroehlichSigal1998, CohenTannoudjiPnA}.

\item The annihilation part $ A(v) $ in particle--momentum representation reads
\begin{equation}
\begin{aligned}
	(A(v) \Psi)(\bP,\bK) &= \left(\sum_{j=1}^M A_j(v) \Psi \right)(\bP,\bK)\\
	&= \sum_{j=1}^M \sqrt{N+1} \int v(\tilde{\bk})^* \Psi(\bP - e_j \tilde{\bk}, \bK, \tilde{\bk}) \; d\tilde{\bk} \;.
\end{aligned}
\label{eq:ap}
\end{equation}
In particle--position representation,
\begin{equation}
\begin{aligned}
	(A(v) \Psi)(\bX,\bY) &= \left(\sum_{j=1}^M A_j(v) \Psi \right)(\bX,\bY)\\
	&= \sum_{j=1}^M \sqrt{N+1} \int \check{v}(\tilde{\by} - \bx_j)^* \Psi(\bX, \bY, \tilde{\by}) \; d\tilde{\by} \;.
\end{aligned}
\label{eq:ax}
\end{equation}
Note that for $ v \in \fh $, both $ A^{\dagger}(v) $ and $ A(v) $ can be defined as operators on a dense domain in $ \sF $, where $ A^\dagger(v) $ is the adjoint of $ A(v) $.\\

\item The self-energy $ E_\infty $ is a formal multiplication operator of the form $ E_\infty = d \Gamma_x (E_1) $ with $ E_1: \RRR^d \to \RRR $. In particle--momentum representation,
\begin{equation}
	(E_\infty \Psi)(\bP,\bK) = \sum_{j=1}^M E_1(\bp_j) \Psi(\bP,\bK) \;.
\label{eq:Ep}
\end{equation}
We will consider the expression
\begin{equation}
	E_1(\bp) = \int - \frac{v(\bk)^* v(\bk)}{\omega(\bk)} \; d\bk \;.
\label{eq:E_1}
\end{equation}
With cutoffs $ \sigma, \Lambda $ applied to $ v $, this integral would be finite in many physically interesting situations. However, without cutoffs it is typically divergent and hence the operator $ E_\infty $ becomes a formal expression. We will define it as a map $ \FB \to \FB_{\ex} $ in Proposition \ref{prop:operatorextension}.\\

\end{itemize}

\subsection{Scaling Degrees}
\label{subsec:scaling}

The convergence of integrals appearing in formal calculations depends on how $ \theta(\bp), \omega(\bk) $ and $ v(\bk) $ scale at $ |\bp|, |\bk| $ going to $ \infty $ or $ 0 $. We assume a polynomial scaling, of which we keep track using scaling degrees $ m $ (UV) and $ \beta $ (IR). The scaling degree $ m $ in the UV regime is commonly used for symbols $ s \in \cS^m $. As we choose $ \theta, \omega, v \in C^\infty (\RRR^d \setminus \{0\}) $, an additional pole might appear at $ |\bk| \to 0 $, which we assume to be of order $ \beta $. Both scalings are assumed to be exact, i.e., polynomial bounds exist from above and below:\\

\begin{definition}[scaling]
Consider $ s \in C^\infty (\RRR^d \setminus \{0\}) $. We say that $ s $ has a \textbf{polynomial scaling} if there are some scaling degrees $ m_s \le \beta_s \in \RRR $ and constants $ C_1, C_2 \in \RRR $ such that
\begin{equation}
	|s(\bk)| \le C_1 |\bk|^{\beta_s} + C_2 |\bk|^{m_s} \qquad \forall \bk \in \RRR^d \setminus \{0\} \;.
\label{eq:scalingcondition}
\end{equation}
The space of all \textbf{symbols with polynomial scaling} is denoted
\begin{equation}
	\Sdot_1 := \big\{ s \in C^\infty(\RRR^d \setminus \{0\}) \; \mid \; \eqref{eq:scalingcondition} \text{ holds} \big\} \;.
\label{eq:Sdot1}
\end{equation}
\label{def:scaling}
\end{definition}

\begin{definition}[exact scaling]
Consider a symbol $ s \in C^\infty (\RRR^d \setminus \{0\}) $. We say that the symbol has an \textbf{exact polynomial scaling} (see Figure \ref{fig:Sdot_1bounds}) if there are some $ 0 < \underline{\varepsilon} < \overline{\varepsilon} \in \RRR $ such that
\begin{itemize}
\item There are a UV scaling degree $ m_s \in \RRR $ and $ c_1, C_1 > 0 $ with
\begin{equation}
	c_1 |\bk|^{m_s} \le |s(\bk)| \le C_1 |\bk|^{m_s} \qquad \forall |\bk| > \overline{\varepsilon} \;.
\label{eq:uvcondition}
\end{equation}
\item There are an IR scaling degree $ \beta_s \in \RRR $ and $ c_2, C_2 > 0 $ with
\begin{equation}
	c_2 |\bk|^{\beta_s} \le |s(\bk)| \le C_2 |\bk|^{\beta_s} \qquad \forall 0 < |\bk| < \underline{\varepsilon} \;.
\label{eq:ircondition}
\end{equation}
\item There is a constant $ c_3 > 0 $ with
\begin{equation}
	c_3 \le |s(\bk)| \qquad \forall \underline{\varepsilon} \le |\bk| \le \overline{\varepsilon} \;.
\label{eq:constantcondition}
\end{equation}
\end{itemize}
The space of all \textbf{symbols with exact polynomial scaling} is denoted
\begin{equation}
	\Sdot_{1,>} := \big\{ s \in C^\infty(\RRR^d \setminus \{0\}) \; \mid \; \eqref{eq:uvcondition}, \eqref{eq:ircondition} \text{ and } \eqref{eq:constantcondition} \text{ hold} \big\} \subseteq \Sdot_1 \;.
\end{equation}
The ``$ > $'' refers to the fact that $ |s(\bk)| > 0 $, so $ \frac{1}{s} $ is defined everywhere except from $ \bk = 0 $ and also scales polynomially.
\label{def:exactscaling}
\end{definition}

\begin{figure}[hbt]
	\centering
	\scalebox{1.0}{\begin{tikzpicture}[declare function = {
	f(\x) = 0.2/(\x);
	g(\x) = 0.02*\x*\x;
	h(\x) = 0.2/(\x) + 0.02*\x*\x;}]

\draw[thick, ->] (0,0) -- (8,0) node[anchor = south west] {$ |\bk| $};
\draw[thick, ->] (0,-0.5) -- (0,2) node[blue!50!red, anchor = south west] {$ |v(\bk)| $};
\draw[thick, samples = 50,domain = 0.1:8, blue!50!red] plot( \x , { h(\x) }  );
\fill[samples = 50,domain = 0.1:8, blue!50!red, opacity = 0.1] plot( \x , { h(\x) }  ) -- (8,0) -- (0,0) -- (0,{h(0.1)});
\node at (-0.2,-0.2) {0};

\draw[dashed, samples = 50,domain = 0.1:1.5] plot( \x , { 0.8*f(\x) }  );
\draw[dashed, samples = 50,domain = 0.16:1.5] plot( \x , { 1.6*f(\x) }  );
\node at (1,1.5) {$ \propto |\bk|^\beta  $};

\draw[dashed] (1.5,2) -- ++(0,-2.4) node [anchor = north] {$ \underline{\varepsilon} $};

\draw[dashed] (1.5,0.1) -- node[anchor = north] {$ c_3 $} ++(3.5,0);

\draw[dashed] (5,2) -- ++(0,-2.4) node [anchor = north] {$ \overline{\varepsilon} $};
\draw[dashed, samples = 50,domain = 5:8] plot( \x , { 0.8*g(\x) }  );
\draw[dashed, samples = 50,domain = 5:8] plot( \x , { 1.2*g(\x) }  );
\node at (7,1.5) {$ \propto |\bk|^m  $};

\end{tikzpicture}}
	\caption{Scaling degrees in the special case of a radially symmetric, positive function $ v $ with $ |v(\bk)| > 0 $.}
	\label{fig:Sdot_1bounds}
\end{figure}

Obviously, $ C_c^\infty(\RRR^d) \subset \Sdot_1 $, so $ \Sdot_1 \cap \fh $ is dense in $ \fh $.\\

For obtaining a well-defined renormalized Hamiltonian $ \widetilde{H} $ in Lemma \ref{lem:denselydefined}, we will assume
\begin{equation}
	 \theta, \omega, v \in \Sdot_{1,>} \;,
\label{eq:scalingfunctions}
\end{equation}
(see Figure \ref{fig:Sdot_1}) and denote their scaling degrees with $ m_\theta, m_\omega, m_v $ and $ \beta_\theta, \beta_\omega, \beta_v $, respectively. Further, we will assume in Lemma \ref{lem:denselydefined}:
\begin{equation}
	m_\theta, m_\omega, \beta_\theta, \beta_\omega \ge 0 \;.
\label{eq:scalingdispersion}
\end{equation}

\begin{figure}[hbt]
	\centering
	\scalebox{1.0}{\begin{tikzpicture}[declare function = {
	f(\x) = 0.6*sqrt(\x*\x+1)) ;
	g(\x) = 0.1*1/(\x*\x) + 0.3*\x*\x;
	h(\x) = 0.23*1/(\x);}]

\draw[thick, ->] (-1,0) -- (2,0) node[anchor = south west] {$ \bk $};
\draw[thick, ->] (0,-0.5) -- (0,2) node[blue!50!red, anchor = south west] {$ |v(\bk)| $};
\draw[thick, samples = 50,domain = -1:-0.12, blue!50!red] plot( \x , {- h(\x) }  );
\fill[samples = 50,domain = -1:-0.12, blue!50!red, opacity = 0.1] plot( \x , {- h(\x) }  ) -- ++ (0.12,0) -- (0,0) -- (-1,0);
\draw[thick, samples = 50,domain = 0.12:2, blue!50!red] plot( \x , { h(\x) }  );
\fill[samples = 50,domain = 0.12:2, blue!50!red, opacity = 0.1] plot( \x , { h(\x) }  ) -- (2,0) -- (0,0) -- (0,{h(0.12)});
\node at (-0.2,-0.2) {0};

\draw[thick, ->] (3,0) -- (6,0) node[anchor = south west] {$ \bk $};
\draw[thick, ->] (4,-0.5) -- (4,2) node[blue!50!red, anchor = south west] {$ |v(\bk)| $};
\draw[thick, samples = 50,domain = 3:3.77, blue!50!red] plot( \x , { g(\x-4) }  );
\fill[samples = 50,domain = 3:3.77, blue!50!red, opacity = 0.1] plot( \x , { g(\x-4) }  ) -- ++(0.23,0) -- (4,0) -- (3,0);
\draw[thick, samples = 50,domain = 4.23:6, blue!50!red] plot( \x , { g(\x-4) }  ) ;
\fill[samples = 50,domain = 4.23:6, blue!50!red, opacity = 0.1] plot( \x , { g(\x-4) }  ) -- (6,0) -- (4,0) -- (4,{g(0.23)});
\node at (3.8,-0.2) {0};

\draw[thick, ->] (7,0) -- (10,0) node[anchor = south west] {$ \bk $};
\draw[thick, ->] (8,-0.5) -- (8,2) node[red, anchor = south west] {$ |\omega(\bk)| $};
\draw[thick, samples = 50,domain = 7:10, red] plot( \x , { f(\x-8) }  ) ;
\fill[samples = 50,domain = 7:10, red, opacity = 0.1] plot( \x , { f(\x-8) }  ) -- (10,0) -- (7,0) -- cycle;
\node at (7.8,-0.2) {0};

\end{tikzpicture}}
	\caption{Examples for absolute values of functions $ v, \omega \in \Sdot_{1,>} $. The functions $ v, \omega $ are complex-valued. Note that $ \beta_\omega, m_\omega \ge 0 $.}
	\label{fig:Sdot_1}
\end{figure}

This implies, that for $ \theta $ and $ \omega $, there is no pole at the origin. QFT models often use dispersion relations based on symbols with $ m_\theta, m_\omega, \beta_\theta, \beta_\omega \in \{0,1,2\} $, which all satisfy this condition.\\

In order to obtain symmetric operators, we will also need to impose a symmetry condition in Lemma \ref{lem:denselydefined}:

\begin{equation}
	\theta(\bp) = \theta(-\bp) \;, \qquad \omega(\bk) = \omega(-\bk) \;, \qquad v(\bk) = v(-\bk) \;.
\label{eq:symmetry}
\end{equation}

\section{Construction of the ESS}	
\label{sec:ext}

\subsection{Wave Function Renormalization Factors}	
\label{subsec:ren}

The motivation behind the introduction of an ESS is to rigorously define the formal dressing operator
\begin{equation}
	W(s) = W_M(s) \ldots W_1(s) \;, \qquad \text{with } W_j(s) = e^{A_j^\dagger(s) - A_j(s)}, \; s = -\frac{v}{\omega} \;,
\label{eq:W}
\end{equation}
in the case where $ s \in \Sdot_1 \subseteq C^\infty(\RRR^d \setminus \{0\}) $, but $ s \notin \fh = L^2(\RRR^d,\CCC) $. Here, the number $ M $ corresponds to the fermion sector, so the second term in \eqref{eq:W} is a rather symbolic expression requiring us to adapt $ M $ in dependence of the fermion sector (similar to \eqref{eq:H0}). Formally, if we apply $ W(s) $ to a state vector $ \Psi_1 $ where the boson field is in the vacuum $ \Omega_y $ and $ \Psi_x \in \sF_x^{(1)} $ is a one-fermion state,
\begin{equation}
	 \Psi_1 = \Psi_x \otimes \Omega_y \;,
\end{equation}
then we obtain the following expression in momentum space:
\begin{equation}
	(W(s) \Psi_1)(\bp,\bK) = e^{-\frac{\Vert s \Vert^2}{2}}\frac{1}{\sqrt{N!}} \left( \prod_{\ell = 1}^N s(\bk_{\ell}) \right) \Psi_x\left(\bp + \sum_{\ell = 1}^N \bk_{\ell} \right) =: e^{-\frac{\Vert s \Vert^2}{2}} \Psi_0(\bp,\bK) \;.
\label{eq:Ws1}
\end{equation}
For $ s \in \fh $, the expression $ W(s) \Psi_1 $ is a Fock space vector with norm $ \Vert W(s) \Psi_1 \Vert = \Vert \Psi_1 \Vert $. However, if we set $ s \notin \fh $, two problems arise:
\begin{itemize}
\item $ \Vert s \Vert^2 = \langle s, s \rangle $ is formally a divergent integral. So the wave function renormalization factor $ e^{-\frac{\Vert s \Vert^2}{2}} $ is not well-defined.
\item $ \left( \prod_{\ell = 1}^N s(\bk_{\ell}) \right) $ (on every sector $ N \ge 1 $) is a non-square integrable function, so even without the infinite renormalization factor $ e^{-\frac{\Vert s \Vert^2}{2}} $, \eqref{eq:Ws1} does not describe an element of $ \sF $.
\end{itemize}

The second problem is tackled by defining a space $ \Sdot_\sF $ containing non-square integrable functions, including the above product. As $ s \in \Sdot_1 $, the product above is a smooth function apart from the zero boson momentum configuration set
\begin{equation}
	\exists(\bk = 0) := \big\{ q = (\bP,\bK) \in \cQ \; \big\vert \; \exists \ell \in \{1, \ldots, N\}: \; \bk_{\ell} = 0 \big\} \;.
\label{eq:existsk0}
\end{equation}

\begin{figure}[hbt]
	\centering
	\scalebox{0.7}{\begin{tikzpicture}

\filldraw[blue, fill opacity = 0.1] (-0.6,-3.4) -- (0.6,-0.6) -- (0.6,3) -- (-0.6,1) -- cycle;
\filldraw[blue, fill opacity = 0.1] (-3,-2) rectangle (3,2);
\draw[thick, ->] (-3,0) -- (3,0) node[anchor = south west] {$ k_1 $};
\draw[thick, ->] (-0.6,-1.2) -- (0.6,1.2) node[anchor = south west] {$ k_2 $};
\draw[thick, ->] (0,-2) -- (0,2) node[anchor = south west] {$ p_1 $};
\node at (-0.2,-0.2) {0};
\draw[blue, thick] (-5,1.5) node[anchor = east] {\Large $ \exists(k=0) $} -- (-3,1);
\draw[thick,->] (4,2.5) node[anchor = west] {\Large $\dot \cQ $} -- (2.5,2.5);

\end{tikzpicture}}
	\caption{The set $ \exists(\bk=0) $ on the $ (M,N) = (1,2) $-sector in configuration space $ \cQ $ for $ d=1 $.}
	\label{fig:existsk0}
\end{figure}

The restriction of $ \exists(\bk = 0) $ to any sector is a union of hyperplanes (see Figure \ref{fig:existsk0}), which have Lebesgue measure 0 and hence, also $ \exists(\bk = 0) \subseteq \cQ $ is null. Excluding this set from configuration space, we obtain
\begin{equation}
	\Qdot := \cQ \setminus \exists(\bk = 0) \;,
\end{equation}
and on $ q \in \Qdot $, the function $ \Psi_0 $ in \eqref{eq:Ws1} is smooth. For $ |\bk| \to 0 $ or $ |\bp|, |\bk| \to \infty $ we require the following scaling conditions to hold:
\begin{equation}
	\lim_{|\bk_{\ell}| \to 0} \frac{\Psi_0(\bP,\bK)}{|\bk_{\ell}|^\beta} = \lim_{|\bk_{\ell}| \to \infty} \frac{\Psi_0(\bP,\bK)}{|\bk_{\ell}|^m} = \lim_{|\bp_j| \to \infty} \frac{\Psi_0(\bP,\bK)}{|\bp_j|^m} = 0 \;,
\label{eq:psiscaling}
\end{equation}
for all fixed $ q = (\bP, \bK) \in \Qdot $ and some $ \beta, m \in \RRR $. We hence interpret $ \Psi_0 $ as an element of the following space:
\begin{equation}
	\Sdot_\sF := \big\{ \Psi \in C^\infty(\Qdot) \; \big\vert \; \eqref{eq:psiscaling} \text{ holds} \big\} \;.
\label{eq:SdotsF}
\end{equation}
Note that $ \Sdot_\sF \cap \sF $ contains $ C_c^\infty(\cQ) $, which is a dense set in $ \sF $. So $ \Sdot_\sF \cap \sF $ is dense in $ \sF $. The Fourier transform of some $ \Psi_0 \in \Sdot_\sF $ can be taken, if it is an element of the space
\begin{equation}
	\Sdot_{\sF,\loc} := \Sdot_\sF \cap L^1_{\loc}(\cQ) \;.
\end{equation}

In order to address the first problem, it is necessary to assign mathematical meaning to the expression $ \langle s, s \rangle $ (called \textbf{renormalization factor}). It may be convenient to take a Fourier transform of $ s \in \Sdot_1 $, for instance, in order to define a particle--position representation. This can be done whenever $ s \in L^1_{\loc}(\RRR^d) $, since then $ s \in \cS'(\RRR^d) $, i.e., it is a tempered distribution. We may hence distinguish two interesting function spaces for $ s $:
\begin{itemize}
\item The generic case is given by $ s \in \Sdot_1 $
\item The special case is given by $ s \in \Sdot_1 \cap L^1_{\loc}(\RRR^d) =: \Sdot_{1,\loc} $ , which allows for taking the Fourier transform.
\end{itemize}
In the following, we will present a construction based on the generic case, since the special one can be treated by the same means. Clearly, $ \fh \cap \Sdot_1 \subseteq \Sdot_{1,\loc} $, so even $ \Sdot_{1,\loc} $ contains a dense subspace of $ \fh $. If $ s \in \fh $, then
\begin{equation}
	\langle s, s \rangle = \int |s(\bk)|^2 \; d \bk \in \CCC \;.
\end{equation}
Otherwise, the renormalization factor $ \langle s , s \rangle $ is a symbolic expression, corresponding to a divergent integral. We interpret it as an element of the following vector space:

\begin{definition}
Consider the free $ \CCC $-vector space $ F(\Sdot_1 \times \Sdot_1) $. By a free vector space, we mean the set of all finite $ \CCC $-linear combinations of pairings, denoted $ \fr = \sum_{m=1}^M c_m \langle s_m, t_m \rangle $, with $ s_m, t_m \in \Sdot_1, \; c \in \CCC $ and the sum $ \sum_{m=1}^M $ being commutative. The \textbf{space of renormalization integrals} is defined as the quotient space
\begin{equation}
	\Ren_{01} := F \left( \Sdot_1 \times \Sdot_1 \right)/ _{\sim_{\Ren_{01}}} \;,
\end{equation}
where the equivalence relation $ \sim_{\Ren_{01}} $ of formal equality is given by
\begin{equation}
\begin{aligned}
	&\sum_{m=1}^M c_m \langle s_m,t_m \rangle \sim_{\Ren_{01}} \sum_{m=1}^{\tilde{M}} \tilde c_m \langle \tilde s_m, \tilde t_m \rangle\\
	:\Leftrightarrow \quad &\sum_{m=1}^M c_m \overline{ s_m (\bk)} t_m(\bk) = \sum_{m=1}^{\tilde{M}} \tilde c_m \overline{ \tilde s_m(\bk)} \tilde t_m(\bk) \quad \forall \bk \in \RRR^d \setminus \{0\} \;.
\end{aligned}
\end{equation}
\label{def:ren01}
\end{definition}

There is a natural one-to-one identification of renormalization integrals with functions $ \Ren_{01} \cong \Sdot_1 $: It is easy to see that the following map is an embedding (by definition of $ \sim_{\Ren_{01}} $):
\begin{equation}
\begin{aligned}
	\iota_1: \Ren_{01} &\to \Sdot_1 \;,\\
	 \sum_{m=1}^M c_m \langle s_m,t_m \rangle = \fr &\mapsto r(\bk) = \sum_{m=1}^M c_m \overline{ s_m (\bk)} t_m(\bk) \;.\\
\end{aligned}
\end{equation}
Conversely, for a given $ r \in \Sdot_1 $, the element $ \langle r, \chi_{\RRR^d \setminus \{0\}} \rangle \in \Ren_1 $ (with $ \chi $ being the indicator function) is identified with $ r $, so $ \iota_1 $ is indeed an isomorphism.\\

However, it can be misleading to think of $ \fr \in \Ren_1 $ as a function $ r $. Renormalization integrals are formally $ \fr = \int r(\bk) \; d\bk $, so they can be added but not directly multiplied. Hence the notation $ \langle \cdot, \cdot \rangle $, resembling an $ L^2 $ scalar product. By contrast, functions $ r_1,r_2 \in \Sdot_1 $ can be multiplied which leads again to a function in $ \Sdot_1 $.\\

In case $ r \in L^1(\RRR^d) $, we can identify $ \fr \in \Ren_{01} $ with the $ \CCC $-number
\begin{equation}
	\fr = \int_{\RRR^d} \sum_{m=1}^M c_m\overline{s_m(\bk)} t_m(\bk) \; d\bk = \int_{\RRR^d} r(\bk) \; d\bk \in \CCC \;.
\end{equation}

Now, several $ \fr \in \Ren_{01} $ get identified with the same $ \CCC $-number, e.g., all $ \fr $ corresponding to a function $ r $ with $ \int r(\bk) \; d \bk = 0 $ are identified with 0. We remove this ambiguity by modding out another equivalence relation:

\begin{definition}
	The \text{renormalization factor space} $ \Ren_1 $ is defined as
\begin{equation}
	\Ren_1 := \Ren_{01}/_{\sim_{\Ren_1}} \;,
\end{equation}
where for $ \fr_1, \fr_2 \in \Ren_{01} $ we define
\begin{equation}
	 \fr_1 \sim_{\Ren_1} \fr_2 \; :\Leftrightarrow \; \iota_1(\fr_1) - \iota_1(\fr_2) \in L^1 \quad \text{and} \quad \int (\iota_1(\fr_1) - \iota_1(\fr_2))(\bk) \; d \bk = 0 \;.
\end{equation}
\label{def:ren1}
\end{definition}

Elements of $ \Ren_1 $ will be denoted equally to a representative $ \fr $. Note that we can identify $ \Ren_1 $ with the quotient space
\begin{equation}
	\Ren_1 \cong (\CCC \oplus \Ren_{01})/D \quad \text{with} \quad D:= \big\{(-\smallint \iota_1(\fr)(\bk) \; d\bk, \fr) \; \big\vert \; \iota_1(\fr) \in L^1(\RRR^d) \big\} \;,
\label{eq:quotientspace}
\end{equation}
where the isomorphism is given by identification of $ (c,\fr) $ with that class $ [\fr'] \in \Ren_1 $ where $ \smallint(\iota_1(\fr')(\bk) - \iota_1(\fr)(\bk)) \; d\bk = c $.\\

We will encounter the special case where $ r(\bk) $ only takes values in $ [0,\infty) $.
\begin{definition}
The \textbf{positive renormalization factor cone} $ \Ren_{1+} $ is the set of all $ [\fr] \in \Ren_1 $, where at least one representative $ \fr \in \Ren_{01} $ is identified with a positive-valued function $ \iota_1(\fr) :\RRR^d \setminus \{0\} \to [0,\infty) $.
\label{eq:ren1+}
\end{definition}
Scalar multiplication by $ c \in [0,\infty) $ is well-defined, making this indeed a cone.\\

\vspace{0.5cm}
It is also convenient to make sense of products and polynomials of factors $ \fr \in \Ren_1 $. We hence define the following vector spaces.

\begin{definition}
Consider the free $ \CCC $-vector space of all finite linear combinations of products of up to $ P \in \NNN_0 $ renormalization factors (i.e., formal polynomials of degree $ P $)
\begin{equation}
	\Pol_P := \left\{ \fR = \sum_{m=1}^M c_m \fr_{m,1} \cdot \ldots \cdot \fr_{m,p_m} \; \middle\vert \; 0 \le p_m \le P, \; c_m \in \CCC, \; \fr_{m,p} \in \Ren_1 \right\} \;,
\end{equation}
with the sum $ \sum_{m=1}^M $ and products being commutative. Then, the space of \textbf{renormalization factor polynomials} of order $ \le P $ is the quotient space
\begin{equation}
	\Ren_P = \Pol_P/_{\sim_{\Ren_P}} \;,
\end{equation}
with the equivalence relation $ \sim_{\Ren_P} $ of formal equality generated by
\begin{equation}
\begin{aligned}
	\fr_1 \fr_2 \ldots \fr_{p_1} + \fR
	\; &\sim_{\Ren_P} \; c \fr_2 \ldots \fr_{p_1} + \fR \qquad \text{if } \fr_1 = c \in \CCC \;,\\
	(c_1 c_2) \fr_1 \ldots\fr_{p_1} + \fR
	\; &\sim_{\Ren_P} \; c_1 (c_2 \fr_1) \ldots\fr_{p_1} + \fR \;,\\
\end{aligned}
\end{equation}
with $ P_1 \le P $, $ \fR \in \Pol_P $, $ c_1,c_2 \in \CCC $ and $ \fr_m \in \Ren_1 $ \;.
\label{def:renP}
\end{definition}

The bound $ P $ on the polynomial order is removed by taking the union over all orders.
\begin{definition}
The space of \textbf{renormalization factor polynomials} is given by
\begin{equation}
	 \Ren = \bigcup_{P=1}^\infty \Ren_P \;.
\end{equation}
\label{def:ren}
\end{definition}

The spaces defined in this section follow the hierarchy
\begin{equation}
\Ren_{1+} \subseteq \Ren_1 \subseteq \Ren_2 \subseteq \ldots \subseteq \Ren_P \subseteq \ldots \subseteq \Ren \;.
\end{equation}

\subsection{Exponentials of Renormalization Factors and the Field $ \eRen $}	
\label{subsec:classes}

The state vector $ W(s) \Psi_1 $ in \eqref{eq:Ws1} contains an exponential of a renormalization factor, i.e., an expression $ \fc = e^{\fr} $ with $ \fr \in \Ren_1 $. More generally, we would like to give a meaning to sums of exponentials with different (perhaps infinite) renormalization factors $ \fr_1, \fr_2 \in \Ren_1 $. Formally, we would like to identify
\begin{equation}
	e^{\fr_1} + e^{\fr_2} = e^{\fr_1}(1 + e^{\fr_2-\fr_1}) \;.\\
\end{equation}
The bracket is a $ \CCC $-number, whenever $ \fr_1-\fr_2 \in \CCC $, which defines an equivalence relation
\begin{equation}
	\fr_1 \sim_1 \fr_2 \quad :\Leftrightarrow \quad \fr_1-\fr_2 \in \CCC \;.
\end{equation}
\begin{definition}
The space of \textbf{renormalization factor classes} is then given by the quotient space
\begin{equation}
	\Clas_1 = \Ren_1/_{\sim_1} \;.
\end{equation}
\label{def:clas1}
\end{definition}
Whenever two elements $ \fr_1, \fr_2 $ are of two different classes, we have $ \fc = e^{\fr_1-\fr_2} \notin \CCC $ and may think of $ \fc $ as an ``infinite constant''.\\

We may also interpret $ \Clas_1 $ as a subspace of $ \Ren_1 $: The elements in the class of zero, $ b \in \Ren_1 : b \sim_1 0 $, form a subspace $ V \subset \Ren_1 $, $ V \cong \CCC $. Now we can find a basis $ B $ (containing one element) of $ V $ and, by the axiom of choice \cite{Jech}, extend it to a basis $ B \cup B' $ of $ \Ren_1 $. By defining $ W := \span(B') $, we obtain a decomposition
\begin{equation}
	\Ren_1 = V \oplus W \qquad \Rightarrow \qquad W \cong \Ren_1 / V \;.
\end{equation} 
Now, the following map is a bijection from $ W $ to $ \Clas_1 $:
\begin{equation}
	W \ni w \mapsto [w]_{\sim_{\Ren_1}} \in \Clas_1 \;.
\end{equation}
Note that this bijection is not unique, since it depends on the choice of $ B' $, i.e., of representatives within each class.\\

We now consider elements of the group algebra $ \CCC[\Ren_1] \ni c_1 e^{\fr_1} + \ldots + c_M e^{\fr_M} $ of the additive group of the vector space $ \Ren_1 $. The formal exponentials make an addition in the group $ e^{\fr_1 + \fr_2} = e^{\fr_1} e^{\fr_2} $ appear as a multiplication. Here, we would like to consider two summands as equal, if ``parts of the complex number can be pulled into the exponent'', i.e., $ c e^{c'+\fr} = (c e^{c'}) e^\fr $. This is done by defining an ideal $ \cI \subset \CCC[\Ren_1] $ generated by all elements of the form
\begin{equation}
	e^{c + \fr} - e^c e^{\fr} \qquad \text{with} \qquad c \in \CCC, \; \fr \in \Ren_1 \;.
\label{eq:ideal}
\end{equation}
Note that this ideal gives rise to an equivalence relation
\begin{equation}
	u \sim_\cI v \qquad :\Leftrightarrow (u-v) \in \cI \;.
\end{equation}

\begin{proposition}
\begin{equation}
	\CCC[\Ren_1]/ \cI \cong \CCC[W] \;.
\end{equation}
\label{prop:cong}
\end{proposition}
\begin{proof}
We define a map $ \pi: \Ren_1 \to W $, which assigns to $ \fr \in \Ren_1 $ the vector $ w $ within the decomposition $ \Ren_1 = V \oplus W $ above. So $ \fr - \pi(\fr) \in \CCC $. Now within, $ c_1 e^{\fr_1} + \ldots + c_M e^{\fr_M} \in \CCC[\Ren_1] $, we can re-write each summand as
\begin{equation}
	c_m e^{\fr_m} \sim_{\cI} (c_m e^{\fr_m - \pi(\fr_m)})e^{\pi(\fr_m)} \;.
\end{equation}

This gives rise to a re-writing map
\begin{equation}
\begin{aligned}
	\Pi: \CCC[\Ren_1] &\to \CCC[\Ren_1] \;,\\
	c_1 e^{\fr_1} + \ldots + c_M e^{\fr_M} &\mapsto (c_1 e^{\fr_1 - \pi(\fr_1)})e^{\pi(\fr_1)} + \ldots + (c_M e^{\fr_M - \pi(\fr_M)})e^{\pi(\fr_M)} \;.\\
\end{aligned}
\end{equation}
Note that by the computation rules for group algebras, whenever $ e^{\pi(\fr_m)} = e^{\pi(\fr_{m'})} $ for two summands, they can be combined into one.\\

The map $ \Pi $ is an algebra homomorphism: it is linear and respects the multiplication
\begin{equation}
	\Pi(c_1 e^{\fr_1} + \ldots + c_M e^{\fr_M}) \Pi(c'_1 e^{\fr'_1} + \ldots + c'_M e^{\fr'_M}) = \Pi(c_1 e^{\fr_1} + \ldots + c_M e^{\fr_M} + c'_1 e^{\fr'_1} + \ldots + c'_M e^{\fr'_M}) \;.
\end{equation}
The latter is a consequence of $ \pi(\fr_1 + \fr_2) = \pi(\fr_1) + \pi(\fr_2) $ which does only hold true because we have chosen $ W $ as a vector space. Further, $ \Pi $ does not change the equivalence class with respect to $ \sim_{\cI} $. So $ \Pi(I) = 0 $ implies $ I \in \cI $. Conversely, if $ I \in \cI $ then $ I = I_1+\ldots +I_M $ with each $ I_i = A_i B_i C_i $ with $ A_i,C_i \in \CCC[\Ren_1] $ and $ B_i = e^{c_i+\mathfrak{r}_i} - e^{c_i}e^{\mathfrak{r}_i} $, so $ \Pi(B_i) = 0 $ and $ \Pi(I_i) = 0 $ and hence $ \Pi(I) = 0 $. So the kernel of $ \Pi $ is exactly $ \Ker(\Pi) = \cI $.\\
Moreover, the image is $ \Pi[\CCC[\Ren_1]] = \CCC[W] $, since only elements of $ \CCC[W] $ appear in it and any element of $ \CCC[W] \subset \CCC[\Ren_1] $ is mapped to itself.\\
By the isomorphism theorems, we now have that
\begin{equation}
	\CCC[\Ren_1]/\Ker(\Pi) \cong \Pi[\CCC[\Ren_1]] \qquad :\Leftrightarrow \qquad \CCC[\Ren_1]/ \cI \cong \CCC[W] \;,
\end{equation}
as claimed.
\end{proof}

\begin{proposition}
$ \CCC[\Ren_1]/ \cI $ has no proper zero divisors.
\label{prop:zerodiv}
\end{proposition}
($ a, b \in \CCC[\Ren_1]/ \cI, a,b \neq 0 $ are called proper zero divisors if and only if $ a b = 0 $.)
\begin{proof}
By Proposition \ref{prop:cong}, it suffices to show that $ \CCC[W] $ has no proper zero divisors.\\

Now, the additive group $ W $ is an Abelian group that is torsion-free, i.e., for $ g \in W $ and $ n \in \NNN $:
\begin{equation}
	\underbrace{g+g+ \ldots+g}_{n \text{ times}} = 0 \quad \Rightarrow \quad g = 0 \;.
\end{equation}
Now, by \cite[Lemma 26.6]{Passman}, the group $ G = W $ is ordered, so by \cite[Lemma 26.4]{Passman} it is ``t.u.p.'', so by \cite[Lemma 26.2]{Passman} it is ``u.p.'' and $ K[G] = \CCC[W] $ has no proper zero divisors, i.e., it is an entire ring.\\
\end{proof}

Following \cite[II.~§3]{Lang}, it is a theorem that for every ring without proper zero divisors, the quotients form a field. So by Proposition \ref{prop:zerodiv} the following field extension of $ \CCC $ is well-defined:

\begin{definition}[and Corollary]
The \textbf{field of (exponential) wave function renormalizations} is given by all fractions of linear combinations
\begin{equation}
	\eRen := \left\{ \fc = \frac{a_1}{a_2} \; \middle\vert \; a_1,a_2 \in \CCC[\Ren_1]/ \cI \right\} \;.
\end{equation}
\label{def:eren}
\end{definition}

By using representatives of $ \CCC[\Ren_1]/ \cI $, we can write any $ \fc \in \eRen $ as
\begin{equation}
	\fc = \frac{\sum_{m = 1}^M c_m e^{\fr_m}}{\sum_{m' = 1}^{M'} c_{m'} e^{\fr_{m'}}} \quad \text{with} \quad c_m, c_{m'} \in \CCC \;, \quad \fr_m, \fr_{m'} \in \Ren_1 \;.
\end{equation}

In particular, we can view a wave function renormalization $ \fc \in \eRen $ as an ``extended complex number''.\\

\subsection{First ESS}	
\label{subsec:ext1}

With the above definitions, we are able to give meaning to expressions like
\begin{equation}
	\Psi = e^{-\fr} \Psi_0 \qquad \text{or even} \qquad \Psi = \fc \Psi_0 \;,
\label{eq:ext1}
\end{equation}
with $ \fr \in \Ren_1, \fc \in \eRen $ and $ \Psi_0 \in \Sdot_\sF $, as they appear in \eqref{eq:Ws1}. We would like to take linear combinations of them and even handle expressions like $ \fc \fR \Psi_0 $ with $ \fR \in \Ren $. This is done by defining $ \eRen $-vector spaces including such expression, either without $ \fR $ (this will be the first ESS $ \FB $) or with $ \fR $ (this will be the second ESS $ \FB_{\ex} $).\\

\begin{definition}
Consider the free $ \eRen $-vector space of all finite (commutative) sums of the form
\begin{equation}
	\FB_0 := \left\{ \Psi = \sum_{m=1}^M \fc_m \Psi_m \; \middle\vert \; \Psi_m \in \Sdot_\sF, \; \fc_m \in \eRen \right\} \;.
\label{eq:PsiFB}
\end{equation}
Then, the \textbf{first extended state space (ESS)} is the $ \eRen $-quotient space
\begin{equation}
	\FB = \FB_0 /_{\sim_{\F}} \;,
\end{equation}
with equivalence relation $ \sim_{\F} $ generated by
\begin{equation}
	(c \fc) \Psi_a + \Psi
	\;\sim_{\F} \; \fc (c \Psi_a) + \Psi \qquad \text{if } c \in \CCC \;,
\end{equation}
for any $ \fc \in \eRen, \; \Psi \in \FB_0 $ and $ \Psi_a \in \Sdot_\sF $.\\

\label{def:FB}
\end{definition}

The dense subspace $ \Sdot_\sF \cap \sF $ of the Fock space can be naturally embedded into $ \FB $: Every $ \Psi_{\sF} \in \Sdot_\sF \cap \sF $, can be identified with
\begin{equation}
	\Psi_{\sF} = e^0 \Psi_{\sF} \in \FB \;.
\end{equation}
Elements of $ \FB $ do not necessarily satisfy any symmetry conditions. However, we can decompose $ \FB $ into sectors, just as $ \sF $ in \eqref{eq:sectordecay}: Since $ \Sdot_\sF $ consists of functions in $ C^\infty(\cQ) $, we may introduce the projections
\begin{equation}
	P^{(M,N)}: \Sdot_\sF \to \Sdot_\sF \;, \qquad (P^{(M,N)} \Psi)(q) = \begin{cases} \Psi(q) \quad &\text{if } q \in \cQ^{(M,N)}\\ 0 \quad &\text{else} \end{cases} \;.
\end{equation}
By $ \eRen $-linearity, $ P^{(M,N)} $ extends to $ \FB $, which allows to define
\begin{equation}
	\FB^{(M,N)} := P^{(M,N)} \FB \;, \qquad \text{so} \qquad
	\FB = \bigoplus_{M, N \in \NNN_0} \FB^{(M,N)} \;.
\label{eq:FBsectors}
\end{equation}

The coherent state in \eqref{eq:Ws1} can now be seen as an ESS element:
\begin{equation*}
	(W(s) \Psi_1)(\bp,\bK) = \underbrace{e^{-\frac{\Vert s \Vert^2}{2}}}_{= e^\fr} \underbrace{\frac{1}{\sqrt{N!}} \left( \prod_{\ell = 1}^N s(\bk_{\ell}) \right) \Psi_x\left(\bp + \sum_{\ell} \bk_{\ell} \right)}_{ \in \Sdot_\sF } \;.
\end{equation*}

For $ s \in \Sdot_{1,\loc} $ and $ \Psi_x \in C_c^\infty $, the second factor is even in $ \Sdot_{\sF,\loc} $.\\

\subsection{Second ESS}	
\label{subsec:ext2}

Later in this work, we will encounter expressions like
\begin{equation}
	\Psi(q) = e^{-\fr} \fR(q) \Psi_0(q) \;,
\label{eq:ext2}
\end{equation}
with $ q \in \Qdot, \; \fR(q) \in \Ren, \Psi_0(q) \in \CCC $, or linear combinations of them. We interpret $ \Psi_\fR := \fR \Psi_0 $ as a function in the function space
\begin{equation}
	 \Psi_\fR \in \Ren^{\Qdot} := \big\{ \Psi_\fR: \Qdot \to \Ren \big\} \;.
\end{equation}
The second ESS then covers expressions \eqref{eq:ext2} and linear combinations of them:

\begin{definition}
Consider the free $ \eRen $-vector space of all finite (commutative) sums of the form
\begin{equation}
	\FB_{\ex,0} := \left\{ \Psi = \sum_{m=1}^M \fc_m \Psi_m \; \middle\vert \; \fc_m \in \eRen, \Psi_m \in \Ren^{\Qdot} \right \} \;.
\label{eq:PsiFBex}
\end{equation}
Then, the \textbf{second extended state space (ESS)} is the $ \eRen $-quotient space
\begin{equation}
	\FB_{\ex} = \FB_{\ex,0} /_{\sim_{\Fex}} \;,
\end{equation}
with equivalence relation $ \sim_{\Fex} $ generated by
\begin{equation}
	(c \fc) \Psi_a + \Psi
	\; \sim_{\Fex} \; \fc (c \Psi_a) + \Psi \qquad \text{if } c \in \CCC \;,
\end{equation}
where $ \fc \in \eRen, \; \Psi \in \FB_{\ex,0} $ and $ \Psi_a \in \Ren^{\Qdot} $.\\

\label{def:FBex}
\end{definition}

The first ESS can be embedded into the second ESS $ \FB \hookrightarrow \FB_{\ex} $ by interpreting all $ \Psi_m \in \Sdot_\sF $ as elements $ \Psi_m \in \Ren^{\Qdot} $.\\

\clearpage

\section{Operators on the ESS}
\label{sec:opext}

We first prove that creation and annihilation terms $ A^\dagger(v), A(v) $ as in \eqref{eq:adaggerp} and \eqref{eq:ap} can be defined as operators using ESSs. Here, we even permit form factors $ v(\bp, \bk) $, that are allowed to depend on the fermion momentum $ \bp $. The momentum space definition reads
\begin{equation}
\begin{aligned}
	(A_j^\dagger(v) \Psi)(\bP,\bK) &= \frac{1}{\sqrt{N}} \sum_{\ell = 1}^N v(\bp_j, \bk_{\ell}) \Psi(\bP + e_j \bk_{\ell},\bK \setminus \bk_{\ell}) \;,\\
	(A_j(v) \Psi)(\bP,\bK) &= \sqrt{N+1} \int v(\bp_j - \tilde{\bk}, \tilde{\bk})^* \Psi(\bP - e_j \tilde{\bk}, \bK, \tilde{\bk}) \; d\tilde{\bk} \;,\\
	A^\dagger(v) &= \sum_{j = 1}^M A^\dagger_j(v), \qquad A(v) = \sum_{j = 1}^M A_j(v) \;,\\
\end{aligned}
\label{eq:aadaggervpk}
\end{equation}
which can be seen as a generalization of \eqref{eq:adaggerp} and \eqref{eq:ap}.

\begin{lemma}[$ A^\dagger, A $ are well-defined for $ v(\bp, \bk) $]
Let $ v: \RRR^d \times (\RRR^d \setminus \{0\}) \to \CCC $ be smooth and satisfying the scaling condition \eqref{eq:psiscaling}. Then, \eqref{eq:aadaggervpk} entails well-defined operators
\begin{equation}
	A_j^\dagger(v): \FB \to \FB \;, \qquad
	A_j(v): \FB \to \FB_{\ex} \;,
\end{equation}
which may be restricted\footnote{We use the same notation for all extended or restricted versions of operators, here.} to
\begin{equation}
	A_j^\dagger(v): \Sdot_\sF \to \Sdot_\sF \;, \qquad
	A_j(v): \Sdot_\sF \to \Ren^{\Qdot} \;.
\end{equation}
We may even extend $ A_j^\dagger(v): \FB_{\ex} \to \FB_{\ex} $.\\
\label{lem:aadaggerwelldefined}
\end{lemma}
\begin{proof}
Suppose, $ \Psi \in \Sdot_\sF $ and consider the expression $ A_j^\dagger(v) \Psi $ in \eqref{eq:aadaggervpk}. This is a finite sum over products consisting of two factors. By definition of $ \Sdot_1 $, the first factor $ v(\bp_j, \bk_{\ell}) $ in each product is smooth everywhere except where $ \bk_{\ell} = 0 $. By definition of $ \Sdot_\sF $, the second factor $ \Psi(\bP + e_j \bk_{\ell}, \bK \setminus \bk_{\ell}) $ is smooth at all configurations, at which $ \bK \setminus \bk_{\ell} $ contains no coordinate $ \bk_{\ell'} = 0 $. So the product is smooth on $ \Qdot $. In addition, the factors $ v $ and $ \Psi $ scale polynomially as in \eqref{eq:psiscaling}, so $ A^\dagger(v) \Psi \in \Sdot_\sF $.\\

The expression for $ A(v) \Psi $ in \eqref{eq:ap} is a (possibly divergent) integral for each fixed $ (\bP,\bK) \in \Qdot $. Since both the functions $ \tilde{\bk} \mapsto v(\bp_j - \tilde{\bk}, \tilde{\bk}) $ and $ \tilde{\bk} \mapsto \Psi(\bP - e_j \tilde{\bk}, \bK, \tilde{\bk}) $ are in $ \Sdot_1 $, the integral defines an element $ \fR \in \Ren_1 \subset \Ren $ for each fixed $ (\bP,\bK) $. Thus, $ A(v) \Psi  \in \Ren^{\Qdot} $.\\

Both operators can be extended to $ \FB $ by taking $ \eRen $-linear combinations.\\
For $ \Psi \in \Ren^{\Qdot} $, the expression $ A_j^\dagger(v) \Psi $ in \eqref{eq:aadaggervpk} is again a sum of products, that are all in $ \Ren^{\Qdot} $. By $ \eRen $-linearity, we can then extend $ A_j^\dagger(v) $ to $ \FB_{\ex} $.\\
\end{proof}

We may also define the constituent operators of $ H = H_{0, y} + A^\dagger(v) + A(v) - E_\infty $, as well as $ H_{0, x} $ on $ \FB $:

\begin{proposition}[Constituents of $ H $ are well-defined]
Consider the momentum space definitions of $ H_0 $ \eqref{eq:H0}, $ A^\dagger(v) $ \eqref{eq:adaggerp}, $ A(v) $ \eqref{eq:ap} and $ E_\infty $ \eqref{eq:Ep}, where we only assume $ \theta, \omega, v \in \Sdot_1 $ and an arbitrary self-energy function $ E_1: \RRR^d \to \Ren $. Then, the above four momentum space definitions entail well-defined operators
\begin{equation}
\begin{aligned}
	H_0: &\FB \to \FB \;, \qquad 
	&A^\dagger(v): &\FB \to \FB \;,\\
	A(v): &\FB \to \FB_{\ex} \;, \qquad 	
	&E_\infty: &\FB \to \FB_{\ex} \;.\\
\end{aligned}
\end{equation}
It is also possible to restrict
\begin{equation}
\begin{aligned}
	H_0: &\Sdot_\sF \to \Sdot_\sF \;, \qquad 
	&A^\dagger(v): &\Sdot_\sF \to \Sdot_\sF \;,\\
	A(v): &\Sdot_\sF \to \Ren^{\Qdot} \;, \qquad 	
	&E_\infty: &\Sdot_\sF \to \Ren^{\Qdot} \;,\\
\end{aligned}
\end{equation}
or to extend
\begin{equation}
	H_0: \FB_{\ex} \to \FB_{\ex} \;, \qquad 
	A^\dagger(v): \FB_{\ex} \to \FB_{\ex} \;.
\end{equation}
Further, the statements for $ H_0 $ equally hold true for $ H_{0, x} $ and $ H_{0, y} $.\\
\label{prop:operatorextension}
\end{proposition}
\begin{proof}
Well-definedness and the mapping properties of $ A^\dagger(v) $ and $ A(v) $ are a consequence of Lemma \ref{lem:aadaggerwelldefined}: The function $ v(\bk) $ in \eqref{eq:adaggerp} and \eqref{eq:ap} can be seen as a special case of $ v(\bp, \bk) $ in \eqref{eq:aadaggervpk}. Taking the finite linear sum over $ j \in \{1, \ldots, M\} $ sustains the mapping properties.\\

The operator $ H_0 $ as in \eqref{eq:H0} just multiplies with a function in $ \Sdot_\sF $ in momentum space, so for $ \Psi \in \Sdot_\sF $, we also have $ H_0 \Psi \in \Sdot_\sF $. The same holds true for $ H_{0, x} $ and $ H_{0, y} $. By an analogous argument, $ H_0 $ can be extended to $ H_0: \Ren^{\Qdot} \to \Ren^{\Qdot} $.\\

Finally, we consider $ E_\infty \Psi $ for $ \Psi \in \Sdot_\sF $. By \eqref{eq:Ep}, at each fixed $ (\bP,\bK) \in \Qdot $ the expression $ E_\infty \Psi $ is defined as a finite sum over terms $ E_1(\bp) \Psi(\bP,\bK) \in \Ren $. So indeed, $ E_\infty \Psi \in \Ren^{\Qdot} $.\\

Extensions to $ \FB $ or $ \FB_{\ex} $ can again be done by $ \eRen $-linear combination.\\
\end{proof}

Thus, also the linear operator $ H: \FB \to \FB_{\ex} $ is well-defined.\\

Finally, we prove that the momentum space definition \eqref{eq:aadaggervpk} indeed entails certain canonical commutation relations on the ESS:

\begin{lemma}[Extended CCR]
For $ \phi, \varphi \in \Sdot_1 $, the definitions \eqref{eq:adaggerp} and \eqref{eq:ap} imply the commutation relations
\begin{equation}
	[A_j^\dagger(\varphi), A_{j'}^\dagger(\phi)] = 0 \;, \qquad
	[A_j(\varphi), A_{j'}^\dagger(\phi)] = \begin{cases} \langle \varphi, \phi \rangle \qquad &\text{for } j = j'\\ V_{jj'}(\varphi^* \phi) \qquad &\text{for } j \neq j' \end{cases} \;,
\label{eq:CCR}
\end{equation}
as a strong operator identity. That is, we have operators
\begin{equation}
	[A_j^\dagger(\varphi), A_{j'}^\dagger(\phi)]: \FB \to \FB \;, \qquad
	[A_j(\varphi), A_{j'}^\dagger(\phi)]: \FB \to \FB_{\ex} \;.
\label{eq:CCRoperators}
\end{equation}
Here, the interaction potential $ V_{jj'}: \FB \to \FB_{\ex} $ for momentum transfer from fermion $ j' $ to $ j $ is given by
\begin{equation}
	V_{jj'}(\varphi^* \phi)(\bP, \bK) := \int \varphi^*(\tilde{\bk}) \phi(\tilde{\bk}) \Psi(\bP + (e_{j'} - e_j) \tilde{\bk}, \bK ) \; d\tilde{\bk} \;.
\label{eq:Vjj}
\end{equation}

\label{lem:CCR}
\end{lemma}

\begin{proof}
	By Lemma \ref{lem:aadaggerwelldefined}, the products $ A_j(\varphi) A_{j'}^\dagger(\phi) $ and $ A_{j'}^\dagger(\phi) A_j(\varphi) $ are well-defined as operators $ \FB \to \FB_{\ex} $. A momentum space evaluation renders
\begin{equation}
\begin{aligned}
	&\big( [A_j(\varphi), A_{j'}^{\dagger}(\phi)] \Psi \big) (\bP,\bK)\\
	= &\int \varphi^*(\tilde{\bk}) \phi(\tilde{\bk}) \Psi(\bP + (e_{j'} - e_j) \tilde{\bk}, \bK ) \; d\tilde{\bk}\\
	= &\begin{cases} \langle \varphi,\phi \rangle \Psi(\bP,\bK) \qquad &\text{for } j = j' \\ \int \varphi^*(\tilde{\bk}) \phi(\tilde{\bk}) \Psi(\bP + (e_{j'} - e_j) \tilde{\bk}, \bK ) \; d\tilde{\bk} \qquad &\text{for } j \neq j'
	\end{cases} \;.\\
\end{aligned}
\label{eq:commutatorV}
\end{equation}

Similarly, Lemma \ref{lem:aadaggerwelldefined} establishes that $ A_j^\dagger(\varphi) A_{j'}^\dagger(\phi) $ and $ A_{j'}^\dagger(\phi) A_j^\dagger(\varphi) $ are well-defined operators $ \FB \to \FB $ and a short momentum space calculation verifies that they are equal.\\

\end{proof}

The operator $ V_{jj'} $ defined above can be seen as an interaction potential operator. Under an inverse Fourier transform $ \cF^{-1} $, it amounts to a multiplication operator in position space via
\begin{equation}
	(V_{jj'}(\varphi^* \phi) \Psi)(\bX, \bY) = \cF^{-1}(\varphi^* s)(\bx_j - \bx_{j'}) \Psi(\bX, \bY) \;,
\label{eq:Vjjpositionspace}
\end{equation}
provided that $ \cF^{-1}(\varphi^* \phi) $ exists (e.g., for exact scaling degrees $ \beta_\varphi + \beta_s > -d \Rightarrow \varphi^* s \in \cS'(\RRR^d) $).\\
The following property about $ V_{jj'} $ will become useful in later proofs:

\begin{lemma}
If either of the functions $ \phi, \varphi \in \Sdot_1 $ is an element of $ C_c^\infty(\RRR^d \setminus \{0\}) $, then we even have
\begin{equation}
	V_{jj'}(\varphi^* \phi): \FB \to \FB \;.
\end{equation}
\label{lem:Vrange}
\end{lemma}

\begin{proof}
First, let us consider $ \Psi \in \Sdot_\sF $. Without loss of generality, assume that $ \varphi \in C_c^\infty(\RRR^d \setminus \{0\}) $, so $ \varphi $ is compactly supported, and the function
\begin{equation*}
	\tilde{\bk} \mapsto \phi(\tilde{\bk}) \Psi(\bP + (e_{j'} - e_j) \tilde{\bk},\bK) \;,
\end{equation*}
is smooth everywhere on that support. So the function
\begin{equation*}
	\tilde{\bk} \mapsto \varphi^*(\tilde{\bk}) \phi(\tilde{\bk}) \Psi(\bP + (e_{j'} - e_j) \tilde{\bk} ,\bK) \quad \text{is in} \quad C_c^\infty \;,
\end{equation*}
and the integral over it converges to a $ \CCC $-number.\\
We now show that this number depends smoothly on $ (\bP,\bK) \in \Qdot $: Consider any multi-index $ \alpha $ corresponding to a derivative $ \partial^\alpha $ composed of arbitrarily many partial derivatives $ \partial_{\bp_j}, \partial_{\bk_{\ell}} $ with $ j, \ell \in \NNN $. Then also $ \partial^\alpha \Psi \in \Sdot_\sF $ and by the same arguments as above
\begin{equation}
	\tilde{\bk} \mapsto \varphi^*(\tilde{\bk}) \phi(\tilde{\bk}) \partial^\alpha \Psi(\bP + (e_{j'} - e_j) \tilde{\bk},\bK) \quad \text{is in} \quad C_c^\infty \;.
\end{equation}
So the integral converges absolutely, derivative and integral commute, and we obtain $ \partial^\alpha (V_{jj'}(\varphi^* \phi) \Psi)(\bP,\bK) \in \CCC $ for any multi-index $ \alpha $. Hence, $ (V_{jj'}(\varphi^* \phi) \Psi)(\bP,\bK) $ is smooth at $ (\bP,\bK) \in \Qdot $.\\

Polynomial scaling of $ V_{jj'}(\varphi^* \phi) \Psi $ can be seen as follows: Under the coordinate rotation $ \bp_+ := \bp_j + \bp_{j'} $ and $ \bp_- := \bp_j - \bp_{j'} $, the expression $ V_{jj'} \Psi $ becomes a convolution in $ \bp_- $ of a polynomially scaling function with the function $ \varphi^* \phi \in C_c^\infty $. Polynomial scaling bounds are neither affected by the coordinate rotation nor by the convolution with a $ C_c^\infty $-function. So indeed, $ V_{jj'}(\varphi^* \phi) \Psi \in \Sdot_\sF $, and we have that $ V_{jj'}(\varphi^* \phi): \Sdot_\sF \to \Sdot_\sF $.\\

This mapping property extends to $ V_{j'j}(\varphi^* \phi): \FB \to \FB $ by $ \eRen $-linearity.\\
\end{proof}

\paragraph{Remarks.}

\begin{enumerate}
\setcounter{enumi}{\theremarks}
\item \label{rem:aa} In \eqref{eq:CCR}, we have not included the commutation relation for annihilation operators $ [A_j(\varphi), A_{j'}(\phi)] = 0 $. The reason is that products of two or more annihilation operators are not necessarily defined, since we only have $ A_j(\varphi): \FB \to \FB_{\ex} $. An arbitrary product $ A_j(\varphi) A_{j'}(\phi) \Psi $ with $ \Psi \in \Sdot_\sF $ contains a double integral
\begin{equation}
\begin{aligned}
	&\big( A_j(\varphi) A_{j'}(\phi) \Psi \big)(\bP, \bK)\\
	= &\sqrt{(N+1)(N+2)} \int \int \varphi(\tilde\bk)^* \phi(\tilde\bk')^* \Psi(\bP - e_j \tilde\bk - e_{j'} \tilde\bk', \bK, \tilde\bk, \tilde\bk') \; d \tilde\bk d \tilde\bk' \;,
\end{aligned}
\label{eq:doubleintegral}
\end{equation}
where the first integral produces a configuration space function $ \Qdot \to \Ren $. And a second integral over such a function can generally not be interpreted as an element in $ \Ren $.\\

A definition of such operator products would require a modification of $ \FB_{\ex} $ such that it also accommodates general divergent integrals over multiple coordinates of a $ \Sdot_\sF $-function as in \eqref{eq:doubleintegral}. We postpone the investigation of such choices for $ \FB_{\ex} $ to future investigations.\\
\end{enumerate}
\setcounter{remarks}{\theenumi}

\section{Dressing on the ESS}	
\label{sec:Fren}

Our next step is to define a dressing operator $ W(s) $ with $ s \in \Sdot_1 $. To do so, a naive approach would be to start from the expression $ W(s) = e^{A_M^\dagger(s) - A_M(s)} \ldots e^{A_1^\dagger(s) - A_1(s)} $ (with $ M $ depending on the fermion sector) and expand the exponentials into series
\begin{equation*}
	e^{A_j^\dagger(s) - A_j(s)} = \sum_{n \in \NNN_0}\frac{(A_j^\dagger(s) - A_j(s))^n}{n!} \;,
\end{equation*}
which can be multiplied out. There are two difficulties with this approach:
\begin{itemize}
\item Some terms in the resulting sum contain two or more annihilation operators $ A(s) $ (see Remark \ref{rem:aa}).
\item There is an infinite number of such terms.
\end{itemize}
So with the current definition of $ \FB_{\ex} $ we cannot simply define $ W(s) $ as an operator $ \FB \to \FB_{\ex} $. Instead, we pursue a different approach and define $ W(s) : \cD_W \to \sF_{\ex} $. Here, we choose $ \cD_W \subset \FB_{\ex} $ in \eqref{eq:cDW}, such that $ \cD_\sF $ (which is a symmetrized version of $ \cD_W \cap L^2 $) is dense in $ \sF $.\\

If we consider $ W_\sF(\varphi): \sF \to \sF, \varphi \in \fh $ as a unitary operator on Fock space, together with some suitable $ \Psi \in \sF $, then there is a well-defined expression (similar to \eqref{eq:Ws1}) for $ W_\sF(\varphi) \Psi $ as an $ L^2 $-function on momentum--configuration space. For $ \varphi $ replaced by $ s \in \Sdot_1 $, we may then define $ W(s) \Psi \in \FB $ based on the momentum space expression of $ W_\sF(\varphi) \Psi $.\\

The domain $ \cD_W $ in \eqref{eq:cDW} is generated by vectors $ \Psi_C $ of the form
\begin{equation*}
	\Psi_C = W_1(\varphi) A_j^\dagger(v) \Psi_m \qquad \text{or} \qquad
	\Psi_C = X W_1(\varphi) \Psi_m \;,
\end{equation*}
where
\begin{itemize}
\item $ \Psi_m = \Psi_{mx} \otimes \Omega_y \in \sF $ with $ \Psi_{mx} \in \cS(\cQ_x) $ (i.e., we have a Schwartz function) and $ \Omega_y $ describing the boson field in the vacuum.

\item $ W_1(\varphi) = e^{A_1^\dagger(\varphi) - A_1(\varphi)}, \varphi \in \fh \cap \Sdot_1 $ describes a dressing induced by only the first fermion.

\item $ A_j^\dagger(v), v \in \Sdot_1 $ describes creation by only the fermion with number $ j \in \{1, \ldots, M\} $.

\item $ X $ is a linear combination of $ \Ren_1 $-constants and operators $ V_{jj'}(v^* s) $ that formally commute with $ W(s) $.

\end{itemize}
When setting $ X = 1 $ (which formally commutes with $ W(s) $), we see that $ \cD_W $ contains vectors of the kind $ W_1(\varphi) \Psi_m $. We show in Lemma \ref{lem:W1equivalence}, that these are equal to $ W_{\sF, 1}(\varphi) \Psi_m $ and, after symmetrization, span a dense subspace of $ \sF $ (Lemma \ref{lem:coherentdense}). This will later allow for a dense definition of $ \widetilde{H} $.\\

The definition of $ W(s) \Psi_C $ now exactly works as explained above: We establish a momentum space expression in case $ s, v \in \fh \cap \Sdot_1 $ using Lemma \ref{lem:commutationW}. Then, we generalize to $ s, v \in \Sdot_1 $ by a suitable definition. As discussed in the introduction, we remove certain exponential factors of the form $ e^{V_{jj'}} $ in an ad-hoc modification. Thus, $ W_j(\varphi_j) $ differs from $ W_{\sF, j}(\varphi_j) $ for $ \varphi_j \in \fh \cap \Sdot_1 $. However, we show in Lemma \ref{lem:Xcommute} that any $ V_{j'j''} $ commutes with $ W_{\sF, j}(\varphi_j) $, so the omitted factor $ e^{V_{jj'}} $ can heuristically be ``pulled into any position''. Heuristically, this factor disappears when performing a dressing, which justifies the omission within the computation of $ \widetilde{H} $.\\

\subsection{Bosonic Dressing $ W_y(\varphi) $}
\label{subsec:coherent}

The upcoming proofs are based on some well-known facts about coherent states, where only bosons are present, i.e., $ \Psi_y \in \sF_y $. The momentum space representation of the bosonic creation and annihilation operators $ a^\dagger(v), a(v) $ with form factor $ v \in \fh $ given by
\begin{equation}
\begin{aligned}
	(a^{\dagger}(v) \Psi_y)(\bK) &:= \frac{1}{\sqrt{N}} \sum_{\ell = 1}^N v(\bk_{\ell}) \Psi_y(\bK \setminus \bk_{\ell}) \;,\\
	(a(v) \Psi_y)(\bK) &:= \sqrt{N+1} \int v(\tilde{\bk})^* \Psi_y(\bK ,\tilde{\bk}) \; d \tilde{\bk} \;.
\end{aligned}
\label{eq:ay}
\end{equation}

This definition implies that the commutation relations $ [a(v_1), a^\dagger(v_2)] = \langle v_1, v_2 \rangle $ hold as a strong operator identity on a dense domain in $ \sF_y $. These operators $ a^{\dagger}(v), a(v) $ substantially differ from $ A^{\dagger}(v), A(v) $ defined in \eqref{eq:adaggerp}, \eqref{eq:ap}, which create one boson at each position of a fermion, whereas in $ \sF_y $, there are no fermions.\\

Using $ a^{\dagger}(v), a(v) $, we may define a set of displacement operators
\begin{equation}
	W_y(\varphi) = e^{a^{\dagger}(\varphi) - a(\varphi)} \;,
\label{eq:Wy}
\end{equation} 
and coherent states $ \Psi_y(\varphi):= W_y(\varphi) \Omega_y $. Indeed, $ W_y(\varphi) $ is well-defined, since for $ \varphi \in \fh $, we have the bounds
\begin{equation}
	\big\Vert a^{\dagger}(\varphi) \Psi_y \big\Vert \le \big\Vert (N+1)^{1/2} \Psi_y \big\Vert \; \Vert \varphi \Vert \;, \qquad
	\big\Vert a(\varphi) \Psi_y \big\Vert \le \big\Vert N^{1/2} \Psi_y \big\Vert \; \Vert \varphi \Vert \;,
\label{eq:nelsonestimate}
\end{equation}
so the exponential series \eqref{eq:Wy} converges in norm
\begin{equation}
\begin{aligned}
	 \Psi_y(\varphi) = W_y(\varphi) \Omega_y := &\sum_{k=0}^\infty \frac{1}{k!} (a^{\dagger}(\varphi) - a(\varphi))^k \Omega_y\\
	 \text{where} \quad \left\Vert \frac{1}{k!} (a^{\dagger}(\varphi) - a(\varphi))^k \Omega_y \right\Vert \le &\frac{1}{k!} \big\Vert 2^k (k!)^{1/2} \Omega_y \big\Vert \; \Vert \varphi \Vert^k = (k!)^{-1/2} \Vert 2 \varphi \Vert^k \;.
\end{aligned}
\label{eq:expconvergencefock}
\end{equation}

Here, we used in the second line that $ (a^{\dagger}(\varphi) - a(\varphi))^k \Omega_y $ occupies only sectors in Fock space with $ \le k $ particles, so we can set $ N \le (k-1) $ in \eqref{eq:nelsonestimate}. Subsequent application of \eqref{eq:nelsonestimate} leads to the factor $ (k!)^{1/2} $.\\
In momentum space representation and in terms of tensor products,
\begin{equation}
	\Psi_y(\varphi)(\bK) = e^{-\frac{\Vert \varphi \Vert^2}{2}} \frac{1}{\sqrt{N!}} \left( \prod_{\ell = 1}^N \varphi(\bk_{\ell}) \right) \quad \Leftrightarrow \quad \Psi_y(\varphi) = \sum_{N=0}^\infty \frac{e^{-\frac{\Vert \varphi \Vert^2}{2}}}{\sqrt{N!}} \; \underbrace{\varphi \otimes_S \ldots \otimes_S \varphi}_{N \text{ times}} \;.
\label{eq:coherentstatep}
\end{equation}

A calculation similar to \eqref{eq:expconvergencefock} verifies that $ W_y(\varphi) $ can be defined on all $ \Psi_y $ with finite particle number, i.e., $ \Psi_y \in \sF_{\fin, y} $ with
\begin{equation}
	\sF_{\fin, y} := \big\{ \Psi_y \in \sF_y \; \big\vert \; \exists N_{\max} \in \NNN: \; \Psi_y^{(N)} = 0 \; \forall N > N_{\max}  \big\} \;.
\label{eq:sFfiny}
\end{equation}
And by continuity, we can thus define $ W_y(\varphi) $ on all of $ \sF_y $.\\

Moreover, the $ W_y(\varphi) $ are unitary, so $ \big\Vert \Psi_y(\varphi) \big\Vert = 1 $, and they satisfy the Weyl relations. Further, it is a well-known fact that the span of the set of coherent states $ \{\Psi_y(\varphi) \; \mid \; \varphi \in \fh\} $ is dense in $ \sF_y $ \cite[Prop. 12]{I02}. In addition,
\begin{equation}
	\langle \Psi_y(\varphi_1),\Psi_y(\varphi_2) \rangle_{\sF_y} = e^{-\frac{\Vert \varphi_1 \Vert^2 + \Vert \varphi_2 \Vert^2}{2}} e^{\langle \varphi_1,\varphi_2 \rangle} \;,
\label{eq:scalarproductcoherent}
\end{equation}
so
\begin{equation}
	\big\Vert \Psi_y(\varphi_1) - \Psi_y(\varphi_2) \big\Vert^2 =
	\big\Vert \Psi_y(\varphi_1) \big\Vert^2 + \big\Vert \Psi_y(\varphi_2) \big\Vert^2 - 2 e^{-\frac{\Vert \varphi_1 \Vert^2 + \Vert \varphi_2 \Vert^2}{2}} \Re \big( e ^{\langle \varphi_1,\varphi_2 \rangle} \big) \;.
\end{equation}
As $ \fh $ is separable, we can find a countable dense set $ (\varphi_n)_{n \in \NNN} $ in $ \fh $, such that $ (\Psi_y(\varphi_n))_{n \in \NNN} $ is also dense in the coherent states. So
\begin{equation}
	\span \big\{ \Psi_y(\varphi_n) \; \big\vert \; n \in \NNN \big\} \quad \text{is dense in } \sF_y \;.
\label{eq:coherentdense}
\end{equation}

Since $ \fh \cap \Sdot_1 \supset C_c^\infty $ is dense in $ \fh $ the above statements hold true, if we replace $ \varphi_n \in \fh $ by $ \varphi_n \in \fh \cap \Sdot_1 $.\\

\subsection{Dressing Induced by Fermions $ W_1(\varphi) $}
\label{subsec:coherent2}

\subsubsection{Density}

Now, let us turn to the case with two particle species, i.e., $ \sF = \sF_x \otimes \sF_y $ and $ A^{\dagger}, A $ instead of $ a^{\dagger}, a $. In order to make an analogous statement to \eqref{eq:coherentdense} work, we restrict from $ A^{\dagger}(\varphi) = \sum_{j=1}^M A_j^\dagger(\varphi) $ to $ A_1^{\dagger}(\varphi) $, i.e., creation by only the first fermion. Just as $ W_y(\varphi) $, the Fock space operator $ W_{\sF, 1}(\varphi) = e^{A_1^\dagger(\varphi) - A_1^\dagger(\varphi)} $ can be defined in analogy to $ W_y(\varphi) $. The operators $ A_1^{\dagger}(\varphi) $ and $ W_{\sF, 1}(\varphi) $ break the fermionic symmetry, so they map $ \sF \to L^2(\cQ_x) \otimes \sF_y $ (instead of $ \sF \to \sF $). We will therefore proceed by considering vectors $ \Psi \in L^2(\cQ_x) \otimes \sF_y $, i.e., with only bosonic exchange symmetry.\\

As a ``cyclic set'' of vectors $ \Psi_m $, used for generating further domains, we choose
\begin{equation}
	\cC_{WS} := \cS(\cQ_x) \otimes \{\Omega_y\} \subset L^2(\cQ_x) \otimes \sF_y \;.
\label{eq:cCWS}
\end{equation}
Since the boson field is in the vacuum, $ \cC_{WS} $ is obviously not dense in $ L^2(\cQ_x) \otimes \sF_y $. However, it generates a dense subspace by applying operators $ W_{\sF, 1}(\varphi) $ to it. The momentum representation after such an application is given by
\begin{equation}
	 (W_{\sF, 1}(\varphi) \Psi_m)(\bP,\bK) = \frac{e^{-\frac{\Vert \varphi \Vert^2}{2}}}{\sqrt{N!}} \left( \prod_{\ell = 1}^N \varphi(\bk_{\ell}) \right) \Psi_{mx} \left( \bp_1 + \sum_{\ell = 1}^N \bk_{\ell}, \bp_2, \ldots, \bp_M \right) \;.
\label{eq:W1s}
\end{equation}

\begin{definition}
The \textbf{span of coherent states created by the first fermion} is given by
\begin{equation}
	\cD_{WS} := \span \big\{ W_{\sF, 1}(\varphi)\Psi_m \; \big\vert \; \varphi \in \fh \cap \Sdot_1, \; \Psi_m \in \cC_{WS} \big\} \;.
\label{eq:cDWS}
\end{equation}
\label{def:sco}
\end{definition}
With \eqref{eq:cDWS}, the following is true.
\begin{lemma}
$ (S_- \otimes \id)[\cD_{WS}] $ is dense in $ \sF = \sF_x \otimes \sF_y $.
\label{lem:coherentdense}
\end{lemma}

The proof is based on denseness of coherent states in $ \sF_y $ and can be found in Appendix \ref{sec:coherentdense}.\\

\subsubsection{Dressed One-Boson States}

Just as $ W_{\sF, 1}(\varphi) $, we may define $ W_{\sF, j}(\varphi), j \in \NNN $ and $ W_\sF(\varphi) = W_{\sF, M}(\varphi) \ldots W_{\sF, 1}(\varphi) $ with $ \varphi \in \fh \cap \Sdot_1 $ on each $ M $-fermion sector. These operators are all unitary on $ L^2(\cQ_x) \otimes \sF_y $ and well-defined on $ \sF $.\\
We will now establish some useful commutation relations on the dense subspace
\begin{equation}
	\sF_{\fin} := \sF_x \otimes \sF_{\fin, y} \subset \sF \;,
\label{eq:sFfinESS}
\end{equation}
with $ \sF_{\fin, y} $ defined in \eqref{eq:sFfiny}.\\

\begin{lemma}[Commutation relations for $ W_\sF $]
For $ \varphi, \phi \in \fh \cap \Sdot_1 $, we have the following strong operator identities\footnote{By a strong operator identity $ A = B $ for $ A, B: \sF \to L^2(\cQ_x) \otimes \sF_y $, we mean that $ A \Psi = B \Psi \; \forall \Psi \in \sF $, even if possibly $ A \Psi, B \Psi \notin \sF $.} on $ \sF_{\fin} $:
\begin{equation}
\begin{aligned}
	W_{\sF, j}(\varphi) A_{j'}^{\dagger}(\phi) &= \begin{cases} \big( A_{j'}^{\dagger}(\phi) - \langle \varphi,\phi \rangle \big) W_{\sF, j}(\varphi) \quad &\text{if } j = j' \\ \big( A_{j'}^{\dagger}(\phi) - V_{jj'}(\varphi^* \phi) \big) W_{\sF, j}(\varphi) \quad &\text{if } j \neq j' \end{cases} \;,\\
	W_\sF(\varphi) A_{j'}^{\dagger}(\phi) &= \big( A_{j'}^{\dagger}(\phi) - \langle \varphi,\phi \rangle - V_{\bullet j'}(\varphi^* \phi) \big) W_\sF(\varphi) \;,\\
	W_{\sF, j}(\varphi) A^{\dagger}(\phi) &= \big( A^{\dagger}(\phi) - \langle \varphi,\phi \rangle - V_{j \bullet}(\varphi^* \phi) \big) W_{\sF, j}(\varphi) \;,\\
	W_\sF(\varphi) A^{\dagger}(\phi) &= \big( A^{\dagger}(\phi) - M \langle \varphi,\phi \rangle - V(\varphi^* \phi) \big) W_\sF(\varphi) \;,\\
\end{aligned}
\label{eq:WAdaggers}
\end{equation}
as well as
\begin{equation}
\begin{aligned}
	W_{\sF, j'}(\varphi) A_j(\phi) &= \begin{cases} \big( A_j(\phi) - \langle \phi,\varphi \rangle \big) W_{\sF, j'}(\varphi) \quad &\text{if } j = j' \\ \big( A_j(\phi) - V_{jj'}(\phi^* \varphi) \big) W_{\sF, j'}(\varphi) \quad &\text{if } j \neq j' \end{cases} \;,\\
	W_\sF(\varphi) A_j(\phi) &= \big( A_j(\phi) - \langle \phi,\varphi \rangle - V_{j \bullet}(\phi^* \varphi) \big) W_\sF(\varphi) \;,\\
	W_{\sF, j'}(\varphi) A(\phi) &= \big( A(\phi) - \langle \phi,\varphi \rangle - V_{\bullet j'}(\phi^* \varphi) \big) W_{\sF, j'}(\varphi) \;,\\
	W_\sF(\varphi) A(\phi) &= \big( A(\phi) - M \langle \phi,\varphi \rangle - V(\phi^* \varphi) \big) W_\sF(\varphi) \;,\\
\end{aligned}
\label{eq:WAs}
\end{equation}
with $ V_{jj'} $ defined in \eqref{eq:Vjj} and\footnote{Here, $ \sum_{j:j \neq j'} $ is to be understood as a sum over only $ j $, while $ \sum_{j \neq j'} $ is a sum over both $ j $ and $ j' $. The second kind of sum will appear more often, so we give it a shorter notation.}
\begin{equation}
\begin{aligned}
	V_{\bullet j'}(\varphi^* \phi) &:= \sum_{j:j \neq j'} V_{jj'}(\varphi^* \phi) \;, \qquad
	V_{j \bullet}(\varphi^* \phi) 	:= \sum_{j':j \neq j'} V_{jj'}(\varphi^* \phi) \;,\\
	V(\varphi^* \phi) 				&:= \sum_{j \neq j'} V_{jj'}(\varphi^* \phi) \;.\\
\end{aligned}
\label{eq:V1}
\end{equation}

\label{lem:commutationW}
\end{lemma}

The proof of Lemma \ref{lem:commutationW} is straightforward by applying the CCR. We present it in Appendix \ref{sec:commutationW}.\\

\subsection{Extended Dressing $ W(s) $}
\label{subsec:extdress}

We would now like the relations in \eqref{eq:WAdaggers} to also hold true if we replace $ \varphi, \phi \in \fh \cap \Sdot_1 $ by $ s, v \in \Sdot_1 $. In that case, the Fock space operators $ W_{\sF, j} $ turn into extended operators $ W_j $. More precisely, it would be desirable to have
\begin{equation}
	W_j(s) A_{j'}^{\dagger}(v) \Psi_m = \begin{cases} \big( A_{j'}^{\dagger}(v) - \langle s,v \rangle \big) W_j(s) \Psi_m \quad &\text{if } j = j' \\ \big( A_{j'}^{\dagger}(v) - V_{jj'}(s^* v) \big) W_j(s) \Psi_m \quad &\text{if } j \neq j' \end{cases} \;,\\
\label{eq:WAdaggerextended}
\end{equation}
for $ \Psi_m \in \cC_{WS} $. By Lemma \ref{lem:CCR}, and since $ \langle s, v \rangle \in \Ren_1 $, we may obviously interpret
\begin{equation}
	V_{jj'}(s^* v): \FB \to \FB_{\ex} \;, \qquad
	\langle s, v \rangle: \FB \to \FB_{\ex} \;.
\end{equation}
So if we can define $ W_j(s) \Psi_m \in \FB $, then the right-hand side of \eqref{eq:WAdaggerextended} serves as a definition for $ W_j(s) A_{j'}^{\dagger}(v) \Psi_m \in \FB_{\ex} $. And if we can further define products like $ W(s) W_1(\varphi) \Psi_m = W_M(s) \ldots W_1(s) W_1(\varphi) \Psi_m \in \FB $, then a generalization of \eqref{eq:WAdaggers} may even be used to define $ W(s) W_1(\varphi) A_{j'}^{\dagger}(v) \Psi_m \in \FB_{\ex} $.\\
However, before doing so, it is first necessary to specify what $ W_1(s) W_1(\varphi) $ is, which we will do by introducing some ``extended Weyl relations''.\\

\subsubsection{Extended Weyl Relations}
\label{subsec:WB}

In order to treat products of factors $ W_j(s), s \in \Sdot_1 $, we introduce an extended Weyl algebra $ \WB $ that is generated by all $ W_j(s) $ and taken over the field $ \eRen $ (as in Definition \ref{def:eren}). Recall that $ \fc \in \eRen $ is a fraction of linear combinations of exponentials $ e^\fr $, with $ \fr \in \Ren_1 $ being a possibly divergent integral (see Definition \ref{def:ren1}). Multiplication on $ \WB $ is defined by the Weyl relations
\begin{equation}
\begin{aligned}
	W_j(s)^{-1} &= W_j(-s) \;,\\
	W_j(s_1) W_j(s_2) &= e^{-\frac{i}{2} \sigma(s_1,s_2)} W_j(s_1+s_2) \;,\\
\end{aligned}
\label{eq:WeylrelationsESS}
\end{equation}
with symplectic form
\begin{equation}
\begin{aligned}
	\sigma &= \Sdot_1 \times \Sdot_1 \to \Ren_1 \;,\\
	(s_1,s_2) &\mapsto \langle s_1,s_2 \rangle - \langle s_2,s_1 \rangle \;.\\
\end{aligned}
\label{eq:symplecticform}
\end{equation}
Note that $ e^{-\frac{i}{2} \sigma(s_1,s_2)} = e^{\Im\langle s_1,s_2 \rangle} \in \eRen $ is not necessarily a complex number.\\

This $ \WB $ can be seen as an ``almost-extension'' of the usual Weyl algebra generated by $ \{W_j(s) \; \mid \; s \in \fh \} $ (strictly speaking, it only extends some Weyl algebra $ \cW_0 $ generated by $ \{W_j(s) \; \mid \; s \in \fh \cap \Sdot_1\} $).\\

The definition of the extended Weyl algebra now allows us to write
\begin{equation}
\begin{aligned}
	W(s) W_1(\varphi) = &W_M(s) \ldots W_2(s) W_1(s) W_1(\varphi)\\
	= &e^{\Im\langle s,\varphi \rangle} W_M(s) \ldots W_2(s) W_1(s + \varphi) \;.\\
\end{aligned}
\label{eq:WW1rewriting}
\end{equation}
So $ W(s) W_1(\varphi) $ can be brought into the form $ W_M(s_M) \ldots W_1(s_1) $ times an $ \eRen $-factor, which is the same on each $ M $-fermion sector.\\

\subsubsection{Extended Dressing on Coherent States}
\label{subsec:dresssco}

In order to define vectors of the kind $ W_M(s_M) \ldots W_1(s_1) \Psi_m \in \FB $, we make use of the momentum space definition of $ W_{\sF, M}(\varphi_M) \ldots W_{\sF, 1}(\varphi_1) \Psi_m $ for $ \varphi_{\ell} \in \fh \cap \Sdot_1 $. For two dressing operators with $ j \neq j' $, the Baker--Campbell--Hausdorff formula implies
\begin{equation}
\begin{aligned}
	W_{\sF, j}(\varphi_j) W_{\sF, j'}(\varphi_{j'}) \Psi_m
	&= e^{A_j^\dagger(\varphi_j) - A_j(\varphi_j)} e^{A_{j'}^\dagger(\varphi_{j'}) - A_{j'}(\varphi_{j'})} \Psi_m\\
	&= e^{- \frac{\Vert \varphi_j \Vert^2}{2} - \frac{\Vert \varphi_{j'} \Vert^2}{2}} e^{A_j^\dagger(\varphi_j)} e^{- A_j(\varphi_j)} e^{A_{j'}^\dagger(\varphi_{j'})} \underbrace{e^{- A_{j'}(\varphi_{j'})} \Psi_m}_{= \Psi_m}\\
	&= e^{- \frac{\Vert \varphi_j \Vert^2}{2} - \frac{\Vert \varphi_{j'} \Vert^2}{2} - [A_j(\varphi_j), A_{j'}^\dagger(\varphi_{j'})]} e^{A_j^\dagger(\varphi_j)} e^{A_{j'}^\dagger(\varphi_{j'})} e^{- A_j(\varphi_j)} \Psi_m\\
	&= e^{- \frac{\Vert \varphi_j \Vert^2}{2} - \frac{\Vert \varphi_{j'} \Vert^2}{2} - V_{jj'}(\varphi_j^* \varphi_{j'})} e^{A_j^\dagger(\varphi_j) + A_{j'}^\dagger(\varphi_{j'})} \Psi_m \;.
\end{aligned}
\label{eq:WjWjprime}
\end{equation}
Here, we used that $ V_{jj'} $ commutes with all $ A^\dagger_{j''} $ and $ A_{j''} $, which follows by the same arguments as in the proof of Lemma \ref{lem:Xcommute} below. Thus, all double commutators between $ A $- and $ A^\dagger $-operators vanish. The generalization to arbitrarily many factors, $ W_{\sF, M}(\varphi_M) \ldots W_{\sF, 1}(\varphi_1) \Psi_m $ is straightforward, and we obtain an exponential of constants and $ V_{jj'} $-terms, followed by $ e^{A_M^\dagger(\varphi_M) + \ldots + A_1^\dagger(\varphi_1)} \Psi_m $.\\
We now define $ W_M(\varphi_M) \ldots W_1(\varphi_1) $ by dropping the $ V_{jj'} $-terms in \eqref{eq:WjWjprime}, which yields the following momentum space expression for $ \Psi_m = \Psi_{mx} \otimes \Omega_y \in \cC_{WS} $
\begin{equation}
	(W_M(\varphi_M) \ldots W_1(\varphi_1) \Psi_m)(\bP, \bK) := \frac{1}{\sqrt{N!}} e^{-\sum_{j = 1}^M \frac{\Vert \varphi_j \Vert^2}{2}} \sum_{\sigma} \left( \prod_{\ell = 1}^N \varphi_{\sigma(\ell)}(\bk_{\ell}) \right) \Psi_{mx} \left( \bP' \right) \;,
\label{eq:WWWPsimomentum}
\end{equation}
where the sum over $ \sigma $ runs over all $ M^N $ maps
\begin{equation}
	\sigma: \{1, \ldots, N\} \to \{1, \ldots, M\} \;,
\end{equation}
assigning each boson $ \ell $ to a fermion $ j = \sigma(\ell) $. The shifted momentum, illustrated in Figure \ref{fig:dressingmomentum}, is then
\begin{equation}
	\bP' := \bP + \sum_{\ell} e_{\sigma(\ell )} \bk_{\ell} \;.
\end{equation}
\begin{figure}[hbt]
	\centering
	\begin{tikzpicture}
\begin{feynhand}
	
	\vertex (a1) at (0,0);
	\vertex (a2) at (6,0);
	\vertex (b1) at (0,0.5);
	\vertex (b2) at (6,0.5);
	\vertex (c1) at (6,1);
	\vertex (c2) at (6,1.5);
	\vertex (c3) at (6,2);
	\vertex [dot] (o1) at (1.8,0.5) {};
	\vertex [dot] (o2) at (3.3,0) {};
	\vertex [dot] (o3) at (4.8,0.5) {};
	\propag [fermion] (a1) node[anchor = east]{$ \bp_1 + \bk_2 $} to (a2) node[anchor = west]{$ \bp_1 $};
	\propag [fermion] (b1) node[anchor = east]{$ \bp_2 + \bk_1 + \bk_3 $} to (b2) node[anchor = west]{$ \bp_2 $};
	\propag [photon] (o1) to [out = 70, in = 180] (c1) node[anchor = west]{$ \bk_3 \quad \sigma(3) = 2 $};
	\propag [photon] (o2) to [out = 90, in = 180] (c2) node[anchor = west]{$ \bk_2 \quad \sigma(2) = 1 $};
	\propag [photon] (o3) to [out = 90, in = 180] (c3) node[anchor = west]{$ \bk_1 \quad \sigma(1) = 2 $};
	\draw [decorate,decoration={brace, amplitude=5pt, mirror}]
(9.5,0.8) -- (9.5,2.2)node  [midway,anchor = west,xshift = 5pt] {$ \bK $};
	\draw [decorate,decoration={brace, amplitude=5pt, mirror}]
(9.5,-0.2) -- (9.5,0.7)node  [midway,anchor = west,xshift = 5pt] {$ \bP $};
	\draw [decorate,decoration={brace, amplitude=5pt}]
(-3,-0.2) -- (-3,0.7)node  [midway,anchor = east,xshift = -5pt] {$ \bP' $};
\end{feynhand}
\end{tikzpicture}
	\caption{Visual representation of the momentum shift within the dressing. The wiggly lines represent bosons that get attached to fermions (solid lines). $ \sum_\sigma $ runs over all $ M^N $ possible ways to attach all bosons to the fermions. For a fixed $ \sigma $, we obtain the entries of $ \bP' $ by summing up all particle momenta on the respective fermion line.}
	\label{fig:dressingmomentum}
\end{figure}

\begin{lemma}[Products of $ W $ are well-defined]
Consider a sequence $ (s_j)_{j \in \NNN} \subset \Sdot_1 $ and $ \Psi_m \in \cC_{WS} $. Then, the momentum space definition \eqref{eq:WWWPsimomentum} renders a well-defined vector
\begin{equation}
	W_M(s_M) \ldots W_1(s_1) \Psi_m \in \FB \;,
\label{eq:WWPsim}
\end{equation}
where \eqref{eq:WWPsim} is to be interpreted as a sector-wise definition in $ M \in \NNN $.\\
\label{lem:WjsWjs}
\end{lemma}
\begin{proof}
Copying the momentum space definition \eqref{eq:WWWPsimomentum}, we obtain
\begin{equation}
	(W_M(s_M) \ldots W_1(s_1) \Psi_m)(\bP, \bK) := \frac{1}{\sqrt{N!}} e^{-\sum_{j = 1}^M \frac{\Vert s_j \Vert^2}{2}} \sum_{\sigma} \left( \prod_{\ell = 1}^N s_{\sigma(\ell)}(\bk_{\ell}) \right) \Psi_{mx} \left( \bP' \right) \;.
\label{eq:WWWsPsimomentum}
\end{equation}
Obviously, $\left( \prod_{\ell = 1}^N s_{\sigma(\ell)}(\bk_{\ell}) \right) \Psi_{mx} \left( \bP' \right) $ defines a function in $ \Sdot_\sF $, which is still true after taking the finite sum over $ \sigma $.\\

Further, we have $ \Vert s_j \Vert^2 = \langle s_j, s_j \rangle \in \Ren_1 $, so $ e^{-\sum_{j = 1}^M \frac{\Vert s_j \Vert^2}{2}} \in \eRen $.\\

Therefore, the expression \eqref{eq:WWWPsimomentum} defines an element of $ \FB $.\\
\end{proof}

This already allows us to define $ W_M(s_M) \ldots W_1(s_1) $ on vectors $ \Psi_m \in \cC_{WS} $ with the boson field in the vacuum. In order to define $ W_M(s_M) \ldots W_1(s_1) $ also on a dense domain in $ \sF $, we extend the definition to vectors $ W_{\sF, 1}(\varphi) \Psi_m \in \cD_{WS} $, whose symmetrized span, by Lemma \ref{lem:coherentdense}, is dense in $ \sF $. This extension is done by assuming that $ W_{\sF, 1}(\varphi) $ can be merged into $ W_1(s_1) $, just as $ W_1(\varphi) $ in \eqref{eq:WW1rewriting}.\\
We will also allow for a treatment of state vectors by using the operator $ (S_- \otimes \id) $, which can obviously be extended to $ (S_- \otimes \id): \FB \to \FB $ or $ (S_- \otimes \id): \FB_{\ex} \to \FB_{\ex} $, using the momentum space definition \eqref{eq:S+S-}.\\

\begin{definition}
Let $ (s_j)_{j \in \NNN} \subset \Sdot_1 $. Then, by Lemma \ref{lem:WjsWjs}, copying the momentum space definition \eqref{eq:WWWPsimomentum} results in a well-defined \textbf{product of dressing operators}
\begin{equation}
\begin{aligned}
	W_M(s_M) \ldots W_1(s_1)&: \cD_{WS} \to \FB \;,\\
	W_M(s_M) \ldots W_1(s_1) W_{\sF, 1}(\varphi) \Psi_m &:= e^{\Im \langle s_1, \varphi \rangle} W_M(s_M) \ldots W_1(s_1 + \varphi) \Psi_m \;,
\end{aligned}
\end{equation}
where $ M $ is the respective fermion number on each sector. Further, we define the extension to \textbf{symmetrized vectors}
\begin{equation}
	W_M(s_M) \ldots W_1(s_1): (S_- \otimes \id)[\cD_{WS}] \cup \cD_{WS} \to \FB \;,
\end{equation}
by imposing that $ W_M(s_M) \ldots W_1(s_1) $ shall commute with the symmetrization operator $ (S_- \otimes \id) $.\\
\label{def:WjsWjs}
\end{definition}

With this definition, it is true that
\begin{lemma}
For all $ \varphi \in \fh \cap \Sdot_1 $ and $ \Psi_m \in \cC_{WS} $, it holds that
\begin{equation}
	W_{\sF, 1}(\varphi) \Psi_m = W_1(\varphi) \Psi_m \;,
\label{eq:W1equivalence}
\end{equation}
in terms of momentum space functions.\\
\label{lem:W1equivalence}
\end{lemma}
\begin{proof}
Consider \eqref{eq:WjWjprime} with $ j = 1 $ and $ \varphi_{j'} = 0 $. Then, the $ V_{jj'} $-term vanishes, so no $ V_{jj'} $-terms are dropped when copying momentum space expressions in the transition $ W_{\sF, 1} \to W_1 $ and indeed $ W_{\sF, 1}(\varphi) \Psi_m = W_1(\varphi) \Psi_m $.\\
\end{proof}

\paragraph{Remarks.}

\begin{enumerate}
\setcounter{enumi}{\theremarks}
\item It may seem natural to extend Definition \ref{def:WjsWjs} to a general $ \Psi \in \sF $. By Lemma \ref{lem:coherentdense}, we can write $ \Psi $ as a symmetrized version of $ \Psi' = \sum_{nn'} W_1(\varphi_n) \Psi_{n'} $ with $ \varphi_n \in \fh \cap \Sdot_1 $ and $ \Psi_{n'} \in \cC_{WS} $. In that case, $ W(s) \Psi = \sum_{nn'} W(s) W_1(\varphi_n) \Psi_{n'} $ contains a possibly infinite sum over functions $ \Qdot \to \CCC $, which may not converge.\\
However, our aim is to give a dense definition of $ \widetilde{H}: \sF \supset \dom(\widetilde{H}) \to \sF $, so it suffices to consider the action of $ W(s) $, and $ H W(s) $ on a dense subset of $ \sF $, such as $ (S_- \otimes \id)[\cD_{WS}] $.\\

\item Concerning the renormalization classes: Two ESS vectors $ W(s) W_1(\varphi)\Psi_m $ and $ W(s) W_1(\tilde{\varphi})\Psi_{m'} $ with $ \Psi_m $ and $ \Psi_{m'} $ concentrated on the same $ M $-fermion sector can be added if the wave function renormalizations $ \fc = e^{\fr}, \fr = - \frac{\Vert s \Vert^2}{2} $ belong to the same renormalization factor class, i.e.,
\begin{equation}
\begin{aligned}
	\fr - \tilde{\fr} \in \CCC \quad &\Leftrightarrow \quad \left\vert \; \Vert s + \varphi \Vert^2 - \Vert s + \tilde{\varphi} \Vert^2 + 2\Im \langle s,\varphi - \tilde{\varphi} \rangle \; \right\vert < \infty\\
	&\Leftrightarrow \quad \left\vert 2 \Re \langle s,\varphi - \tilde{\varphi} \rangle + 2 \Im \langle s,\varphi - \tilde{\varphi} \rangle + \underbrace{\Vert \varphi \Vert^2}_{< \infty} - \underbrace{\Vert \tilde{\varphi} \Vert^2}_{< \infty} \right\vert < \infty\\
	& \Leftarrow \quad | \langle s,\varphi - \tilde{\varphi} \rangle | < \infty \;.
\end{aligned}
\label{eq:rclass}
\end{equation}
That means, convergence of the integral $ \int s(\bk)^*(\varphi(\bk) - \tilde{\varphi}(\bk)) \; d\bk $ ensures that the renormalization classes coincide. Note that both $ \Re $ and $ \Im $ above may be infinite, but cancel each other out.\\
\end{enumerate}
\setcounter{remarks}{\theenumi}

\subsubsection{Extended Dressing on One-Boson States}
\label{subsec:dressscofin}

Now, as announced, when replacing $ \varphi, \phi \in \fh \cap \Sdot_1 $ by $ v, s \in \Sdot_1 $ in \eqref{eq:WAdaggers}, we obtain a well-defined right-hand side. This allows for the following extension of dressing operator products

\begin{definition}
Let $ v \in \Sdot_1, (s_j)_{j \in \NNN} \subset \Sdot_1 $ and $ \Psi_m \in \cC_{WS} $. We extend the \textbf{product of dressing operators} to one-boson states via
\begin{equation}
	W_M(s_M) \ldots W_1(s_1) A_{j'}^\dagger(v) \Psi_m \in \FB_{\ex} \;,
\end{equation}
where $ M $ is the respective fermion number on each sector, via
\begin{equation}
\begin{aligned}
	W_M(s_M) \ldots W_1(s_1) A_{j'}^\dagger(v) \Psi_m
	:= &\left( A_{j'}^\dagger(v) - \sum_{j = 1}^M X_j \right) W_M(s_M) \ldots W_1(s_1) \Psi_m \\
	\text{with} \quad X_j = &\begin{cases} \langle s,v \rangle \quad &\text{if } j = j' \\ V_{j j'}(s^* v) \quad &\text{if } j \neq j' \end{cases}\;.
\end{aligned}
\label{eq:WjsWjsAdagger}
\end{equation}
This operator can further be extended to \textbf{symmetrized vectors} by imposing that $ W_M(s_M) \ldots W_1(s_1) $ shall commute with the symmetrization operator $ (S_- \otimes \id) $.
\label{def:WjsWjsAdagger}
\end{definition}

It is easy to see that the right-hand side of \eqref{eq:WjsWjsAdagger} makes sense: By Lemma \ref{lem:WjsWjs}, we have $ W_M(s_M) \ldots W_1(s_1) \Psi_m \in \FB $. Lemma \ref{lem:aadaggerwelldefined} implies that $ A_{j'}^\dagger(v): \FB \to \FB $ and by Lemma \ref{lem:CCR} and $ \langle s,v \rangle \in \Ren_1 $, we have that $ X_j: \FB \to \FB_{\ex} $.\\

Heuristically, the factors $ X_j $ now commute with $ W_{j'}(s) $, since
\begin{lemma}
For $ \varphi, \varphi', \phi \in \fh $ it is true that
\begin{equation}
	[W_{\sF, j}(\varphi'), V_{j'j''}(\varphi^* \phi)] = 0 \qquad \text{and} \qquad [W_{\sF, j}(\varphi'), \langle \varphi,\phi \rangle ] = 0 \;,
\label{eq:Xcommute}
\end{equation}
as a strong operator identity on $ \sF $.
\label{lem:Xcommute}
\end{lemma}
\begin{proof}
Since $ \varphi, \phi \in L^2 $, we have $ \varphi^* \phi \in L^1 $, so after a Fourier transform, the operator $ V_{j'j''}(\varphi^* \phi) $ amounts to a multiplication by an $ L^\infty $-function, and is hence bounded. Further, $ W_{\sF, j}(\varphi') $ is unitary on $ L^2(\cQ_x) \otimes \sF_y $ (and hence bounded). So the commutator is defined on all of $ L^2(\cQ_x) \otimes \sF_y $ and hence $ \sF $.\\
Now, in position space, both $ A_j^\dagger(\varphi') $ and $ A_j(\varphi') $ can be decomposed into a fiber integral by fiber-decomposing $ L^2(\cQ_x) \otimes \sF_y = \int_{\cQ_x} \sF_y \; d \bX $ (see \eqref{eq:adaggerx} and \eqref{eq:ax}). So we can also decompose $ W_{\sF, j}(\varphi') = e^{A_j^\dagger(\varphi') - A_j(\varphi')} $ into a fiber integral. And by \eqref{eq:Vjjpositionspace}, the operator  $ V_{j'j''}(\varphi^* \phi) $ just amounts to a multiplication by a complex constant on each fiber Hilbert space. So the fiber operators commute on all fibers and hence the original operators commute on $ L^2(\cQ_x) \otimes \sF_y $ and $ \sF $.\\

The expression $ \langle \varphi,\phi \rangle $ is just a constant, so it trivially commutes with $ W_{\sF, j}(\varphi') $.\\
\end{proof}

Mathematically, if we replace $ \varphi, \varphi', \phi \in \fh $ by $ s, s', v \in \Sdot_1 $, then the commutation relations \eqref{eq:Xcommute} are not a priori valid, since $ W_j(s) $ is not necessarily defined on vectors of the kind $ V_{j'j''}(\varphi^* \phi) \Psi_m $ or $ \langle \varphi,\phi \rangle \Psi_m $. We enforce their validity by taking \eqref{eq:Xcommute} as a definition for an extension of $ W_j(s) $.
\begin{definition}
Let $ (s_j)_{j \in \NNN} \subset \Sdot_1 $, $ \varphi \in \fh \cap \Sdot_1 $ and $ \Psi_m \in \cC_{WS} $, and let $ X $ be an element of the set of operators
\begin{equation}
	\cX := \span_{\eRen} \big\{ \langle s,v \rangle, \; V_{j j'}(s^* v) \; \big\vert \; s, v \in \Sdot_1 \big\} \;,
\label{eq:cX}
\end{equation}
so $ X $ formally commutes with all $ W_j(s_j) $. Then we extend the \textbf{product of dressing operators} via
\begin{equation}
\begin{aligned}
	W_M(s_M) \ldots W_1(s_1) X \Psi_m
	:= &X W_M(s_M) \ldots W_1(s_1) \Psi_m \;,\\
	W_M(s_M) \ldots W_1(s_1) X W_{\sF, 1}(\varphi) \Psi_m
	:= &X W_M(s_M) \ldots W_1(s_1) W_{\sF, 1}(\varphi) \Psi_m \\
	= &X W_M(s_M) \ldots W_1(s_1) W_1(\varphi) \Psi_m\;,\\
\end{aligned}
\label{eq:WjsWjsX}
\end{equation}
with $ M $ being the respective fermion number on each sector and where the last equality in \eqref{eq:WjsWjsX} holds by Lemma \ref{lem:Xcommute}. Again, we may extend the definition to \textbf{symmetrized vectors} by imposing that $ W_M(s_M) \ldots W_1(s_1) $ shall commute with the symmetrization operator $ (S_- \otimes \id) $.\\
\label{def:WjsWjsX}
\end{definition}

Again, it is easy to see that this definition makes sense: By Lemma \ref{lem:WjsWjs}, we have $ W_M(s_M) \ldots W_1(s_1) \Psi_m \in \FB $. And since $ X \in \cX $ maps $ \FB \to \FB_{\ex} $, indeed
\begin{equation}
	X W_M(s_M) \ldots W_1(s_1) \Psi_m \in \FB_{\ex} \;,
\end{equation}
so the right-hand sides of \eqref{eq:WjsWjsX} are well-defined.\\

\paragraph{Remarks.}

\begin{enumerate}
\setcounter{enumi}{\theremarks}
\item It seems natural to define \eqref{eq:WjsWjsX} for all operators $ X $ which commute with $ A_j^\dagger(s') $ in a sufficiently regular case and $ A_j(s') $. However, since we have only defined $ A_j(s'): \FB \to \FB_{\ex} $ and $ V_{j'j''}(s^* v): \FB \to \FB_{\ex} $, it is not clear how to interpret the commutator $ [A_j(s'), V_{j'j''}(s^* v)] $. So $ V_{j'j''}(s^* v) $ would then not be a valid $ X $-operator, although it commutes with $ A_j^\dagger(s') $ and $ A_j(s') $ for $ s, s', v \in \fh $.\\
If one succeeded to modify the definition of $ \FB, \FB_{\ex} $ such that commutators as $ [A_j(s'), V_{jj'}(s^* v)] $ are well-defined operators, then it seems reasonable to change the set of allowed $ X $ in Definition \ref{def:WjsWjsX} to all $ X $ with $ [A_j^\dagger(s'), X] = [A_j(s'), X] = 0 $.\\
\end{enumerate}
\setcounter{remarks}{\theenumi}

\subsubsection{Final definition of $ W(s) $}
\label{subsec:finaldef}

With Definitions \ref{def:WjsWjsAdagger} and \ref{def:WjsWjsX}, we may now provide the final domains for the product $ W_M(s_M) \ldots W_1(s_1) $: The \textbf{extended dressing domain} $ \cD_W $ is defined as
\begin{equation}
\begin{aligned}
	 &\cD_W := \\
	 &\span_{\eRen} \big\{ W_1(\varphi) A_j^{\dagger}(v) \Psi_m, \; X  W_1(\varphi) \Psi_m \; \big\vert \; \varphi \in \fh \cap \Sdot_1, \; v \in \Sdot_1, \; X \in \cX, \; \Psi_m \in \cC_{WS} \big\} \;,
\end{aligned}
\label{eq:cDW}
\end{equation}

with $ \Sdot_1 $ defined in \eqref{eq:Sdot1}, $ \cX $ defined in \eqref{eq:cX} and $ \cC_{WS} $ defined in \eqref{eq:cCWS}. Well-definedness of $ W(s) $ on $ \cD_W $ can be seen by combining \eqref{eq:WW1rewriting} with Definitions \ref{def:WjsWjsAdagger} and \ref{def:WjsWjsX}. By imposing that $ W_M(s_M) \ldots W_1(s_1) $ shall commute with $ (S_- \otimes \id) $, we extend $ W_M(s_M) \ldots W_1(s_1) $ to $ (S_- \otimes \id)[\cD_W] \cup \cD_W $.\\
The maximal domain of $ W(s) $ in Fock space is now given by the \textbf{large domain}
\begin{equation}
	\cD_\sF := (S_- \otimes \id) \big[ \cD_W \cap (L^2(\cQ_x) \otimes \sF_y) \big] \;.
\label{eq:cDsF2}
\end{equation}
The symmetrization operator $ (S_- \otimes \id) $ ensures that indeed $ \cD_\sF \subset \sF $. With this definition, it holds true that
\begin{lemma}
We have the inclusion
\begin{equation}
	\cD_{WS} \subset \cD_W \cap (L^2(\cQ_x) \otimes \sF_y) \;,
\label{eq:cDWSinclusion}
\end{equation}
and in particular, $ \cD_\sF $ is dense in $ \sF $.\\
\label{lem:cDsFdense}
\end{lemma}
\begin{proof}
Setting $ X = 1 $ and using Lemma \ref{lem:W1equivalence}, we see that $ W_{\sF, 1}(\varphi) \Psi_m = W_1(\varphi) \Psi_m \in \cD_{WS} $ with $ \varphi \in \fh \cap \Sdot_1 $ and $ \Psi_m \in \cC_{WS} $ is also an element of $ \cD_W $. Further, $ W_{\sF, 1}(\varphi) \Psi_m \in L^2(\cQ_x) \otimes \sF_y $, which yields the inclusion relation \eqref{eq:cDWSinclusion}.\\

Hence, the symmetrized version $ (S_- \otimes \id)[\cD_{WS}] $ is included in $ \cD_\sF $. And since the former is dense in $ \sF $ (Lemma \ref{lem:coherentdense}), also the latter is.\\
\end{proof}

In order to define the renormalized Hamiltonian $ \widetilde{H} = W(s)^{-1} H W(s) $, we also need to have a well-defined inverse $ W(s)^{-1} $.\\

\begin{lemma}
$ W(s) $ with $ s \in \Sdot_1 $ is invertible on $ \cD_\sF $.
\label{lem:Wsinvertible}
\end{lemma}

In the proof of this Lemma we employ another general statement about coherent states.\\

\begin{lemma}
For $ k \in \{1, \ldots, K\}, K \in \NNN $, consider $ \Psi'_{m, k} \in L^2(\cQ_x) \otimes \{\Omega_y\} $ and $ \varphi_k \in \fh \cap \Sdot_1 $. Further, choose any partition $ \{1, \ldots, K\} = \cK_{WA} \cup \cK_W $, as well as $ v_k \in \fh \cap \Sdot_1 $ and $ j_k \in \NNN $ for $ k \in \cK_{WA} $, and define
\begin{equation}
	\Psi_k := \begin{cases}
		W_{\sF, 1}(\varphi_k) A^\dagger_{j_k}(v_k) \Psi'_{m, k} \quad &\text{if } k \in \cK_{WA} \\
		W_{\sF, 1}(\varphi_k) \Psi'_{m, k} \quad &\text{if } k \in \cK_W	
	\end{cases} \;,
\label{eq:Psikform}
\end{equation}
such that $ \Psi_k \neq 0 $. Further, assume that  $ \varphi_k \neq \varphi_{k'} $ whenever $ k \neq k' $ both belong to either $ \cK_{WA} $ or $ \cK_W $. Then the set
\begin{equation}
	\big\{ \Psi_k \; \mid \; k \in \{1, \ldots, K\} \big\} \subset L^2(\cQ_x) \otimes \sF_y \;,
\label{eq:setWWA}
\end{equation}
is linearly independent.\\
\label{lem:coherentcombination}
\end{lemma}
Heuristically speaking, the proof relies on the argument, that there is a ``largest $ \varphi_k $'', for which the term $ \varphi_k^{\otimes N} $, occurring in a coherent state, eventually grows ``too large to be canceled by the $ K - 1 $ other terms'' as $ N \to \infty $. The proof itself is rather technical and can be found in Appendix \ref{sec:coherentcombination}.\\

\begin{proof}[Proof of Lemma \ref{lem:Wsinvertible}]
We need to show that $ W(s) $ is injective on $ \cD_\sF $. That is, there is no $ \Psi \in \cD_\sF, \Psi \neq 0 $ with $ W(s) \Psi = 0 $.

First of all, note that by definition of $ \cD_W $ \eqref{eq:cDW} and $ \cD_\sF $ \eqref{eq:cDsF2}, any $ \Psi \in \cD_\sF $ can be written as a finite sum $ \Psi = \sum_{k = 1}^{K} \Psi_k $ with
\begin{equation}
	\Psi_k = W_1(\varphi_k) A^\dagger_{j_k}(v_k) \Psi_{m, k} \qquad \text{or} \qquad \Psi_k = X_k W_1(\varphi_k) \Psi_{m, k} \;,
\end{equation}
where $ 0 \neq \Psi_{m, k} = \Psi_{mx, k} \otimes \Omega_y $. Since $ \Psi_k \in L^2 $, we have $ \varphi_k, v_k \in \fh \cap \Sdot_1 $ and without loss of generality we may assume that $ v_k \neq 0 $. And since $ X_k $ just multiplies by a function depending on the fermion momenta, $ X_k $ commutes with $ W_1(\varphi_k) $ so it can be absorbed into $ \Psi_{m, k} $. That is, we may re-define $ X_k \Psi_{m, k} $ to be the new $ \Psi_{m, k} $ and obtain that, without loss of generality, we could have chosen
\begin{equation}
	\Psi_k = W_1(\varphi_k) A^\dagger_{j_k}(v_k) \Psi_{m, k} \qquad \text{or} \qquad \Psi_k = W_1(\varphi_k) \Psi_{m, k} \;,
\end{equation}
with $ \Psi_{m, k} \in L^2 $. So we may define a disjoint union $ \{1, \ldots, K\} = \cK_{WA} \cup \cK_W $, such that
\begin{equation}
	\Psi = \sum_{k \in \cK_{WA}} W_1(\varphi_k) A^\dagger_{j_k}(v_k) \Psi_{m, k} +
	\sum_{k \in \cK_W} W_1(\varphi_k) \Psi_{m, k} \;.
\label{eq:Psiform}
\end{equation}

Now assume there was some $ \Psi \neq 0 $ with $ W(s) \Psi = 0 $. We define a ``compression operator'' $ B $ which ``compresses'' $ W(s) \Psi $ into $ L^2 $. For this purpose, let $ m_s $ and $ \beta_s $ be the UV/IR-scaling degrees of $ s $, respectively, and pick some real numbers $ m_b < - m_s - d/2 $ and $ \beta_b > - \beta_s - d/2 $. Choose a function $ b \in \Sdot_{1, >} $ (so $ b: \RRR^d \setminus \{ 0 \} \to \CCC $ is invertible) which has exact UV/IR-scaling degrees $ m_b $ and $ \beta_b $. With that choice, the ``compressed'' product function $ \bk \mapsto s(\bk) b(\bk) $ is in $ \fh \cap \Sdot_1 $, as is $  \bk \mapsto \varphi_k(\bk) b(\bk) $. Now, we define the compression operator $ B: \FB_{\ex} \to \FB_{\ex} $ as
\begin{equation}
	(B \Psi)(\bP, \bK) = \left( \prod_{\ell = 1}^N b(\bk_\ell) \right) \Psi(\bP, \bK) \;.
\end{equation}
It is easy to see that $ B $ maps $ \FB \to \FB $, $ \Sdot_\sF \to \Sdot_\sF $ and $ B^{-1}: \FB_{\ex} \to \FB_{\ex} $ exists with
\begin{equation}
	(B^{-1} \Psi)(\bP, \bK) = \left( \prod_{\ell = 1}^N \frac{1}{b(\bk_\ell)} \right) \Psi(\bP, \bK) \;.
\end{equation}
So $ W(s) \Psi = 0 $, if and only if $ B W(s) \Psi = 0 $.\\

Further, for $ \Psi_k = W_1(\varphi_k) \Psi_{m, k} $, a momentum space calculation renders the following identity:
\begin{equation}
\begin{aligned}
	&(B W(s) W_1(\varphi_k) \Psi_{m, k})(\bP, \bK)\\
	\overset{\eqref{eq:WW1rewriting}}{=} &(B e^{\Im \langle s, \varphi_k \rangle} W_M(s) \ldots W_2(s) W_1(s + \varphi_k) \Psi_{m, k})(\bP, \bK)\\
	\overset{\eqref{eq:WWWPsimomentum}}{=} &\frac{1}{\sqrt{N!}} e^{\Im \langle s, \varphi_k \rangle - \frac{(M-1)\Vert s \Vert^2 + \Vert s + \varphi_k \Vert^2}{2}} \sum_{\sigma} \left( \prod_{\ell = 1}^N b(\bk_{\ell}) s_{k, \sigma(\ell)}(\bk_{\ell}) \right) \Psi_{mx, k} \left( \bP' \right)\\
	= &\left( e^{\Im \langle s, \varphi_k \rangle - \frac{(M-1)\Vert s \Vert^2 + \Vert s + \varphi_k \Vert^2}{2}} e^{A^\dagger(b s) + A_1^\dagger(b \varphi_k)} \Psi_{m, k} \right) (\bP, \bK) \;,
\end{aligned}
\label{eq:BWW}
\end{equation}
with $ s_{k, 1} := s + \varphi_k $, as well as $ s_{k, 2} = \ldots = s_{k, M} := s $. As above, the sum is running over all maps $ \sigma: \{1, \ldots, N \} \to \{1, \ldots, M \} $ and $ \bP' = \bP + \sum_{\ell = 1}^N e_{\sigma(\ell)} \bk_{\ell} $. On the other hand, for unitary Fock space operators $ W_\sF, W_{\sF, j} $, the Weyl relations yield
\begin{equation}
\begin{aligned}
	&W_\sF(b s) W_{\sF, 1}(b \varphi_k) \Psi_{m, k}\\
	= &e^{\Im \langle b s, b \varphi_k \rangle} W_{\sF, M}(b s) \ldots W_{\sF, 2}(b s) W_{\sF, 1}(b (s + \varphi_k)) \Psi_{m, k}\\
	\overset{\eqref{eq:WjWjprime}}{=} &e^{\Im \langle b s, b \varphi_k \rangle - \frac{(M-1)\Vert b s \Vert^2 + \Vert b (s + \varphi_k) \Vert^2}{2}} e^{-\sum_{j > j'} V_{jj'}(s_{k, j}^* b^* b s_{k, j'})} e^{A^\dagger(b s) + A_1^\dagger(b \varphi_k) } \Psi_{m, k} \;.\\
\end{aligned}
\label{eq:WWb}
\end{equation}
Since $ b s_{k, j} \in L^2 $, the operators $ V_{jj'}(s_{k, j}^* b^* b s_{k, j'}) $ amount to a convolution with an $ L^1 $-function, which, after a Fourier transformation, is equivalent to a multiplication by a bounded function. So the $ V_{jj'} $-operators are all bounded and likewise, the exponential $ e^{-\sum_{j > j'} V_{jj'}(s_{k, j}^* b^* b s_{k, j'})} $ is bounded. Further, this bounded exponential commutes with $ e^{A^\dagger(b s) + A_1^\dagger(b \varphi_k)} $ by a similar fiber decomposition argument as in the proof of Lemma \ref{lem:Xcommute}. So comparing \eqref{eq:BWW} with \eqref{eq:WWb}, we obtain
\begin{equation}
	B W(s) W_1(\varphi_k) \Psi_{m, k} = \fc_k \cdot W_\sF (b s) W_{\sF, 1}(b \varphi_k) e^{-\sum_{j > j'} V_{jj'}(s_{k, j}^* b^* b s_{k, j'})} \Psi_{m, k} \;,
\label{eq:BWWcomparison1}
\end{equation}
for some $ \fc_k \in \eRen, \fc_k \neq 0 $. An analogous momentum space calculation yields
\begin{equation}
	B W(s) W_1(\varphi_k) A_{j_k}^\dagger(v_k) \Psi_{m, k} = \fc_k \cdot W_\sF (b s) W_{\sF, 1}(b \varphi_k) A_{j_k}^\dagger(b v_k) e^{-\sum_{j > j'} V_{jj'}(s_{k, j}^* b^* b s_{k, j'})} \Psi_{m, k} \;.
\label{eq:BWWcomparison2}
\end{equation}
Thus, we have the following chain of implications:
\begin{equation}
\begin{aligned}
	&&0 = &W(s) \Psi\\
	\Leftrightarrow \qquad && 0 = &B W(s) \Psi\\
	\overset{\eqref{eq:Psiform}}{\Leftrightarrow} \qquad && 0 = &\sum_{k \in \cK_{WA}} B W(s) W_1(\varphi_k) A^\dagger_{j_k}(v_k) \Psi_{m, k} +
	\sum_{k \in \cK_W} B W(s) W_1(\varphi_k) \Psi_{m, k}\\
	\Rightarrow \qquad && 0 = &\sum_{k \in \cK_{WA}} \fc_k W_\sF (b s) W_{\sF, 1}(b \varphi_k) A^\dagger_{j_k}(b v_k)  e^{-\sum_{j > j'} V_{jj'}(s_{k, j}^* b^* b s_{k, j'})} \Psi_{m, k}\\
	&&&+ \sum_{k \in \cK_W} \fc_k W_\sF (b s) W_{\sF, 1}(b \varphi_k) e^{-\sum_{j > j'} V_{jj'}(s_{k, j}^* b^* b s_{k, j'})} \Psi_{m, k}\\
	\Leftrightarrow \qquad && 0 = &\sum_{k \in \cK_{WA}} \fc_k W_{\sF, 1}(b \varphi_k) A^\dagger_{j_k}(b v_k) e^{-\sum_{j > j'} V_{jj'}(s_{k, j}^* b^* b s_{k, j'})} \Psi_{m, k}\\
	&&&+ \sum_{k \in \cK_W} \fc_k W_{\sF, 1}(b \varphi_k) e^{-\sum_{j > j'} V_{jj'}(s_{k, j}^* b^* b s_{k, j'})} \Psi_{m, k} \;,\\
\end{aligned}
\label{eq:zerocondition}
\end{equation}
where we exploited the unitarity of $ W_\sF (b s) $ in the last step.\\

We may now partition the indices $ k $ into several equivalence classes according to the relation
\begin{equation}
	k \sim k' \quad :\Leftrightarrow \quad \fc_k = c \fc_{k'} \text{ for some } c \in \CCC \;.
\end{equation}
The last equation of \eqref{eq:zerocondition} now implies that for each equivalence class $ [k'] $, the following sum with $ c_k := \frac{\fc_k}{\fc_k'} $ must vanish:
\begin{equation}
\begin{aligned}
	0 = &\sum_{k \in [k'] \cap \cK_{WA}} c_k W_{\sF, 1}(b \varphi_k) A^\dagger_{j_k}(b v_k) e^{-\sum_{j > j'} V_{jj'}(s_{k, j}^* b^* b s_{k, j'})} \Psi_{m, k}\\
	&+ \sum_{k \in [k'] \cap \cK_W} c_k W_{\sF, 1}(b \varphi_k) e^{-\sum_{j > j'} V_{jj'}(s_{k, j}^* b^* b s_{k, j'})} \Psi_{m, k} \;.
\end{aligned}
\end{equation}
This is exactly a linear combinations of vectors of the form \eqref{eq:Psikform} in Lemma \ref{lem:coherentcombination} with $ b \varphi_k, b v_k \in \fh \cap \Sdot_1 $ where all $ b \varphi_k $ are distinct, and with $ \Psi'_{m, k} := e^{-\sum_{j > j'} V_{jj'}(s_{k, j}^* b^* b s_{k, j'})} \Psi_{m, k} $. Since $ \fc_k \neq 0 $, we also have $ c_k \neq 0 $. Only the premise $ \Psi'_{m, k} \neq 0 $ is missing in order for Lemma \ref{lem:coherentcombination} to apply. However, if this premise was true, then Lemma \ref{lem:coherentcombination} would imply that $ c_k = 0 $ for all $ k $, which we rtuled out above. Thus, we conclude that $ \Psi'_{m, k} = 0 $ for at least one $ k $, and by repeatedly excluding this $ k $ from the linear combination and applying the argument to exclude further $ k $, we arrive at
\begin{equation}
	\Psi'_{m, k} = e^{-\sum_{j > j'} V_{jj'}(s_{k, j}^* b^* b s_{k, j'})} \Psi_{m, k} = 0 \qquad \forall k \in \{1, \ldots, K\} \;.
\end{equation}
By boundedness of the operator $ V_{jj'}(s_{k, j}^* b^* b s_{k, j'}) $, we can now conclude that the exponential is semibounded from below, that is,
\begin{equation}
	\Big\Vert e^{-\sum_{j > j'} V_{jj'}(s_{k, j}^* b^* b s_{k, j'})} \Psi' \Big\Vert
	\ge c' \Vert \Psi' \Vert \quad \forall \Psi' \in L^2(\cQ_x) \otimes \sF_y \qquad \text{for some } c' > 0 \;.
\end{equation}
So $ e^{-\sum_{j > j'} V_{jj'}(s_{k, j}^* b^* b s_{k, j'})} \Psi_{m, k} = 0 $ implies $ \Psi_{m, k} = 0 $ for all $ k $, which contradicts $ \Psi = 0 $ and concludes the proof.\\
\end{proof}

\paragraph{Remarks.}

\begin{enumerate}
\setcounter{enumi}{\theremarks}
\item In essence, we just transferred the commutation relations \eqref{eq:WAdaggers} for creation operators $ A^\dagger_j $ from Lemma \ref{lem:commutationW} in a certain sense to extended dressing operators $ W_j(s) $. This was done by imposing definitions such that these commutation relations still hold true. What about the commutation relations \eqref{eq:WAs} for annihilation operators $ A_j $?\\
In fact, these relations cannot be imposed by definition, but one may show that they are an immediate consequence of Definition \ref{def:WjsWjs}. This is proved in the following lemma:
\end{enumerate}
\setcounter{remarks}{\theenumi}
\begin{lemma}
Let $ s, v \in \Sdot_1 $ and $ \Psi_m \in \cC_{WS} $. Then, we have the commutation relations
\begin{equation}
\begin{aligned}
	W_{j'}(s) A_j(v) \Psi_m &= \begin{cases} \big( A_j(v) - \langle v,s \rangle \big) W_{j'}(s) \Psi_m \quad &\text{if } j = j' \\ \big( A_j(v) - V_{jj'}(v^* s) \big) W_{j'}(s) \Psi_m \quad &\text{if } j \neq j' \end{cases} \;,\\
	W(s) A_j(v) \Psi_m &= \big( A_j(v) - \langle v,s \rangle - V_{j \bullet}(v^* s) \big) W(s) \Psi_m \;,\\
	W_{j'}(s) A(v) \Psi_m &= \big( A(v) - \langle v,s \rangle - V_{\bullet j'}(v^* s) \big) W_{j'}(s) \Psi_m \;,\\
	W(s) A(v) \Psi_m &= \big( A(v) - M \langle v,s \rangle - V(v^* s) \big) W(s) \Psi_m \;.\\
\end{aligned}
\label{eq:WAsextended}
\end{equation}
\label{lem:WAextended}
\end{lemma}
\begin{proof}
First, note that $ A_j(v) \Psi_m = 0 $. The first line in \eqref{eq:WAsextended} then follows by momentum space definitions \eqref{eq:ap} and \eqref{eq:WWWsPsimomentum}
\begin{equation}
\begin{aligned}
	&\big( A_j(v) W_{j'}(s) \Psi_m \big) (\bP, \bK)\\
	= &\frac{e^{-\frac{\Vert s \Vert^2}{2}}}{\sqrt{N!}} \int v(\tilde{\bk})^* s(\tilde{\bk}) \left( \prod_{\ell = 1}^N s(\bk_{\ell}) \right) \Psi_{mx} (\bP' + (e_{j'} - e_j) \tilde{\bk}) \; d \tilde{\bk}\\
	= &\begin{cases} \langle v,s \rangle \Psi_m(\bP, \bK) \quad &\text{if } j = j' \\ \big(V_{jj'}(v^* s)\Psi_m \big) (\bP, \bK) \quad &\text{if } j \neq j' \end{cases} \;,
\end{aligned}
\end{equation}
with $ \bP' = \bP + e_{j'} \sum_{\ell = 1}^N \bk_{\ell} $.\\
 
The second line in \eqref{eq:WAsextended} is established similarly. We use again \eqref{eq:ap} and \eqref{eq:WWWsPsimomentum}, yielding:
\begin{equation}
\begin{aligned}
	&\big( A_j(v) W(s) \Psi_m \big) (\bP, \bK)\\
	= &\frac{e^{-\frac{\Vert s \Vert^2}{2}}}{\sqrt{N!}} \sum_{\tilde\sigma} \int v(\tilde{\bk})^* s(\tilde{\bk}) \left( \prod_{\ell = 1}^N s(\bk_{\ell}) \right) \Psi_{mx} (\bP' + (e_{\tilde\sigma(N+1)} - e_j) \tilde{\bk}) \; d \tilde{\bk} \;,\\
\end{aligned}
\end{equation}
where the sum runs over all $ \tilde\sigma: \{1, \ldots, N+1\} \to \{1, \ldots, M\} $ and we have set $ \bP' = \bP + \sum_{\ell = 1}^N e_{\tilde\sigma(\ell)} \bk_{\ell} $, as well as $ \bk_{N+1} = \tilde\bk $. We can split this sum into a sum over $ (\sigma, j) $ with $ \sigma: \{1, \ldots, N\} \to \{1, \ldots, M\}, \sigma(\ell) = \tilde\sigma(\ell) $ and $ j' \in \{1, \ldots, M\}, j' = \tilde\sigma(N+1) $:
\begin{equation}
\begin{aligned}
	&\big( A_j(v) W(s) \Psi_m \big) (\bP, \bK)\\
	= &\frac{e^{-\frac{\Vert s \Vert^2}{2}}}{\sqrt{N!}} \sum_{j'} \sum_{\sigma} \int v(\tilde{\bk})^* s(\tilde{\bk}) \left( \prod_{\ell = 1}^N s(\bk_{\ell}) \right) \Psi_{mx} (\bP' + (e_{j'} - e_j) \tilde{\bk}) \; d \tilde{\bk} \;.\\
\end{aligned}
\end{equation}
Now, the term with $ j' = j $ renders the contribution $ \langle s, v \rangle \Psi_m $ and all other $ M - 1 $ terms add up to $ \sum_{j': j \neq j'} V_{jj'}(v^* s) \Psi_m = V_{j \bullet}(v^* s) \Psi_m $, which is exactly the desired contribution.\\

Lines three and four of \eqref{eq:WAsextended} just follow by summing over $ j \in \{1, \ldots, M\} $ in the first two lines.\\
\end{proof}

\section{Pulling Back the Hamiltonian}	
\label{sec:pullback}

This section concerns taking a formal Hamiltonian
\begin{equation*}
	H = H_{0, y} + A^{\dagger}(v) + A(v) - E_\infty \;,
\end{equation*}
and pulling it back under the dressing transformation $ W(s) $, i.e., we compute
\begin{equation*}
	\widetilde{H}: \sF_{\ex} \supset \cD_{WS} \to \FB_{\ex} \qquad \text{with} \qquad W(s) \widetilde{H} = H W(s) \;.
\end{equation*}

The computation is split into two steps. In Section \ref{subsec:AE}, we compute the pullback of $ (A(v) - E_\infty) $. Pulling back only $ A(v) $ will result in divergences which are canceled by $ E_\infty $.\\
The pullback of $ (H_{0, y} + A^\dagger(v)) $ is then computed in Section \ref{subsec:H0A}. Combining $ H_{0, y} $ and $ A^\dagger(v) $ yields a particularly easy result.\\

Our main theorem is the following.

\begin{theorem}
Let $ s = - \frac{v}{\omega} $ with $ s, v, \omega \in \Sdot_1 $. Then the pullback of the self-energy renormalized Hamiltonian
\begin{equation}
	\widetilde{H} := H_{0, y} + V(v^* s) \qquad \text{satisfies} \qquad W(s) \widetilde{H} = \big( H_{0, y} + A^{\dagger}(v) + A(v) - E_\infty \big) W(s) \;,
\label{eq:pullback}
\end{equation}
which holds as a strong operator identity on $ \cD_{WS} $ (defined in \eqref{eq:cDWS}), as well as on $ (S_- \otimes \id)[\cD_{WS}] $.\\
\label{thm:pullback}
\end{theorem}

Note that, the potential interaction $ V $ defined in \eqref{eq:V1} via \eqref{eq:Vjj} acts as
\begin{equation*}
	(V \Psi)(\bP,\bK) = \sum_{j \neq j'} \int v(\tilde{\bk})^* s(\tilde{\bk}) \Psi(\bP + (e_j - e_{j'}) \tilde{\bk},\bK) \; d \tilde{\bk} \;.
\end{equation*}

\subsection{Pulling Back $ A - E_\infty $}
\label{subsec:AE}

We recall that by Proposition \ref{prop:operatorextension}, one can define $ E_\infty: \FB \to \FB_{\ex} $ with
\begin{equation}
	(E_\infty \Psi)(\bP,\bK) = M \langle v,s \rangle \; \Psi(\bP,\bK) = \sum_{j=1}^M \int - \frac{v(\bk)^* v(\bk)}{\omega(\bk)} \; d \bk \; \Psi(\bP,\bK) \;,
\label{eq:E}
\end{equation}
even if $ \langle v,s \rangle \notin \CCC $, but $ \langle v,s \rangle \in \Ren_1 $.\\

\begin{lemma}
Let $ \Psi_m \in \cC_{WS} $ and $ s = - \frac{v}{\omega} $ with $ s, v, \omega \in \Sdot_1 $. Then for $ \varphi \in \Sdot_1 \cap \fh $,
\begin{equation}
	\big( A(v) - E_\infty \big) W(s) W_{\sF, 1}(\varphi) \Psi_m = W(s) \big( \res_1(\varphi) + V(v^* s) \big) W_{\sF, 1}(\varphi) \Psi_m \quad \in \FB_{\ex} \;,
\label{eq:AEinftypullback}
\end{equation}
where the \textbf{residual operator}
\begin{equation}
	\res_1(\varphi) = \underbrace{\langle v,\varphi \rangle}_{\in \Ren_1} + V_{\bullet 1}(v^* \varphi) \;,
\label{eq:res1sn}
\end{equation}
is by Lemma \ref{lem:CCR} a well-defined mapping $ \FB \to \FB_{\ex} $.
\label{lem:AE}
\end{lemma}

The proof of Lemma \ref{lem:AE} is given in Appendix \ref{sec:AE}.\\

\subsection{Pulling Back $ H_{0, y} + A^{\dagger} $}
\label{subsec:H0A}

\begin{lemma}
Let $ \Psi_m \in \cC_{WS} $ and $ s = - \frac{v}{\omega} $ with $ s, v, \omega \in \Sdot_1 $. Then for $ \varphi \in \Sdot_1 \cap \fh $,
\begin{equation}
	\big( H_{0, y} + A^{\dagger}(v) \big) W(s) W_{\sF, 1}(\varphi) \Psi_m = W(s) \big( H_{0, y} - \res_1(\varphi) \big) W_{\sF, 1}(\varphi) \Psi_m \quad \in \FB_{\ex} \;,
\label{eq:H0Adaggerundressed}
\end{equation}
with the same residual operator $ \res_1 = \langle v,\varphi \rangle + V_{\bullet 1}(v^* \varphi)$ as in Lemma \ref{lem:AE}.
\label{lem:H0A}
\end{lemma}

As for Lemma \ref{lem:AE}, the proof of Lemma \ref{lem:H0A} is rather technical. It can be found in Appendix \ref{sec:H0A}. With both lemmas at hand, Theorem \ref{thm:pullback} can directly be proved.\\

\begin{proof}[Proof of Theorem \ref{thm:pullback}]
This is a simple consequence of Lemmas \ref{lem:AE} and \ref{lem:H0A}. We put together \eqref{eq:AEinftypullback} and \eqref{eq:H0Adaggerundressed} which yields
\begin{equation}
\begin{aligned}
	W(s) \widetilde{H} = \big( H_{0, y} + A^\dagger(v) + A(v) - E_\infty \big) W(s) = &W(s) \big( H_{0, y} + V(v^* s) + \res_1(\varphi) - \res_1(\varphi) \big)\\
	= &W(s) \big( H_{0, y} + V(v^* s) \big) \;,\\
\end{aligned}
\end{equation}
as a strong operator identity on all $ \Psi = W_{\sF, 1}(\varphi) \Psi_m, \varphi \in \Sdot_1 \cap \fh $. And these $ \Psi $ span $ \cD_{WS} $.\\

Since we imposed in Definitions \ref{def:WjsWjs}, \ref{def:WjsWjsAdagger} and \ref{def:WjsWjsX} that symmetrization $ (S_- \otimes \id) $ shall commute with $ W(s) $, the strong operator identity is also valid on $ (S_- \otimes \id)[\cD_{WS}] $.\\

\end{proof}


\section{Self-Adjointness}
\label{sec:selfadjoint}

In this section, we prove that in certain cases, $ \widetilde{H} $ can indeed be defined as a self-adjoint operator $ \widetilde{H}: \sF \supset \dom(\widetilde{H}) \to \cD_\sF $. So far we have by Theorem \ref{thm:pullback} that $ \widetilde{H}: (S_- \otimes \id)[\cD_{WS}] \to \FB_{\ex} $ is well-defined.\\
In order for the image of $ \widetilde{H} $ to be in $ \cD_\sF $, we need to restrict the domain of $ \widetilde{H} $ even further to some subspace $ \widetilde{\cD}_\sF \subset \cD_\sF $, defined in \eqref{eq:cDsFtilde}, and prove well-definedness of $ \widetilde{H}: \widetilde{\cD}_\sF \to \cD_\sF $ (Lemma \ref{lem:denselydefined}). The existence of a self-adjoint extension on some $ \dom(\widetilde{H}) \supset \widetilde{\cD}_\sF $ is then a simple consequence (Corollary \ref{cor:selfadjointext}).\\

\subsection{Existence of Self-adjoint Extensions}
\label{subsec:selfadjointext}

First, we verify that $ \widetilde{H} = H_{0, y} + V $ is well-defined and symmetric on a dense domain in $ \sF $.

\begin{definition}
Let $ \cQ_{\col, x} $ be the set of \textbf{collision configurations}, i.e., all fermion position space configurations
\begin{equation}
	\cQ_{\col, x} := \big\{ \bX \in \cQ_x \; \big\vert \; \exists j \neq j': \; \bx_j = \bx_{j'} \big\} \;.
\label{eq:Qcolx}
\end{equation}
Denote by
\begin{equation}
	\widetilde{\cC}_{WS} := \big\{ \Psi_m = \Psi_{mx} \otimes \Omega_y \; \big\vert \; \cF^{-1}(\Psi_{mx}) \in C_c^\infty(\cQ_x \setminus \cQ_{\col, x}) \big\} \;,
\label{eq:cCW0x}
\end{equation}
the ``cyclic set'' of functions whose support avoids the collision configurations (where $ \cF^{-1} $ is the inverse Fourier transform). We define the \textbf{small domain}, on which $ \widetilde{H} $ is initially defined as a Fock space operator as
\begin{equation}
	\widetilde{\cD}_\sF := (S_- \otimes \id) \big[ \span \big\{ W_{\sF, 1}(\varphi) \Psi_m \; \big\vert \; \varphi \in C_c^{\infty}(\RRR^d), \; \Psi_m \in \widetilde{\cC}_{WS} \big\} \big] \;,
\label{eq:cDsFtilde}
\end{equation}
see Figure \ref{fig:Qcol}. It is easy to see that $ \widetilde{\cD}_\sF \subseteq \cD_\sF $ and $ \widetilde{\cD}_\sF \subseteq (S_- \otimes \id)[\cD_{WS}] $.\\
\end{definition}

\begin{figure}[hbt]
	\centering
	\begin{tikzpicture}

\filldraw[dashed, fill = gray!10!white] (-6,-1) -- (-1.2,-1) -- (-0.84,0.8) -- (-4.2,0.8) -- cycle;
\filldraw[red!50!blue, fill opacity = 0.2] (-4.5,-0.5) .. controls (-4,-0.8) and (-2.5,-0.8) .. (-2,-0.5) .. controls (-1.5,-0.2) and (-1.5,0.2) .. (-2,0.6) .. controls (-2.5,1) and (-5,-0.2) .. cycle ;
\draw[red!50!blue] (-3,0.4) -- ++(-0.9,0.5) node[anchor = south] {$ \text{supp} \cF^{-1}(\Psi_{mx}) $};
\draw[thick,->] (-5,0) -- (-1,0) node[anchor = south west] {$ \bx_1 $};
\draw[thick,->] (-3.6,-1) -- (-2.32,0.8) node[anchor = south west] {$ \bx_2 $};
\node at (-3,0.2) {0};

\node[red!50!blue] at (-3.6,-1.5) {$ \cF^{-1}(\Psi_{mx}) $ appearing in $ \cC_{WS} $};

\filldraw[dashed, fill = gray!10!white] (7,-1) -- (2.2,-1) -- (1.84,0.8) -- (5.2,0.8) -- cycle;
\filldraw[red!50!blue, fill opacity = 0.2] (2.8,0.4) ellipse (0.4 and 0.12);
\filldraw[red!50!blue, fill opacity = 0.2] (5.35,-0.5) ellipse (0.6 and 0.18);
\draw[red!50!blue] (2.8,0.4) -- ++(-0.9,0.5) node[anchor = south] {$ \text{supp} \cF^{-1}(\Psi_{mx}) $};
\draw[thick,->] (2,0) -- (6,0) node[anchor = south west] {$ \bx_1 $};
\draw[thick,->] (4.6,-1) -- (3.32,0.8) node[anchor = south west] {$ \bx_2 $};
\draw[line width = 2,red] (2.2,-1) -- (5.2,0.78) node[anchor = south west] {$ \cQ_{\col, x} $};
\node at (3.7,0.2) {0};

\node[red!50!blue] at (4.6,-1.5) {$ \cF^{-1}(\Psi_{mx}) $ appearing in $ \widetilde{\cC}_{WS} $};

\end{tikzpicture}
	\caption{Within $ \cC_{WS} $, the fermionic wave functions must be in $ \cS $, allowing for any support (including compact ones). Within $ \widetilde{\cC}_{WS} $, only $ C_c^\infty $-functions are allowed with support avoiding the collision configurations $ \cQ_{\col, x} $.}
	\label{fig:Qcol}
\end{figure}

In the following Lemma, we will use that if $ \omega, v \in \Sdot_1 $ satisfy \eqref{eq:scalingfunctions} and \eqref{eq:scalingdispersion} (so they scale polynomially), then the potential function
\begin{equation}
	\hat{V} := v^* s = -\tfrac{v^* v}{\omega} \;,
\label{eq:Vhat}
\end{equation}
also scales polynomially with
\begin{equation}
	m_V = 2 m_v - m_\omega \;, \qquad \beta_V = 2 \beta_v - \beta_\omega \;.
\label{eq:Vscaling}
\end{equation}
Further, if $ \beta_V > -d $, then the inverse Fourier transform $ V = \cF^{-1}(\hat{V}) \in S'(\RRR^d) $ also exists, so we can make statements about the singular support of $ V $.

\begin{lemma}[$ \widetilde{H} $ is densely defined and symmetric]
The set $ \widetilde{\cD}_\sF $ is dense in $ \sF $.\\
Assume, that $ s, v, \omega \in \Sdot_1 $ with $ \omega, v $ satisfying \eqref{eq:scalingfunctions} and \eqref{eq:scalingdispersion}, as well as $ \beta_V > -d $. If now
\begin{equation}
	\mathrm{sing supp}(V) \subseteq \{0\} \;,
\label{eq:singsupp}
\end{equation}
then $ \widetilde{H} $ maps $ \widetilde{\cD}_\sF \to \cD_\sF $ and is thus densely defined. If in addition the symmetry condition \eqref{eq:symmetry} holds, then $ \widetilde{H} $ is symmetric.
\label{lem:denselydefined}
\end{lemma}
\begin{proof}
Density of $ \widetilde{\cD}_\sF $ in $ \sF $ is established as density of $ (S_- \otimes \id) \cD_{WS} $ in Lemma \ref{lem:coherentdense} (proof in Appendix \ref{sec:coherentdense}). We recall that by the last line of \eqref{eq:D1W1split},
\begin{equation*}
	D_1 W_{\sF, 1}(\varphi) \Psi_m = \Psi_{mx} \otimes W_y(\varphi) \Omega_y \;.
\end{equation*}
In the proof of Lemma \ref{lem:coherentdense}, we argue that $ (S_- \otimes \id) \cD_{WS} $ is dense in $ \sF $, since $ \Psi_{mx} \otimes W_y(\varphi) \Omega_y $ approximates any $ \Psi \in L^2(\cQ_x) \otimes \sF_y $ arbitrarily well. The transition from $ (S_- \otimes \id) \cD_{WS} $ to $ \widetilde{\cD}_\sF $ is achieved by a restriction to $ \Psi_{mx} \in \cF[C_c^\infty(\cQ_x \setminus \cQ_{\col, x} )] $. The set $ \cQ_{\col, x} $ is a union of hyperplanes on each sector in the fermionic configuration space, so $ C_c^\infty(\cQ_x \setminus \cQ_{\col, x} ) $ is dense in $ L^2(\cQ_x) $. The Fourier transform $ \cF $ is an isometry, so the allowed set for $ \Psi_{mx} $ is dense in $ L^2(\cQ_x) $. Thus, we can still approximate any $ \Psi \in L^2(\cQ_x) \otimes \sF_y $ arbitrarily well by $ \Psi_{mx} \otimes W_y(\varphi) \Omega_y $ and by the same arguments as in the proof of Lemma \ref{lem:coherentdense}, $ \widetilde{\cD}_\sF $ is dense in $ \sF $.\\

Now we verify that $ H_{0, y} $ maps $ \widetilde{\cD}_\sF \to \cD_\sF $. By linearity, it suffices to show well-definedness on all vectors of the form $ W_{\sF, 1}(\varphi) \Psi_m, \varphi \in C_c^{\infty}(\RRR^d) $. Denote by $ P_y^{(N)} $ the projection of $ \sF_y $ to the $ N $-boson sector $ \sF_y^{(N)} $, so $ \sum_{N \in \NNN_0} P_y^{(N)} = 1 $. The $ L^2 $-norm squares of $ W_{\sF, 1}(\varphi) \Psi_m $ are Poisson-distributed over $ N $, i.e.,
\begin{equation}
	\big\Vert P_y^{(N)} W_{\sF, 1}(\varphi) \Psi_m \big\Vert^2 = e^{-\Vert \varphi \Vert^2} \frac{\Vert \varphi \Vert^{2N}}{N!} \;,
\end{equation}
so for any $ 0 < q < 1 $ they decay faster than $ q^N $ in $ N $-direction. Now, define
\begin{equation}
	\lambda := \max_{\bk \in \supp (\varphi)} |\omega(\bk)| \;.
\end{equation}
Then,
\begin{equation}
	\big\Vert P_y^{(N)} H_{0, y} W_{\sF, 1}(\varphi) \Psi_m \big\Vert^2 \le N^2 \lambda^2 \big\Vert P_y^{(N)} W_{\sF, 1}(\varphi) \Psi_m \big\Vert^2 \;,
\end{equation}
which still decays faster than $ q^N $ in $ N $-direction. Thus, $ \big\Vert H_{0, y} W_{\sF, 1}(\varphi) \Psi_m \big\Vert < \infty $ and we have that $ H_{0, y} W_{\sF, 1}(\varphi) \Psi_m \in L^2(\cQ_x) \otimes \sF_y $.\\

It remains to be shown that $ V = V(v^* s) $ is well-defined, which amounts to proving that
\begin{equation*}
	(V \Psi)(\bP,\bK) = \sum_{j \neq j'} \int \hat{V}(\tilde{\bk}) \Psi(\bP + (e_j - e_{j'}) \tilde{\bk},\bK) \; d \tilde{\bk} \;,
\end{equation*}
defines an $ L^2 $-function on $ \cQ $. Since $ \beta_V > -d $, we have that $ \hat{V} \in L^1_{\loc} \Rightarrow \hat{V} \in \cS' $, so we can take the Fourier transform as in \eqref{eq:Vjjpositionspace}
\begin{equation*}
	(V \Psi)(\bX,\bY) = \sum_{j \neq j'} V(\bx_j - \bx_{j'}) \Psi(\bX,\bY) \;,
\end{equation*}
with $ V(\bx) = \cF^{-1}(\hat{V})(\bx) $.

Now, for $ \Psi = W_{\sF, 1}(\varphi) \Psi_m $ we obtain the position space representation by Fourier-transforming \eqref{eq:W1s}:
\begin{equation}
	\Psi(\bX,\bY) = \frac{e^{-\frac{\Vert \varphi \Vert^2}{2}}}{\sqrt{N!}} \left( \prod_{\ell = 1}^N \check\varphi(\by_{\ell} - \bx_1) \right) \Psi_{mx}(\bX) \;,
\end{equation}
where $ \check \varphi = \cF^{-1}(\varphi) $ is a Schwartz function, as $ \varphi \in C_c^\infty $ is Schwartz. So as $ \Psi_{mx}(\bX) $ is a smooth function with compact support apart from collision configurations in $ \cQ_x $, also $ \Psi $ is smooth, and it is zero at fermion collision configurations in $ \cQ $. Since the singular support of $ V(\bx) $ is at most $ \{0\} $, the multiplication function $ \sum_{j \neq j'} V(\bx_j - \bx_{j'}) $ is smooth on $ \supp(\Psi_{mx}) $ (which excludes collision configurations). And as $ \supp(\Psi_{mx}) $ is compact, there is some $ C_\Psi \in \RRR $ with
\begin{equation}
	\max_{\bX \in \supp(\Psi_{mx})} \left\vert \sum_{j \neq j'} V(\bx_j - \bx_{j'}) \right\vert \le C_\Psi \;.
\end{equation}
Further, by compactness of support, a maximum occupied fermion sector $ \overline{M} $ exists. So
\begin{equation*}
	\Vert V \Psi \Vert^2 \le \overline{M}^4 C_\Psi^2 \Vert \Psi \Vert^2 < \infty \;,
\end{equation*}
for $ \Psi \in \widetilde{\cD}_\sF $. Thus, $ V \Psi \in L^2(\cQ_x) \otimes \sF_y $.\\

Symmetrization for fermions by $ (S_- \otimes \id) $ is preserved by $ H_{0, y} $ and $ V $. Hence, indeed $ \widetilde{H} \Psi \in \sF $. And by Theorem \ref{thm:pullback}, we have that $ \widetilde{H} \Psi \in (S_- \otimes \id)[\cD_W] $ (otherwise, we could not apply $ W(s) $ to it).\\

Symmetry of $ \widetilde{H} $ is an obvious consequence of the symmetry condition \eqref{eq:symmetry}. And since $ \widetilde{H} $ preserves symmetry, it maps $ \widetilde{\cD}_\sF \to \cD_\sF = (S_- \otimes \id)[\cD_W \cap (L^2(\cQ_x) \otimes \sF_y)] $ (compare \eqref{eq:cDsF2}).\\
\end{proof}

\begin{corollary}[Existence of a self-adjoint extension]
$ \widetilde{H}: \widetilde{\cD}_\sF \to \cD_\sF $ as in Lemma \ref{lem:denselydefined} allows for a self-adjoint extension.
\label{cor:selfadjointext}
\end{corollary}
\begin{proof}
This is a direct consequence of \cite[Thm. X.3]{ReedSimon2} (von Neumann's theorem): For a symmetric operator $ \widetilde{H} $ (called $ A $ within \cite{ReedSimon2}), this theorem asserts that there is a self-adjoint extension, provided that a conjugation operator $ C: \widetilde{\cD}_\sF \to \widetilde{\cD}_\sF $ can be found, such that
\begin{equation}
	C \widetilde{H} = \widetilde{H} C \;.
\label{eq:conjugation}
\end{equation}
As a conjugation, we choose $ (C\Psi)(\bK) = \Psi(-\bK)^* $, which amounts to complex conjugation in particle--position representation. By symmetry \eqref{eq:symmetry} and since $ \omega $ is real-valued, $ \hat{V}(\bk) = \hat{V}(-\bk)^* $, so $ VC = CV $. And analogously, $ C H_{0, y} = H_{0, y} C $. Thus, \eqref{eq:conjugation} holds, and we have at least one self-adjoint extension.
\end{proof}

\paragraph{Remarks.}

\begin{enumerate}
\setcounter{enumi}{\theremarks}
\item \textit{Essential self-adjointness}: We cannot expect $ \widetilde{H} $ to be generally self-adjoint on $ \widetilde{\cD}_\sF $: The $ N = 0 $-sector of $ \widetilde{\cD}_\sF $ \eqref{eq:cDsFtilde} is
\begin{equation}
	\bigoplus_{M = 0}^\infty \widetilde{\cD}_\sF \cap \sF^{(M, 0)} = \cF^{-1}[C_c^\infty(\cQ_x \setminus \cQ_{\col, x})] \;.
\end{equation}
The removed diagonal $ \cQ_{\col, x} $ may now allow for putting various boundary conditions that result in different self-adoint extensions. An easy example is the case $ \theta(\bp) = |\bp|^2, \omega = 0, v = 0 $ for $ d = 3 $ (so $ V = 0 $). In the position representation of the $ N = 2 $-sector,
\begin{equation}
	\sF[\widetilde{\cD}_\sF \cap \sF^{(2, 0)}] = C_c^\infty(\RRR^6 \setminus \{\bx_1 = \bx_2\}) \;.
\end{equation}
Now, $ \widetilde{H} = H = H_0 $ acts as $ (- \Delta_{\bx_1} - \Delta_{\bx_2}) $ on this sector, which, in center-of-mass coordinates ($ \bx = \bx_1 - \bx_2 $) reduces to $ (- 2 \Delta_{\bx}) $ acting on $ C_c^\infty(\RRR^3 \setminus \{0\}) $. The latter Hamiltonian is well-known to have many self-adjoint extensions \cite[Ch.~I.1]{AlbeverioGesztesy} that correspond to either a free particle or a particle in a $ \delta $-potential of arbitrary strength. Including all sectors $ N \in \NNN_0 $, we even obtain more boundaries, which allow for further choices of boundary conditions.\\

However, in case $ \omega \ge 0 $ there is a distinguished self-adjoint extension of $ \widetilde{H} $ as in Lemma \ref{lem:denselydefined}: We may decompose $ \sF $ and $ \widetilde{H} $ into fibers as
\begin{equation}
	\sF = \int_{\bX \in \cQ_x}^\oplus \sF_y \;, \qquad
	\widetilde{H} = \int_{\bX \in \cQ_x}^\oplus \widetilde{H}_{\bX} \;,
\label{eq:fiberdecomposition}
\end{equation}
where
\begin{equation}
	\widetilde{H}_{\bX} = d \Gamma(\omega) + \sum_{j \neq j'} V(\bx_j - \bx_{j'}) \qquad
	\text{for } \bX \in \cQ_x \setminus \cQ_{\col, x} \;.
\label{eq:fiberdecomposition2}
\end{equation}
Since the diagonal $ \cQ_{\col, x} $ is a null set, $ \widetilde{H}_{\bX} $ is defined for almost every $ \bX \in \cQ_x $, so the fiber integral \eqref{eq:fiberdecomposition} makes sense. Since $ \widetilde{H}_{\bX} = d \Gamma(\omega) $ up to an $ \bX $-dependent constant, for $ \omega \ge 0 $, each $ \widetilde{H}_{\bX} $ is positive and can be realized as a self-adjoint Friedrichs extension on $ \sF_y $. The fiber integral \eqref{eq:fiberdecomposition} then readily defines a self-adjoint operator $ \widetilde{H} $.\\

\end{enumerate}
\setcounter{remarks}{\theenumi}

\section{Further Applications of the ESS}	
\label{sec:furtherdress}

Employing an ESS construction, the cancellation of infinities within IBC renormalization mentioned in the introduction can be presented in a particularly convenient way: Formally, one would re-arrange divergent terms such that they cancel and consider only the steps after the cancellation as mathematically rigorous. Using the ESS framework, we can turn the formal manipulations before and including the ``cancellation of infinities'' into a rigorous statement.\\
Effectively, the IBC construction uses a non-unitary dressing operator $ W = W_{\IBC}^{-1} = (1 + H_0^{-1} A^{\dagger})^{-1} $ defined on a subspace of $ \sF $, which replaces the $ W(s) $ from the above sections. We prove that under rather general assumptions, $ W_{\IBC} $ can be defined as a bijection $ W_{\IBC}: \FB \to \FB $ (Proposition \ref{prop:IBC}) and can thus be used to construct a renormalized Fock space.\\
Another example for (non-unitary) dressing transformations comes from constructive QFT. These transformations depend on the model and can be quite involved. We will use the ESS to define a prototypical dressing operator $ T = e^{- \bLambda/2} e^{-H_0^{-1} A^{\dagger}}: \FB \to \FB $, resembling typical CQFT dressing operators \cite[Ch. 3]{Friedrichs1965}, \cite[§ 3]{Glimm1968}. The question, how to define concrete CQFT dressing operators and if they allow for a cutoff-free renormalization, is an interesting topic for future research.\\

\subsection{Rigorous Cancellation of Infinities within IBC Renormalization}	
\label{subsec:IBCinfinities}

Let us briefly sketch the IBC renormalization technique used in \cite{I03, I04, I06, I07, I10} and show how the ESS framework serves for making cancellations of infinities both explicit and rigorous. We assume that $ \theta, \omega, v \in \Sdot_1 $ with $ \theta(\bp), \omega(\bk) > 0 $. For any self-energy operator $ E: \FB \to \FB_{\ex} $ we can define $ H_{\IBC} = H_0 + A^{\dagger} + A - E $ as $ H_{\IBC}: \FB \to \FB_{\ex} $ by Proposition \ref{prop:operatorextension}. The following (usually formal) manipulation now becomes a rigorous statement about equality of operators $ \FB \to \FB_{\ex} $:
\begin{equation}
\begin{aligned}
	H_{\IBC} = &H_0 + A^{\dagger} + A &&-E\\
	=&H_0 + A^{\dagger} + A + A H_0^{-1} A^{\dagger} &&- A H_0^{-1} A^{\dagger}-E\\
	=&\underbrace{(1 + A H_0^{-1}) H_0^{1/2}}_{=:S^*} \underbrace{H_0^{1/2} (1+H_0^{-1} A^{\dagger})}_{=:S} &&\underbrace{- A H_0^{-1} A^{\dagger}-E}_{=: T}\\
	=&S^* S + T \;.\\
\end{aligned}
\label{eq:formalmanipulationIBC}
\end{equation}
That is because $ H_0 $ is just a multiplication on momentum configuration space by a strictly positive function, so by the same argument as in the proof of Proposition \ref{prop:operatorextension}, $ H_0^{-1}, H_0^{1/2}: \FB \to \FB $ are well-defined.\\
 
In \eqref{eq:formalmanipulationIBC}, $ S^* $ is just a formal expression and not the adjoint of $ S $. However, in sufficiently regular cases, $ S^* S $ can be densely defined on $ \sF $ as a closed operator, with $ S^* $ being the adjoint of $ S $, so it is self-adjoint by \cite[Thm. X.25]{ReedSimon2}. Now, if $ T: \sF \supset \cD(T) \to \sF $ is a densely defined Kato-perturbation of $ S^* S $, that is
\begin{equation}
	\Vert T \Psi \Vert \le a \Vert S^* S \Psi \Vert + b \Vert \Psi \Vert \qquad \forall \Psi \in \dom(S^*S) \;,
\end{equation}
with $ a < 1 $, then $ H_{\IBC} $ is self-adjoint on $ \dom(S^*S) $ by the Kato--Rellich theorem \cite[X.12]{ReedSimon2}. In \cite{I03, I04, I06, I07, I10} this argument was successfully employed in order to establish self-adjointness of IBC Hamiltonians $ S^* S + T $ on the domain
\begin{equation}
	\dom(S^*S) = \big\{ \Psi \in \sF \; \big\vert \; (1 + H_0^{-1} A^{\dagger}) \Psi \in \dom(H_0) \big\} \;.
\label{eq:IBCdomain}
\end{equation} 
The condition $ (1 + H_0^{-1} A^{\dagger}) \Psi \in \dom(H_0) $ is called \textbf{interior--} or \textbf{abstract boundary condition}.\\

We remark that earlier works on IBC \cite{I00, I00b, I02} consider models where the integral $ \int \frac{v^2}{\omega} $ and thus $ E $ is finite. In essence, these works consider the same formal Hamiltonian as we do in \eqref{eq:formalhamiltonian2} and explicitly diagonalize it by a Weyl transformation as in Theorem \ref{thm:pullback}, but for much less singular $ v $ and $ \omega $.\\
In the most singular case, $ \int \frac{v^2}{\omega^2} = \infty \Leftrightarrow \frac{v}{\omega} \notin L^2 $, for any $ \Psi \in \dom(S^*S), \Psi^{(0)} \neq 0 $, we formally have
\begin{equation}
	((1 + H_0^{-1} A^{\dagger}) \Psi)^{(1)}(\bk) = \Psi^{(1)}(\bk) + \frac{v(\bk)}{\omega(\bk)} \Psi^{(0)} \;,
\end{equation}
which is required to define a function in $ \dom(\omega) \subset L^2 $. However, for $ \frac{v}{\omega} \notin L^2 $, this would imply $ \Psi^{(1)} \notin L^2 $, so $ \Psi $ cannot be an element of Fock space. Similar problems appear if $ \Psi^{(0)} = 0 $ but $ \Psi \neq 0 $, when replacing $ \Psi^{(0)} $ with the lowest occupied sector of $ \Psi $. So for $ \frac{v}{\omega} \notin L^2 $, the only Fock space vector in $ \dom(S^*S) $ is 0, and we expect IBC renormalization as above to fail. Indeed, there even exists a no-go theorem \cite{I11}, excluding any interactions mediated by boundary conditions at a point source for $ \omega $ corresponding to a Dirac particle.\\
Getting beyond the threshold where $ \dom(S^*S) $ becomes $ \{0\} $ as vectors ``move out of Fock space'' was a central motivation for introducing the ESS framework. Indeed, with Proposition \ref{prop:IBC}, $ W_{\IBC} = (1 + H_0^{-1} A^{\dagger}) $ can be defined as a bijective operator on the ESS $ \FB $, so we are able to define a generalized IBC domain
\begin{equation}
	\big\{ \Psi \in \FB \; \big\vert \; (1 + H_0^{-1} A^{\dagger}) \Psi \in \dom(H_0) \big\} = W_{\IBC}^{-1}[\dom(H_0) \cap \Sdot_\sF] \;,
\label{eq:ESSIBCdomain}
\end{equation}
which is identified by $ W_{\IBC} $ with the dense subspace $ \dom(H_0) \cap \Sdot_\sF \subset \sF $.\\

There are also recent works on IBC, which use strategies different from the one above to obtain quantum dynamics. In \cite{I09}, a self-adjoint Hamiltonian is reconstructed from resolvents. By contrast, in \cite{I05} quantum dynamics are formulated in terms of a multi-time integral equation with boundary conditions and solutions are constructed by a fixed point argument.\\
Further, there is a recent series of works about direct renormalization on spin--boson type models \cite{Lonigro2021, Lonigro2022, Lonigro2023}, which recover the IBC domain \eqref{eq:IBCdomain}, but use an explicit computation of deficiency spaces to prove self-adjointness. These works cover cases where $ v \notin L^2 $ and even $ \int \frac{v^2}{\omega} = \infty $, but still $ \int \frac{v^2}{\omega^2} < \infty $. We hope that the ESS framework allows progressing also to $ \int \frac{v^2}{\omega^2} = \infty $ in models of this kind.\\

\subsection{Renormalized Fock Spaces for IBC}	
\label{subsec:renormfockspace}

In order for \eqref{eq:ESSIBCdomain} to be a valid characterization of the generalized IBC domain, we need $ W_{\IBC}^{-1} $ to be well-defined, which is true under fairly mild assumptions.

\begin{proposition}
Let $ \theta, \omega, v \in \Sdot_1 $ with $ \theta > 0 $, $ \omega(\bk) > 0 $ for $ \bk \neq 0 $, and $ \omega(\bk) \ge c |\bk|^\beta $ in some neighborhood of $ \bk = 0 $ for some $ c, \beta > 0 $. Define $ H_0 $ and $ A^\dagger = A^\dagger(v) $ as in Proposition \ref{prop:operatorextension} (via \eqref{eq:H0} and \eqref{eq:adaggerp}).\\
Then, the operator $ W_{\IBC} = (1 + H_0^{-1} A^\dagger): \FB \to \FB $ is bijective.
\label{prop:IBC}
\end{proposition}
We remark that the technical condition $ \omega(\bk) \ge c |\bk|^\beta $ near $ \bk = 0 $ is only necessary to make $ H_0^{-1} $ well-defined. This would also be true for the stronger condition $ \omega > 0 $, which, however, excludes the physically interesting cases $ \omega(\bk) = |\bk| $ and $ \omega(\bk) = |\bk|^2 $.\\
\begin{proof}
We show that $ W_{\IBC}: \Sdot_\sF \to \Sdot_\sF $ is bijective by constructing $ W_{\IBC}^{-1}: \Sdot_\sF \to \Sdot_\sF $. Bijectivity of $ W_{\IBC} $ on $ \FB $ then readily follows by linearity with respect to the field $ \eRen $.\\

First, note that $ H_0^{-1}: \Sdot_\sF \to \Sdot_\sF $ is well-defined. That is because $ H_0 $ is a multiplication on configuration space by the function
\begin{equation}
	H_0(\bP, \bK) = \sum_{j = 1}^M \theta(\bp_j) + \sum_{\ell = 1}^N \omega(\bk_\ell) \;,
\end{equation}
which is nonzero on all configurations $ (\bP, \bK) $ that do not contain a zero boson momentum, i.e., everywhere apart from $ \exists(\bk = 0) $ (defined in \eqref{eq:existsk0}). So $ H_0(\bP, \bK)^{-1} $ is also defined apart from $ \exists(\bk = 0) $. Now, in the vicinity of $ \exists(\bk = 0) $, $ H_0(\bP, \bK) $ is bounded from below by $ c \sum_{\ell = 1}^N |\bk_\ell|^\beta $, so $ H_0(\bP, \bK)^{-1} $ blows up at most polynomially in each $ |\bk_\ell| $. Thus, the multiplication by $ H_0(\bP, \bK)^{-1} $ is a map $ \Sdot_\sF \to \Sdot_\sF $ that defines the operator $ H_0^{-1} $.\\

On $ \Sdot_\sF $, we now define $ W_{\IBC}^{-1} $ by the Neumann series:
\begin{equation}
	W_{\IBC}^{-1} = (1 + H_0^{-1} A^{\dagger})^{-1} := \sum_{k=0}^\infty (- H_0^{-1} A^\dagger)^k \;.
\label{eq:neumannseries}
\end{equation}

\underline{Claim:} \eqref{eq:neumannseries} defines an operator $ \Sdot_\sF \to \Sdot_\sF $.\\

\underline{Proof of the Claim:} Each $ - H_0^{-1} A^{\dagger} $ increases the boson number by 1. So $ (- H_0^{-1} A^\dagger)^k \Psi $ is only supported on configuration space sectors with $ N \ge k $. Hence, on each $ (\bP,\bK) \in \Qdot $ with $ \bK \in \RRR^{Nd} $, we have that
\begin{equation}
\begin{aligned}
	\Big( (- H_0^{-1} A^\dagger)^k \Psi \Big)(\bP,\bK) &= 0 \qquad \text{for }k > N\\
	\Rightarrow \quad \left( \sum_{k=0}^\infty (- H_0^{-1} A^\dagger)^k \Psi \right)(\bP,\bK) &= \left( \sum_{k=0}^N (- H_0^{-1} A^\dagger)^k \Psi \right)(\bP,\bK) \;.\\
\end{aligned}
\label{eq:creationvanishing}
\end{equation}

Since $ H_0^{-1} A^{\dagger} $ maps $ \Sdot_\sF $ to itself, also all sums $ \sum_{k=0}^n (- H_0^{-1} A^\dagger)^k $ with $ n \in \NNN $ map $ \Sdot_\sF \to \Sdot_\sF $. So the Neumann series is defined on each $ N $-boson sector and hence on all of $ \Sdot_\sF $. \hfill $ \Diamond $\\

\underline{Claim:} The Neumann series \eqref{eq:neumannseries} is the inverse of $ (1 + H_0^{-1} A^{\dagger}) $.\\

\underline{Proof of the Claim:} We use the first line of \eqref{eq:creationvanishing} and perform a sector-wise verification:
\begin{equation}
\begin{aligned}
	&(1 + H_0^{-1} A^{\dagger}) \left( \sum_{k=0}^\infty (- H_0^{-1} A^\dagger)^k \Psi \right)(\bP,\bK)\\
	 = &\left( \sum_{k=0}^\infty (- H_0^{-1} A^\dagger)^k \Psi - \sum_{k=1}^\infty (- H_0^{-1} A^\dagger)^k \Psi \right)(\bP,\bK)\\
	 \overset{\eqref{eq:creationvanishing}}{=} &\left( \sum_{k=0}^N (- H_0^{-1} A^\dagger)^k \Psi - \sum_{k=1}^N (- H_0^{-1} A^\dagger)^k \Psi \right)(\bP,\bK)\\
	 = &\Psi(\bP,\bK) \;.\\
\end{aligned}
\end{equation} 
So indeed $ (1 + H_0^{-1} A^{\dagger}) \left( \sum_{k=0}^\infty (- H_0^{-1} A^\dagger)^k \right) = 1 $. \hfill $ \Diamond $\\

So $ W_{\IBC}^{-1}: \Sdot_\sF \to \Sdot_\sF $ defined in \eqref{eq:neumannseries} is the (well-defined) inverse of $ W_{\IBC} $ on $ \Sdot_\sF $, which establishes the proposition.\\
\end{proof}

With Proposition \ref{prop:IBC}, we have a well-defined linear space $ W_{\IBC}^{-1}[\Sdot_\sF \cap \sF] $ that can be equipped with a scalar product
\begin{equation}
	\langle W_{\IBC}^{-1} \Psi, W_{\IBC}^{-1} \Phi \rangle_{\renI} := \langle \Psi, \Phi \rangle \qquad \text{for } \Psi, \Phi \in \Sdot_\sF \cap \sF \;.
\end{equation}
The completion of $ W_{\IBC}^{-1}[\Sdot_\sF \cap \sF] $ with respect $ \langle \cdot, \cdot \rangle_{\renI} $ is a Hilbert space $ \sF_{\renI} $, which we call the \textbf{IBC-renormalized Fock space}. $ H_{\IBC} $ is then defined on $ \sF_{\renI} $. The pullback to $ \sF $ reads
\begin{equation}
	\widetilde{H}_{\IBC} = W_{\IBC} H_{\IBC} W_{\IBC}^{-1} \;.
\label{eq:HIBCpullback1}
\end{equation}
Whenever the expression \eqref{eq:HIBCpullback1} extends to a self-adjoint operator, $ H_{\IBC} $ extends to a self-adjoint operator on $ \sF_{\renI} $.\\

\paragraph{Remarks.}

\begin{enumerate}
\setcounter{enumi}{\theremarks}
\item Another IBC-renormalized Hamiltonian would formally be given by the expression $ (W_{\IBC}^*)^{-1} H_{\IBC} W_{\IBC}^{-1} $ with $ W_{\IBC}^* = 1 + A H_0^{-1} $. In view of \eqref{eq:formalmanipulationIBC}, this expression would even be more desirable, since $ (W_{\IBC}^*)^{-1} $ would eliminate the factor of $  (1 + A H_0^{-1}) $ inside $ S^* $. However, we only know of $ A $ that it maps $ \FB \to \FB_{\ex} $, so $ A H_0^{-1} A $ and hence a Neumann series for $ (W_{\IBC}^*)^{-1} $ is generally ill-defined. If one succeeds to refine the definitions of $ \FB $ and $ \FB_{\ex} $ in such a way that also higher powers of $ A $ are well-defined, then the expression $ (W_{\IBC}^*)^{-1} H_{\IBC} W_{\IBC}^{-1} $ might be used to define a renormalized Hamiltonian. Still, such a refinement is beyond the scope of this paper.\\

\end{enumerate}
\setcounter{remarks}{\theenumi}



\subsection{The $ e^{- H_0^{-1} A^{\dagger}} $-Transformation Inspired by CQFT}	
\label{subsec:hepptrafo}

In the CQFT literature \cite{Friedrichs1965, Glimm1968, Hepp}, operators of a form similar to $ T \propto e^{- H_0^{-1} A^{\dagger}} $ appear. Further, within \cite{Glimm1968}, a renormalized scalar product of the kind
\begin{equation}
	\langle T \Psi, T \Phi \rangle_{\ren} := \lim_{\Lambda \to \infty} \langle T_\Lambda \Psi, T_\Lambda \Phi \rangle e^{-\bLambda_\Lambda} \qquad \forall T \Psi, T \Phi \in T[\cD], \quad \cD \subset \sF 
\end{equation}
is constructed, where the operators $ T, T_\Lambda $ in \cite{Glimm1968} are, however, much more involved. Here, $ \bLambda_\Lambda = 4! \Vert \gamma^{-1} v_{4, \Lambda} \Vert_2^2 $, where $ v_4: \RRR^{4d} \to \CCC $ is a cut-off form factor of the $ \Phi^4 $-interaction and $ \gamma(\bk_1, \ldots, \bk_4) = \sum_{\ell = 1}^4 \omega(\bk_\ell) $. The renormalized Hamiltonian is then constructed by a limiting procedure \cite[Thm.~4.1.1]{Glimm1968}:
\begin{equation}
	\langle T \Psi, H_{\ren} T \Phi \rangle_{\ren} = \lim_{\Lambda \to \infty} \langle T_\Lambda \Psi, H_{\ren}(\Lambda) T_\Lambda \Phi \rangle e^{-\bLambda_\Lambda} \;,
\end{equation}
with $ H_{\ren}(\Lambda) $ containing cut-off counterterms. The limit $ \Lambda \to \infty $ formally leads to an infinite wave function renormalization of the kind $ e^{- \bLambda/2} $.\\

As a simplified divergent term in our prototypical transformation, we choose
\begin{equation}
	\bLambda = 4! \Vert \omega^{-1} v \Vert_2^2 \;.
\end{equation}
The ESS construction then allows to interpret $ \bLambda \in \Ren_1 $. In the simplified case $ T \propto e^{- H_0^{-1} A^{\dagger}} $, we can even define the dressing transformation directly on $ \FB $:
\begin{proposition}
	Under the assumptions of Proposition \ref{prop:IBC}, for $ \Psi \in \Sdot_\sF $, we have
\begin{equation}
	T \Psi := e^{- \bLambda/2} e^{-H_0^{-1} A^{\dagger}} \Psi \in \FB \;,
\end{equation}
with $ e^{- \bLambda/2} \in \eRen $. In particular, $ e^{-H_0^{-1} A^{\dagger}} $ and $ T $ are well-defined linear operators
\begin{equation*}
	e^{-H_0^{-1} A^{\dagger}}: \Sdot_\sF \to \Sdot_\sF \;, \qquad
	T: \Sdot_\sF \to \FB \;,
\end{equation*}
that extend by $ \eRen $-linearity to $ \FB \to \FB $.\\
\label{prop:Hepptrafo}
\end{proposition}
\begin{proof}
As argued in the proof of Proposition \ref{prop:IBC}, $ H_0^{-1} A^\dagger $ is well-defined on $ \Sdot_\sF $ and maps Fock space vectors supported on the $ N $-particle sector to those supported on the $ N+1 $-particle sector. We can now write the exponential series as
\begin{equation}
	e^{- H_0^{-1} A^{\dagger}} = \sum_{k=0}^{\infty} \frac{(- H_0^{-1} A^{\dagger})^k}{k!} = 1 - H_0^{-1} A^{\dagger} + \frac{1}{2} H_0^{-1} A^{\dagger}H_0^{-1} A^{\dagger} - \ldots \;.
\label{eq:eH0Adagger}
\end{equation}
Each $ - H_0^{-1} A^{\dagger} $ maps $ \Sdot_\sF $ to itself and strictly increases the sector number. So the $ N $-sector of $ e^{- H_0^{-1} A^{\dagger}} \Psi $ may only depend on at most $ N+1 $ terms of the series \eqref{eq:eH0Adagger}. All terms are elements of $ \Sdot_\sF \subset \FB $. Hence, $ e^{- H_0^{-1} A^{\dagger}} $ maps $ \Sdot_\sF $ into itself, as claimed. Clearly, the factor $ e^{- \bLambda/2} $ is an element of $ \eRen $, so $ T \Psi \in \FB $.\\
\end{proof}

Of course, the transformation $ T_\Lambda $ appearing in \cite{Glimm1968} is much more involved, and it remains an interesting and highly nontrivial question whether a corresponding $ T $ can be defined using an ESS construction adapted to that model. \\

\begin{appendices}
\addtocontents{toc}{\setcounter{tocdepth}{-2}}

\section{Proof of Lemma \ref{lem:coherentdense}}
\label{sec:coherentdense}

\begin{proof}
We show that $ \cD_{WS} $ is dense in $ L^2(\cQ_x) \otimes \sF_y $. Since $ (S_- \otimes \id): L^2(\cQ_x) \otimes \sF_y \to \sF $ is surjective and bounded with norm 1, we then immediately get density of $ (S_- \otimes \id)[\cD_{WS}] $ within $ \sF $.\\

If $ A^{\dagger}_1(\varphi) $ created a boson without giving a recoil to a fermion, we would be done: In this case, $ W_{\sF, 1}(\varphi) \Psi_m $ would be of the form
\begin{equation}
	\Psi_{mx} \otimes W_y(\varphi) \Omega_y \;.
\label{eq:Wnmequality}
\end{equation}
Now, $ \Psi_{mx} \in \cS(\cQ_x) $, which is dense in $ L^2(\cQ_x) $. Further, we have that $ \span\{W_y(\varphi) \Omega_y \; \mid \; \varphi \in \fh \cap \Sdot_1\} $ is dense in $ \sF_y $, so
\begin{equation}
	\span \big\{ \Psi_{mx} \otimes W_y(\varphi) \Omega_y \; \big\vert \; \varphi \in \fh \cap \Sdot_1, \; \Psi_{mx} \in \cS(\cQ_x) \big\} \;,
\label{eq:Wnmdense}
\end{equation}
is dense in $ L^2(\cQ_x) \otimes \sF_y $.\\
However, $ W_{\sF, 1}(\varphi) \Psi_m $ is not of the form \eqref{eq:Wnmequality}, since $ W_{\sF, 1}(\varphi) $ shifts the momentum $ \bp_1 $ by $ \sum_{\ell = 1}^N \bk_{\ell} $, as in \eqref{eq:Ws1} (the first fermion gets a recoil). The same recoil occurs when applying $ A_1^\dagger(v) $. In other words, creation and dressing ``entangle'' the fermion with the created boson by giving the fermion a recoil. In order to solve this problem, we introduce a \textbf{disentangling operator} $ D_1: L^2(\cQ_x) \otimes \sF_y \to L^2(\cQ_x) \otimes \sF_y $, which removes all recoils:
\begin{equation}
	(D_1 \Psi)(\bP,\bK) := \Psi \left(\bP - \sum_{\ell = 1}^N e_1 \bk_{\ell} , \bK \right) \;.
\end{equation}
Clearly, $ D_1 $ is unitary. Now,
\begin{equation}
\begin{aligned}
	(W_{\sF, 1}(\varphi) \Psi_m)(\bP,\bK) &= \frac{e^{-\frac{\Vert \varphi \Vert^2}{2}}}{\sqrt{N!}} \left( \prod_{\ell = 1}^N \varphi(\bk_{\ell}) \right) \Psi_{mx} \left( \bP + \sum_{\ell = 1}^N e_1 \bk_{\ell} \right)\\
	\Rightarrow (D_1 W_{\sF, 1}(\varphi) \Psi_m)(\bP,\bK) &= \frac{e^{-\frac{\Vert \varphi \Vert^2}{2}}}{\sqrt{N!}} \left( \prod_{\ell = 1}^N \varphi(\bk_{\ell}) \right) \Psi_{mx}(\bP)\\
	\Leftrightarrow D_1 W_{\sF, 1}(\varphi) \Psi_m &= \Psi_{mx} \otimes \sum_{N=0}^\infty \frac{e^{-\frac{\Vert \varphi \Vert^2}{2}}}{\sqrt{N!}} \; \underbrace{\varphi \otimes \ldots \otimes \varphi}_{N \text{ times}}\\
	&= \Psi_{mx} \otimes W_y(\varphi) \Omega_y \;.
\end{aligned}
\label{eq:D1W1split}
\end{equation}

So by density of \eqref{eq:Wnmdense}, we have that
\begin{equation*}
	\span \big\{ D_1 W_{\sF, 1}(\varphi) \Psi_m \; \big\vert \; \varphi \in \fh \cap \Sdot_1, \; \Psi_m \in \cC_{WS} \big\} \;,
\end{equation*}
is dense in $ L^2(\cQ_x) \otimes \sF_y $. And since $ D_1 $ is unitary, its preimage
\begin{equation*}
	\cD_{WS} = \span \big\{ W_{\sF, 1}(\varphi) \Psi_m \; \big\vert \; \varphi \in \fh \cap \Sdot_1, \; \Psi_m \in \cC_{WS} \big\} \;,
\end{equation*}
is dense in $ L^2(\cQ_x) \otimes \sF_y $, as well. So $ (S_- \otimes \id)[\cD_{WS}] $ is dense in $ \sF $.\\
\end{proof}

\section{Proof of Lemma \ref{lem:commutationW}}
\label{sec:commutationW}

\begin{proof}
For evaluating $ [A^\dagger,W] $, we first consider the simple case where $ \Psi $ is replaced by a boson-only vector $ \Psi_y \in \sF_y $. A vector with one boson can be written as:
\begin{equation}
	\phi = a^{\dagger}(\phi) \Omega_y \;,
\label{eq:singleboson}
\end{equation}
with $ \phi \in \fh = \sF_y^{(1)} $. A coherent displacement can then be described using
\begin{equation}
	W_y(\varphi) a^{\dagger}(\phi) = a^{\dagger}(\phi) W_y(\varphi) - [a^{\dagger}(\phi), W_y(\varphi)] \;,
\label{eq:Wa_aW_commutator}
\end{equation}
where $ \varphi \in \fh \cap \Sdot_1 $. The first expression is easily computed in momentum space
\begin{equation}
\begin{aligned}
	(a^{\dagger}(\phi) \Psi_y)(\bK) &= \frac{1}{\sqrt{N}} \sum_{\ell = 1}^N \phi(\bk_{\ell}) \Psi_y(\bK \setminus \bk_{\ell})\\
	\Rightarrow \quad (a^{\dagger}(\phi) W_y(\varphi) \Omega_y)(\bK) &= \frac{e^{-\frac{\Vert \varphi \Vert^2}{2}}}{\sqrt{N!}} \sum_{\ell = 1}^N \phi(\bk_{\ell}) \left( \prod_{\ell' \neq \ell} \varphi(\bk_{\ell'}) \right) \;.\\
\end{aligned}
\label{eq:agadderWns}
\end{equation}

The commutator $ [a^{\dagger}(\phi), W_y(\varphi)] $ is computed using
\begin{equation}
\begin{aligned}
	&(*) \quad &&[a^{\dagger}(\varphi), a^{\dagger}(\phi)] = 0 \;,\\
	&(**) \quad &&[a(\varphi), a^{\dagger}(\phi)] = \langle \varphi,\phi \rangle \;.
\end{aligned}
\label{eq:starstar}
\end{equation}
We have
\begin{equation}
\begin{aligned}
	&[a^{\dagger}(\phi), W_y(\varphi)] = [a^{\dagger}(\phi), e^{a^{\dagger}(\varphi) - a(\varphi)}] = \sum_{k=0}^\infty \frac{1}{k!} [a^{\dagger}(\phi), (a^{\dagger}(\varphi) - a(\varphi))^k]\\
	\overset{(*)}{=} &\sum_{k=0}^\infty \frac{(-1)^k}{k!} \left( [a^{\dagger}(\phi), a(\varphi)] (a^{\dagger}(\varphi) - a(\varphi))^{k-1} \right.\\
	&+ (a^{\dagger}(\varphi) - a(\varphi)) [a^{\dagger}(\phi), a(\varphi)] (a^{\dagger}(\varphi) - a(\varphi))^{k-2}\\
	& \left. + \ldots + (a^{\dagger}(\varphi) - a(\varphi))^{k-1} [a^{\dagger}(\phi), a(\varphi)] \right)\\
	\overset{(**)}{=} &- \sum_{k=1}^\infty \frac{(-1)^k}{k!} k \langle \varphi,\phi \rangle (a^{\dagger}(\varphi) - a(\varphi))^{k-1} = \sum_{k=0}^\infty \frac{(-1)^k}{k!} \langle \varphi,\phi \rangle (a^{\dagger}(\varphi) - a(\varphi))^{k}\\
	= & \langle \varphi,\phi \rangle W_y(\varphi) \;.\\
\end{aligned}
\label{eq:pullthrough}
\end{equation}

It is easy to check that all above formulas hold as strong operator identities on the space of finite-boson states $ \sF_{\fin, y} $ (defined in \eqref{eq:sFfiny}). By \eqref{eq:nelsonestimate}, an application of $ k $ operators of the kind\footnote{Here, we allow $ \sharp $ to be a different superscript in each factor. E.g., $ (a^\sharp)^4 $ would also represent $ a a^\dagger a a^\dagger $.} $ a^\sharp(\varphi) \in \{ a(\varphi), a^\dagger(\varphi) \} $ to $ \Psi_y \in \sF_{\fin, y} $ is bounded by
\begin{equation}
	\left\Vert \big( a(\varphi)^\sharp \big)^k \Psi_y \right\Vert \le 
	\sqrt{\frac{(N_{\max} + k)!}{N_{\max}!}} \Vert \Psi_y \Vert \; \Vert \varphi \Vert^n \;.
\label{eq:nelsonboundn}
\end{equation}
This allows to estimate
\begin{equation}
	\big\Vert a^\sharp(\phi) W_y(\varphi) \Psi \big\Vert \le \sum_{k \in \NNN_0} \frac{1}{k!} \sqrt{\frac{(N_{\max} + k + 1)!}{N_{\max}!}} \Vert \Psi_y \Vert \; \Vert 2 \varphi \Vert^n \; \Vert \phi \Vert < \infty \;,
\end{equation}
and an analogous estimate shows that $ W_y(\varphi) a^\sharp(\phi) $ is well-defined.\\

Putting together \eqref{eq:Wa_aW_commutator} and \eqref{eq:pullthrough}, we obtain the action of $ W_y(\varphi) $ on single-boson states:
\begin{equation}
W_y(\varphi) a^{\dagger}(\phi) \Omega_y = ( a^{\dagger}(\phi) - \langle \varphi,\phi \rangle) W_y(\varphi) \Omega_y \;.
\label{eq:Wadagger}
\end{equation}

Now, we turn to state vectors with many fermions and one boson, $ A_1^{\dagger}(\phi) \Psi_m \in L^2(\cQ_x) \otimes \sF_y $. Further, we go over from dressings by $ W_y(\varphi) $ to $ W_{\sF, j}(\varphi) $, which is done replacing $ a(\phi), a^{\dagger}(\phi) $ by $ A_j(\phi), A_j^{\dagger}(\phi) $. Note that $ A_j(\phi), A_j^{\dagger}(\phi) $ are no longer merely creating and annihilating bosons, but they also shift a fermion's momentum. Computations in \eqref{eq:pullthrough} run through in almost the same manner. We have to replace $ (**) $ by the CCR \eqref{eq:CCR}. If $ j \neq j' $, we further use that $ V_{jj'}(\varphi^* \phi) $ (which replaces $ \langle \varphi,\phi \rangle $ in \eqref{eq:pullthrough}) commutes with\footnote{This is shown by a similar fiber decomposition argument, as used in the proof of Lemma \ref{lem:Xcommute}.} $ A_{j''}(\varphi) $ and $ A_{j''}^\dagger(\varphi) $, so we can still pull it to the left.\\

The final result is
\begin{equation}
	W_{\sF, j}(\varphi) A_{j'}^{\dagger}(\phi) = \begin{cases} (A_{j'}^{\dagger}(\phi) - \langle \varphi,\phi \rangle) W_{\sF, j}(\varphi) \quad &\text{if } j = j' \\ (A_{j'}^{\dagger}(\phi) - V_{jj'}(\varphi^* \phi) ) W_{\sF, j}(\varphi) \quad &\text{if } j \neq j' \end{cases} \;.
\label{eq:Ws1nAdagger1}
\end{equation}
This is one of the four identities claimed in \eqref{eq:WAdaggers}. The other three identities follow by summation over $ j $ or $ j' $.\\

The four identities in \eqref{eq:WAs} are obtained analogously. In place of $ (**) $, we use the CCR \eqref{eq:CCR} together with $ [A_j(\varphi), A_{j'}(\phi)] = 0 $ and keep in mind that the factor $ (-1)^k $ drops out in \eqref{eq:pullthrough}, which yields the desired result.\\

It is easy to check that all identities above hold as strong operator identities on the domain $ \sF_{\fin} $ defined in \eqref{eq:sFfinESS}. The momentum space definitions of $ A_j^\dagger(\varphi), A_j(\varphi) $ in \eqref{eq:adaggerp} and \eqref{eq:ap} directly yield the well-known estimates
\begin{equation}
	\big\Vert A_j^\dagger(\varphi) \Psi \big\Vert \le \big\Vert (N+1)^{1/2} \Psi \big\Vert \; \Vert \varphi \Vert \;, \qquad
	\big\Vert A_j(\varphi) \Psi \big\Vert \le \big\Vert N^{1/2} \Psi \big\Vert \; \Vert \varphi \Vert \;,
\label{eq:nelsonestimateA}
\end{equation}
which are analogous to \eqref{eq:nelsonestimate} and allow for employing the same arguments as below \eqref{eq:sFfiny}. Thus, all expressions in \eqref{eq:WAdaggers} and \eqref{eq:WAs} are well-defined on $ \sF_{\fin} $.\\
\end{proof}

\paragraph{Remarks.}

\begin{enumerate}
\setcounter{enumi}{\theremarks}
\item If we replace the form factors $ \varphi(\bk) $ and $ \phi(\bk) $ in $ A^{\dagger}(\phi) $ by $ \varphi(\bp, \bk) $ and $ \phi(\bp, \bk) $ (as they appear in more realistic QFT models), then $ [A_j^\dagger(\varphi), A_{j'}^\dagger(\phi)] = 0 $ will no longer hold true. The operators $ A_j^{\dagger}(\phi) $ change the momentum of the fermion which is emitting a boson. So when one fermion creates two bosons with different form factors, which depend on the fermion momentum, then it can make a difference, which boson is created first. In addition, $ V_{j'j''} $ will no longer commute with $ A_j $ and $ A_j^{\dagger} $. In this case, several multi-commutators of the form $ [A_{j_1}^{\sharp},[A_{j_2}^{\sharp}, \ldots, [A_{j_n}^{\sharp},A_{j'}^{\sharp}]]] $ with $ A_j^\sharp \in \{A_j ,A_j^\dagger \} $ appear.
\end{enumerate}
\setcounter{remarks}{\theenumi}

\section{Proof of Lemma \ref{lem:coherentcombination}}
\label{sec:coherentcombination}

\begin{proof}
First, note that $ \Psi'_{m, k} \neq 0 $ and $ v_k \neq 0 $, since otherwise at least one vector $ \Psi_k $ would be 0.\\

We now establish the linear independence by an induction over $ K $.\\
For $ K = 1 $, linear independence is an obvious fact. Also for $ |\cK_{WA} | = | \cK_W | = 1 $, linear independence is easy to see: The $ (N) $-sectors have the norms
\begin{equation}
\begin{aligned}
	\left\Vert (W_{\sF, 1} (\varphi_1) \Psi'_{m, 1})^{(N)} \right\Vert^2 = &\frac{e^{- \Vert \varphi_1 \Vert^2}}{N!} \Vert \varphi_1 \Vert^{2N} \Vert \Psi'_{m, 1} \Vert^2 \;,\\ 
	\left\Vert (W_{\sF, 1} (\varphi_2) A^\dagger_{j_2}(v_2) \Psi'_{m, 2})^{(N)} \right\Vert^2 = &\frac{e^{- \Vert \varphi_2 \Vert^2}}{(N-1)!} \Vert \varphi_2 \Vert^{2(N-1)} \Vert v_2 \Vert^2 \Vert \Psi'_{m, 2} \Vert^2 \;.
\end{aligned} 
\end{equation}
These cannot agree for all $ N \in \NNN_0 $ simultaneously, as $ \frac{N!}{(N-1)!} \frac{\Vert \varphi_1 \Vert^{2N}}{\Vert \varphi_2 \Vert^{2(N-1)}} $ is never constant in $ N $.\\

Now assume (induction assumption), we have shown linear independence for any set containing $ \le K - 1 $ vectors $ \Psi_k $ of the kind \eqref{eq:Psikform} and consider any other set of $ K $ vectors of the kind \eqref{eq:Psikform}, where $ |\cK_{WA} | \ge 2 $ or $ | \cK_W | \ge 2 $, so we have at least two distinct form factors $ \varphi_k $. Suppose, there was a linear combination
\begin{equation}
	\Psi
	:= \sum_{k = 1}^K c_k \Psi_k
	= \sum_{k \in \cK_{WA}} c_k W_{\sF, 1}(\varphi_k) A^\dagger_{j_k}(v_k) \Psi'_{m, k}
	+ \sum_{k \in \cK_W} c_k W_{\sF, 1}(\varphi_k) \Psi'_{m, k}
	= 0 \;,
\label{eq:WWAlincomb}
\end{equation}
with $ c_k \in \CCC $ not being all zero. Without loss of generality, we may assume $ c_k \neq 0 $ for all $ k $.\\

Now, by premise of the lemma, each $ \varphi_k $ can appear at most twice in \eqref{eq:WWAlincomb}. We group equal form factors together by introducing the partition 
\begin{equation}
	\{1, \ldots, K\} = \bigcup_{z = 1}^Z \cK_z \;, \qquad \cK_z \cap \cK_{z'} = \emptyset \quad \text{for} \quad z \neq z' \;,
\end{equation}
such that $ \varphi_k = \varphi_{k'} $ if and only if $ k $ and $ k' $ belong to the same index set $ \cK_z $. In particular, $ K/2 \le Z \le K $ and $ |\cK_z| \le 2 $. Since there are at least two distinct $ \varphi_k $, we also have
\begin{equation}
	\max_{k} \Vert \varphi_k \Vert > 0 \;.
\end{equation}
We pick an index $ \overline{z} $, such that $ \Vert \varphi_{\overline{k}} \Vert $ attains this maximum for $ \overline{k} \in \cK_{\overline{z}} $. Regarding $ \cK_{\overline{z}} $, there are now two cases which can occur as follows.\\

\underline{Case 1}: $ \cK_{\overline{z}} \cap \cK_{WA} = \emptyset $. For $ z' \in \{1, \ldots, Z\} $ with the unique associated form factor $ \varphi_{k'}, k' \in \cK_{z'} $, consider the ``$ (N) $-sector trial state''
\begin{equation}
	\Phi_{z'}^{(N)}(\bK) = \prod_{\ell = 1}^N \varphi_{k'}(\bk_{\ell}) \;, \qquad \Phi_{z'}^{(N)} \in L^2(\RRR^{Nd}) \;.
\label{eq:trialstate1}
\end{equation}
We choose a fixed representative function for $ L^2 $-elements, and use the following abbreviations
\begin{equation}
\begin{aligned}
	\Psi''_k := e^{-\frac{\Vert \varphi_k \Vert^2}{2}} c_k \Psi'_{m, k} \;, \quad
	w_{k, \bP}(\bk) := v_k(\bk) \Psi''_k(\bP + (e_{j_k} - e_1) \bk) \;, \quad
	&\xi_{z'z} := \langle \varphi_{k'}, \varphi_k \rangle \;,
\end{aligned} 
\label{eq:Psiprimeprime}
\end{equation}
where $ w_{k, \bP} $ is only defined for $ k \in \cK_{WA} $. Integrating $ \Phi_{z'}^{(N)} $ against the $ (N) $-boson sector of $ \sum_{k \in \cK_z} \Psi_k $ (which contains at most two terms), we now obtain for almost all\footnote{There is no guarantee that $ w_{k, \bP} \in L^2 $, so integrals involving $ w_{k, \bP} $ might be ill-defined. However, since $ \Psi, v \in L^2 $, we have $ w_{k, \bP} \notin L^2 $ only for $ \bP $ within some null set in $ \cQ_x $, so \eqref{eq:xiequality} holds almost everywhere.} $ \bP \in \cQ_x $:
\begin{equation}
\begin{aligned}
	&\int \Phi_{z'}^{(N)}(\bK)^* \left( \sum_{k \in \cK_z} c_k \Psi_k \left(\bP - e_1 \sum_{\ell = 1}^N \bk_{\ell}, \bK \right) \right) \; d \bK\\
	= &\sum_{k \in \cK_z \cap \cK_{WA}} (N!)^{-1/2} N \langle \varphi_{k'}, \varphi_k \rangle^{N-1} \langle \varphi_{k'}, w_{k, \bP} \rangle\\
	&+ \sum_{k \in \cK_z \cap \cK_W} (N!)^{-1/2} \langle \varphi_{k'}, \varphi_k \rangle^N \Psi''_k(\bP)\\
	= & N (N!)^{-1/2} \xi_{z' z}^{N-1} \widetilde{\Psi}_{z' z, 1}(\bP)
	+ (N!)^{-1/2} \xi_{z' z}^{N-1} \widetilde{\Psi}_{z' z, 2}(\bP) \;,
\end{aligned}
\label{eq:xiequality}
\end{equation} 
where $ \widetilde{\Psi}_{z' z, 1}, \widetilde{\Psi}_{z' z, 2} $ do not depend on $ N $ and are given by
\begin{equation}
	\widetilde{\Psi}_{z' z, 1}(\bP) := \sum_{k \in \cK_z \cap \cK_{WA}} \langle \varphi_{k'}, w_{k, \bP} \rangle \;, \qquad 
	\widetilde{\Psi}_{z' z, 2}(\bP) := \sum_{k \in \cK_z \cap \cK_W} \xi_{z' z} \Psi''_k(\bP) \;.
\label{eq:widetildePsi}
\end{equation}
Note that each of the sums over $ \cK_z \cap \cK_{WA} $ or $ \cK_z \cap \cK_W $ contain at most one term, so the $ k $ associated with $ z $ in \eqref{eq:widetildePsi} is unique. Now, consider the ratio
\begin{equation}
	R := \esssup_{\bP \in \cQ_x \setminus \cN_x} \frac{ \left\vert \widetilde{\Psi}_{\overline{z} \overline{z}, 2}(\bP) \right\vert}{ \sum_{z \neq \overline{z}} \left\vert \widetilde{\Psi}_{\overline{z} z, 1}(\bP) \right\vert + \sum_{z \neq \overline{z}} \left\vert  \widetilde{\Psi}_{\overline{z} z, 2}(\bP) \right\vert } \ge 0 \;,
\label{eq:R}
\end{equation}
where $ \cN_x $ denotes the set of all $ \bP $, where the above fraction amounts to $ \frac{0}{0} $. Since $ \cK_{\overline{z}} \neq \emptyset $, the set $ \cK_{\overline{z}} \cap \cK_W $ contains exactly one element $ \overline{k} $, with $ \Psi''_{\overline{k}}(\bP) \neq 0 $ on a set of positive measure in $ \bP \in \cQ_x $. And as $ \varphi_{\overline{k}} \neq 0 $, we have $ \widetilde{\Psi}_{\overline{z} \overline{z}, 2}(\bP) \neq 0 $ for those $ \bP $, so $ R > 0 $. Hence, the set
\begin{equation}
	U = \left\{ \bP \in \cQ_x \; \middle\vert \; \left\vert \widetilde{\Psi}_{\overline{z} \overline{z}, 2}(\bP) \right\vert
	> \frac{R}{2} \left(  \sum_{z \neq \overline{z}} \left\vert \widetilde{\Psi}_{\overline{z} z, 1}(\bP) \right\vert + \sum_{z \neq \overline{z}} \left\vert  \widetilde{\Psi}_{\overline{z} z, 2}(\bP) \right\vert \right) \text{ and \eqref{eq:xiequality} holds} \right\} \;,
\label{eq:U}
\end{equation}
has positive measure. Further, as $ \Vert \varphi_{\overline{k}} \Vert $ is maximal and $ \varphi_k \neq \varphi_{\overline{k}} $ for $ k \neq \overline{k} $, we conclude that
\begin{equation}
	\langle \varphi_{\overline{k}}, \varphi_{\overline{k}} \rangle > | \langle \varphi_{\overline{k}}, \varphi_k \rangle | \quad \Leftrightarrow \quad
	|\xi_{\overline{z} \overline{z}}| > |\xi_{\overline{z} z}| \qquad \text{for} \quad z \neq \overline{z} \;.
\label{eq:ximax}
\end{equation}
Hence, there is an $ N \in \NNN $ with
\begin{equation}
	| \xi_{\overline{z} \overline{z}}^{N-1} | > \frac{4}{R} | N \xi_{\overline{z} z}^{N-1} | \quad \text{for} \quad z \neq \overline{z} \;.
\label{eq:Requation}
\end{equation}
Now, for this $ N $ and $ \bP \in U $, the term $ \widetilde{\Psi}_{\overline{z} \overline{z}, 2} $ is dominant in the following sense:
\begin{equation}
\begin{aligned}
	&\left\vert \int \Phi_{\overline{k}}^{(N)}(\bK)^* \left( \sum_{k = 1}^K c_k \Psi_k \left(\bP - e_1 \sum_{\ell = 1}^N \bk_{\ell}, \bK \right) \right) \; d \bK \right\vert\\
	\overset{\eqref{eq:xiequality}}{=} & (N!)^{-1/2} \left\vert \sum_{z = 1}^Z N  \xi_{\overline{z} z}^{N-1} \widetilde{\Psi}_{\overline{z} z, 1}(\bP)
	+ \sum_{z = 1}^Z \xi_{\overline{z} z}^{N-1} \widetilde{\Psi}_{\overline{z} z, 2}(\bP) \right\vert\\
	\ge & (N!)^{-1/2} \left( \left\vert \xi_{\overline{z} \overline{z}}^{N-1} \widetilde{\Psi}_{\overline{z} \overline{z}, 2}(\bP) \right\vert
	- \sum_{z \neq \overline{z}} \left\vert N \xi_{\overline{z} z}^{N-1} \widetilde{\Psi}_{\overline{z} z, 1}(\bP) \right\vert - \sum_{z \neq \overline{z}}  \left\vert\xi_{\overline{z} z}^{N-1} \widetilde{\Psi}_{\overline{z} z, 2}(\bP) \right\vert \right)\\
	\overset{\eqref{eq:Requation}}{\ge} & (N!)^{-1/2} |\xi_{\overline{z} \overline{z}}|^{N-1} \left( \left\vert \widetilde{\Psi}_{\overline{z} \overline{z}, 2}(\bP) \right\vert
	- \frac{R}{4} \left( \sum_{z \neq \overline{z}} \left\vert \widetilde{\Psi}_{\overline{z} z, 1}(\bP) \right\vert + \sum_{z \neq \overline{z}} \left\vert  \widetilde{\Psi}_{\overline{z} z, 2}(\bP) \right\vert \right) \right)\\
	\overset{\eqref{eq:U}}{\ge} & (N!)^{-1/2} |\xi_{\overline{z} \overline{z}}|^{N-1} \left( \left\vert \widetilde{\Psi}_{\overline{z} \overline{z}, 2}(\bP) \right\vert
	- \frac{1}{2} \left\vert \widetilde{\Psi}_{\overline{z} \overline{z}, 2}(\bP) \right\vert \right)
	> 0 \;.
\end{aligned}
\label{eq:Psiineq1}
\end{equation}
So the $ L^2 $-function $ \bK \mapsto \Psi \left(\bP - e_1 \sum_{\ell = 1}^N \bk_{\ell}, \bK \right) $ cannot be zero for $ \bP \in U $. Since $ U $ has positive measure, we conclude $ \Psi \neq 0 $, which contradicts \eqref{eq:WWAlincomb}.\\

\underline{Case 2}: $ \cK_{\overline{z}} \cap \cK_{WA} \neq \emptyset $. Since $ \cK_{\overline{z}} \cap \cK_{WA} $ contains at most one element, we can set $ \cK_{\overline{z}} \cap \cK_{WA} =: \{ \overline{k} \} $. We copy the notations \eqref{eq:Psiprimeprime} and \eqref{eq:widetildePsi} from Case 1 and choose as a trial state for $ N \ge 1, k' \in \cK_{z'} \cap \cK_{WA} $ and $ \bP \in \cQ_x $:
\begin{equation}
	\Xi_{z', \bP}^{(N)}(\bK) = w_{k', \bP}(\bk_1) \left( \prod_{\ell = 2}^N \varphi_{k'}(\bk_{\ell}) \right) \;.
\label{eq:trialstate2}
\end{equation}
For almost all $ \bP \in \cQ_x $, we then obtain
\begin{equation}
\begin{aligned}
	&\int \Xi_{z', \bP}^{(N)}(\bK)^* \left( \sum_{k \in \cK_z} c_k \Psi_k \left(\bP - e_1 \sum_{\ell = 1}^N \bk_{\ell}, \bK \right) \right) \; d \bK\\
	= &\sum_{k \in \cK_z \cap \cK_{WA}} (N!)^{-1/2} \langle \varphi_{k'}, \varphi_k \rangle^{N-1} \langle w_{k', \bP}, w_{k, \bP} \rangle\\
	&+ \sum_{k \in \cK_z \cap \cK_{WA}} (N!)^{-1/2} (N-1) \langle \varphi_{k'}, \varphi_k \rangle^{N-2} \langle w_{k', \bP}, \varphi_k \rangle \langle \varphi_{k'}, w_{k, \bP} \rangle\\
	&+ \sum_{k \in \cK_z \cap \cK_W} (N!)^{-1/2} \langle \varphi_{k'}, \varphi_k \rangle^{N-1} \langle w_{k', \bP}, \varphi_k \rangle \Psi''_k(\bP)\\
	= &(N!)^{-1/2} \left( (N-1) \xi_{z' z}^{N-2} \widetilde{\Psi}'_{z' z, 1}(\bP)
	+ \xi_{z' z}^{N-2} \widetilde{\Psi}'_{z' z, 2}(\bP) \right) \;,
\end{aligned}
\label{eq:xiequality2}
\end{equation} 
with
\begin{equation}
\begin{aligned}
	\widetilde{\Psi}'_{z' z, 1}(\bP) := &\sum_{k \in \cK_z \cap \cK_{WA}} \langle w_{k', \bP}, \varphi_k \rangle \langle \varphi_{k'}, w_{k, \bP} \rangle \;,\\
	\widetilde{\Psi}'_{z' z, 2}(\bP) := &\sum_{k \in \cK_z \cap \cK_{WA}}  \xi_{z' z} \langle w_{k', \bP}, w_{k, \bP} \rangle
	+ \sum_{k \in \cK_z \cap \cK_W} \xi_{z' z} \langle w_{k', \bP}, \varphi_k \rangle \Psi''_k(\bP) \;.
\end{aligned}
\label{eq:widetildePsi2}
\end{equation}
First, consider the sub-case where $ \langle w_{\overline{k}, \bP}, \varphi_{\overline{k}} \rangle \neq 0 $ holds for all $ \bP $ inside some subset of $ \cQ_x $ with positive measure. So $ \widetilde{\Psi}'_{\overline{z} \overline{z}, 1}(\bP) \neq 0 $ for those $ \bP $. In this case, the term with $ \widetilde{\Psi}'_{\overline{z} \overline{z}, 1} $ will be dominant. Here, the ratio
\begin{equation}
	R' := \esssup_{\bP \in \cQ_x \setminus \cN_x'} \frac{ \left\vert \widetilde{\Psi}'_{\overline{z} \overline{z}, 1}(\bP) \right\vert}{ \sum_{z \neq \overline{z}} \left\vert \widetilde{\Psi}'_{\overline{z} z, 1}(\bP) \right\vert + \sum_{z = 1}^Z \left\vert \widetilde{\Psi}'_{\overline{z} z, 2}(\bP) \right\vert } \;,
\label{eq:Rprime}
\end{equation}
is strictly positive. Here, $ \cN_x' $ denotes the set of all $ \bP $ with the above fraction amounting to $ \frac{0}{0} $. Therefore, the set
\begin{equation}
	U' = \left\{ \bP \in \cQ_x \; \middle\vert \; \left\vert \widetilde{\Psi}'_{\overline{z} \overline{z}, 1}(\bP) \right\vert
	> \frac{R'}{2} \left( \sum_{z \neq \overline{z}} \left\vert \widetilde{\Psi}'_{\overline{z} z, 1}(\bP) \right\vert + \sum_{z = 1}^Z \left\vert \widetilde{\Psi}'_{\overline{z} z, 2}(\bP) \right\vert \right) \text{ and \eqref{eq:xiequality2} holds} \right\} \;,
\label{eq:Uprime}
\end{equation}
has positive measure. In this case, we can find some $ N \in \NNN $, such that
\begin{equation}
	(N-1) > \frac{4}{R'} \quad \text{and} \quad | \xi_{\overline{z} \overline{z}}^{N-2} | > \frac{4}{R'} | \xi_{\overline{z} z}^{N-2} | \quad \text{for} \quad z \neq \overline{z} \;.
\label{eq:Rprimeequation}
\end{equation}
Hence, for this $ N $ and $ \bP \in U $, we obtain:
\begin{equation}
\begin{aligned}
	&\left\vert \int \Xi_{\overline{z}, \bP}^{(N)}(\bK)^* \left( \sum_{k = 1}^K c_k \Psi_k \left( \bP - e_1 \sum_{\ell = 1}^N \bk_{\ell}, \bK \right) \right) \; d \bK \right\vert\\
	\overset{\eqref{eq:xiequality2}}{=}& (N!)^{-1/2} \left\vert \sum_{z = 1}^Z (N-1) \xi_{\overline{z} z}^{N-2} \widetilde{\Psi}'_{\overline{z} z, 1}(\bP)
	+ \sum_{z = 1}^Z \xi_{\overline{z} z}^{N-2} \widetilde{\Psi}'_{\overline{z} z, 2}(\bP) \right\vert\\
	\ge & (N!)^{-1/2} \left( \left\vert (N-1) \xi_{\overline{z} \overline{z}}^{N-2} \widetilde{\Psi}'_{\overline{z} \overline{z}, 1}(\bP) \right\vert
	- \sum_{z \neq \overline{z}} \left\vert (N-1) \xi_{\overline{z} z}^{N-2} \widetilde{\Psi}'_{\overline{z} z, 1}(\bP) \right\vert - \sum_{z = 1}^Z \left\vert \xi_{\overline{z} z}^{N-2} \widetilde{\Psi}'_{\overline{z} z, 2}(\bP) \right\vert \right)\\
	\overset{\eqref{eq:Rprimeequation}}{\ge} & (N-1) (N!)^{-1/2} |\xi_{\overline{z} \overline{z}}|^{N-2} \left( \left\vert \widetilde{\Psi}'_{\overline{z} \overline{z}, 1}(\bP) \right\vert
	- \frac{R'}{4} \left( \sum_{z \neq \overline{z}} \left\vert \widetilde{\Psi}'_{\overline{z} z, 1}(\bP) \right\vert + \sum_{z = 1}^Z \left\vert \widetilde{\Psi}'_{\overline{z} z, 2}(\bP) \right\vert \right) \right)\\
	\overset{\eqref{eq:Uprime}}{\ge} & (N-1) (N!)^{-1/2} |\xi_{\overline{z} \overline{z}}|^{N-2} \left( \left\vert \widetilde{\Psi}'_{\overline{z} \overline{z}, 1}(\bP) \right\vert
	- \frac{1}{2} \left\vert \widetilde{\Psi}'_{\overline{z} \overline{z}, 1}(\bP) \right\vert \right)
	> 0 \;.
\end{aligned}
\end{equation}
By the same argument as that below \eqref{eq:Psiineq1}, we conclude $ \Psi \neq 0 $ which contradicts \eqref{eq:WWAlincomb}.\\

It remains to establish a contradiction in the sub-case where $ \langle w_{\overline{k}, \bP}, \varphi_{\overline{k}} \rangle = 0 $ for almost all $ \bP \in \cQ_x $. In that case, the dominant term is no longer $ \widetilde{\Psi}'_{\overline{z} \overline{z}, 1}(\bP) $ (which is zero almost everywhere), but instead
\begin{equation}
	\widetilde{\Psi}'_{\overline{z} \overline{z}, 2}(\bP) = \xi_{\overline{z} \overline{z}} \langle w_{\overline{k}, \bP}, w_{\overline{k}, \bP} \rangle \;.
\label{eq:widetildePsi3}
\end{equation}
Since $ \Psi'_{m, \overline{k}} \neq 0 $, we have $ w_{\overline{k}, \bP} \neq 0 $ for all $ \bP $ inside some subset of $ \cQ_x $ with positive measure, which implies $ \widetilde{\Psi}'_{\overline{z} \overline{z}, 2}(\bP) \neq 0 $ for these $ \bP $. So we are in a similar situation as in Case 1 meaning that $ \widetilde{\Psi}'_{\overline{z} \overline{z}, 1} $ is zero almost everywhere while $ \widetilde{\Psi}'_{\overline{z} \overline{z}, 2} $ is not. We may thus copy the definitions of $ R $ and $ U $ and employ the same arguments as in \eqref{eq:ximax}--\eqref{eq:Psiineq1}, replacing $ N $ by $ N - 1 $ in the required positions. This yields the desired contradiction.\\

\end{proof}

\section{Proof of Lemma \ref{lem:AE}}
\label{sec:AE}

\begin{proof}
We evaluate the pullbacks of $ A(v) $ and $ E_\infty $, keeping in mind that by Lemma \ref{lem:W1equivalence}, $ \Psi = W(s) W_{\sF, 1}(\varphi) \Psi_m = W(s) W_1(\varphi) \Psi_m $.\\

\begin{itemize}
\item $ A(v) $: The commutator $ [W(s) W_1(\varphi),A(v)] $ is evaluated using the extended commutation relations established in Lemma \ref{lem:WAextended}
\begin{equation}
\begin{aligned}
	A(v) \Psi &= A(v) W(s) W_1(\varphi) \Psi_m\\
	&= [A(v), W(s) W_1(\varphi)] \Psi_m + W(s) W_1(\varphi) \underbrace{A(v) \Psi_m}_{=0}\\
	&= \big( M \langle v,s \rangle + \langle v,\varphi \rangle + V_{\bullet 1}(v^* \varphi) + V(v^* s) \big) \Psi \;.
\end{aligned}
\label{eq:APsi}
\end{equation}

\item $ E_\infty $: The self-energy operator \eqref{eq:E} now exactly removes the term $ M \langle v,s \rangle $, so we have
\begin{equation}
	(A(v) - E_\infty) \Psi = \big( \langle v,\varphi \rangle + V_{\bullet 1}(v^* \varphi) + V(v^* s) \big) \Psi \;.
\label{eq:AEinftyPsi}
\end{equation}

\end{itemize}

Now, by Definition \ref{def:WjsWjsX}, we may commute $ \langle v,\varphi \rangle, V $ and $ V_{\bullet 1} $ with $ W(s) $ yielding
\begin{equation}
\begin{aligned}
	W^{-1}(s) \big( A(v) - E_\infty \big) W(s) W_1(\varphi) \Psi_m &= \big( \langle v,\varphi \rangle + V_{\bullet 1}(v^* \varphi) + V(v^* s) \big) W_1(\varphi) \Psi_m\\
	&= \left( \res_1(\varphi) + V \right) W_1(\varphi) \Psi_m \quad \in \FB_{\ex} \;.
\end{aligned}
\end{equation}
Now, Lemma \ref{lem:W1equivalence} allows replacing $ W_1 $ by $ W_{\sF, 1} $, which renders the desired result.\\
\end{proof}

\paragraph{Remarks.}

\begin{enumerate}
\setcounter{enumi}{\theremarks}
\item If the form factor $ v $ depended on the fermion momentum $ \bp $, then $ s = -\frac{v}{\omega} $ would also depend on $ \bp $. In that case, $ W^{-1}(s) V W(s) = V $ will no longer hold true in general, as the function $ s(\bk,\bp) $ in the dressing operator $ W(s) $ then depends on $ \bp $ and fermion momenta are changed by $ V $. If $ (A(v) - E_\infty)W(s) W_{\sF, 1}(\varphi) \Psi_m $ is then expressed in terms of dressed coherent states, multiple commutators appear.\\
\end{enumerate}
\setcounter{remarks}{\theenumi}

\section{Proof of Lemma \ref{lem:H0A}}
\label{sec:H0A}

\begin{proof}
Again, we may use Lemmas \ref{lem:WjsWjs} and \ref{lem:W1equivalence}, to write
\begin{equation*}
	\Psi = W(s) W_{\sF, 1}(\varphi) \Psi_m = W(s) W_1(\varphi) \Psi_m \in \FB \;.
\end{equation*}
Following Proposition \ref{prop:operatorextension} we further have $ H_{0, y} \Psi, A^{\dagger}(v) \Psi \in \FB $. We first evaluate $ (H_{0, y} + A^\dagger(v)) W(s) W_1(\varphi) \Psi_m $ in momentum space and then use Definition \ref{def:WjsWjsAdagger} to pull $ W(s) $ to the left.

\begin{itemize}
\item $ H_{0, y} $: The important point is that the expression $ (H_{0, y} + A^{\dagger}(v)) W(s) $ cancels terms in $ W(s) H_{0, y} $. We therefore investigate the commutator expression $ [H_{0, y}, W(s)] W_1(\varphi) \Psi_m $ and compare it to $ A^\dagger(v) W(s) W_1(\varphi) \Psi_m $.\\

First, we note that the application of $ H_{0, y} $ to a dressed state just changes one single photon dispersion relation. Therefore, it is equivalent to applying a creation operator to the dressed state
\begin{equation}
\begin{aligned}
	&(H_{0, y} W_1(\varphi)\Psi_m)(\bP,\bK)\\
	= &\frac{e^{-\frac{\Vert \varphi\Vert^2}{2}}}{\sqrt{N!}} \left( \sum_{\ell' = 1}^N \omega(\bk_{\ell'}) \right) \left( \prod_{\ell = 1}^N \varphi(\bk_{\ell}) \right) \Psi_{mx} \left( \bP' \right)\\
	= &\frac{e^{-\frac{\Vert \varphi\Vert^2}{2}}}{\sqrt{N!}} \sum_{\ell = 1}^N  \omega(\bk_{\ell}) \varphi(\bk_{\ell}) \left( \prod_{\ell' \neq \ell}^N \varphi(\bk_{\ell'}) \right) \Psi_{mx} \left( \bP' \right)\\
	= &(A_1^\dagger(\omega \varphi) W_1(\varphi)\Psi_m)(\bP,\bK) \;,\\
\end{aligned}
\label{eq:H0W}
\end{equation}
with $ \bP' = \big( \bp_1 + \sum_{\ell' = 1}^N \bk_{\ell'}, \bp_2, \ldots, \bp_M \big) $. Replacing the dressing $ W_1(\varphi) $ by $ W(s) W_1(\varphi) $ and applying the commutation relations from Definitions \ref{def:WjsWjsAdagger} and \ref{def:WjsWjsX}, we obtain
\begin{equation}
\begin{aligned}
	&H_{0, y} W(s) W_1(\varphi) \Psi_m\\
	= &(A^\dagger(\omega s) + A_1^\dagger(\omega \varphi)) W(s) W_1(\varphi) \Psi_m\\
	= &-A^\dagger(v) W(s) W_1(\varphi) \Psi_m\\
	 &+ W(s) W_1(\varphi) \big(A_1^\dagger(\omega \varphi) - \langle v, \varphi \rangle - V_{\bullet 1}(v^* \varphi) + \langle \varphi, \omega \varphi \rangle \big) \Psi_m \;.
\end{aligned}
\label{eq:H0WW}
\end{equation}
Here, we used $ s^* \omega = \omega s = -v $ and Definition \ref{def:WjsWjsX}, which allows us to pull the formal scalar products and the $ V_{\bullet 1} $-term past the dressing operators.\\
By means of \eqref{eq:H0W} and the commutation relations from Definitions \ref{def:WjsWjsAdagger} and \ref{def:WjsWjsX}, we also have
\begin{equation}
\begin{aligned}
	&W(s) H_{0, y} W_1(\varphi) \Psi_m\\
	= &W(s) W_1(\varphi) A_1^\dagger(\omega \varphi) \Psi_m
	+ W(s) W_1(\varphi) \langle \varphi, \omega \varphi \rangle \Psi_m \;.
\end{aligned}
\label{eq:WH0W}
\end{equation}
Combining \eqref{eq:H0WW} and \eqref{eq:WH0W}, we finally obtain
\begin{equation}
\begin{aligned}
	&[H_{0, y}, W(s)] W_1(\varphi) \Psi_m\\
	= &-A^\dagger(v) W(s) W_1(\varphi) \Psi_m - W(s) W_1(\varphi) \big( \langle v, \varphi \rangle + V_{\bullet 1}(v^* \varphi) \big) \Psi_m \;.
\end{aligned}
\label{eq:H0Wcommutator}
\end{equation}

\item $ A^\dagger(v) $: Here, we do not need to perform any calculations. The appearing term simply cancels the $ -A^\dagger(v) W(s) W_1(\varphi) \Psi_m $ from \eqref{eq:H0Wcommutator}.

\end{itemize}

Now, adding both terms and using Definition \ref{def:WjsWjsX}, \eqref{eq:H0Wcommutator} amounts to
\begin{equation}
\begin{aligned}
	&\big( [H_{0, y}, W(s)] + A^\dagger(v)  W(s) \big) W_1(\varphi) \Psi_m\\
	= &W(s) W_1(\varphi) \big( - \langle v, \varphi \rangle - V_{\bullet 1}(v^* \varphi) \big) \Psi_m\\
	= &W(s) \big( - \res_1(\varphi) \big) W_1(\varphi) \Psi_m \;.\\
\end{aligned}
\end{equation}
This is equivalent to
\begin{equation}
	\big( H_{0, y} + A^\dagger(v) \big) W(s)  W_1(\varphi) \Psi_m = W(s) \big( H_{0, y} - \res_1(\varphi) \big) W_1(\varphi) \Psi_m \;.
\end{equation}
Lemma \ref{lem:W1equivalence} allows again to replace $ W_1 $ by $ W_{\sF, 1} $, which yields the final result.\\

\end{proof}

\end{appendices}

\bigskip



\noindent\textit{Acknowledgments.}
The idea of introducing a Fock space extension arose within numerous discussions about mathematical QFT. Important information and/or relevant literature was pointed out by Stefan Teufel, Jonas Lampart, Gianluca Panati, Gian Michele Graf, Jürg Fröhlich, Volker Bach, Wojciech Dybalski, Micha\l\; Wrochna, Kasia Rejzner, Holger Gies, Jan Mandrysch, Joscha Henheik, Martin Oelker, as well as some further members of the LQP2 network and at Tübingen. In particular, I am grateful to Roderich Tumulka for valuable suggestions and discussions concerning the Extended state space construction.\\
This work was financially supported by the Wilhelm Schuler--Stiftung T\"ubingen, by DAAD (Deutscher Akademischer Austauschdienst), by the Basque Government through the BERC 2018--2021 program and by the Ministry of Science, Innovation and Universities: BCAM Severo Ochoa accreditation SEV--2017--0718, as well as by the European Union (ERC FermiMath, grant agreement nr. 101040991). Views and opinions expressed are however those of the author(s) only and do not necessarily reflect those of the European Union or the European Research Council Executive Agency. Neither the European Union nor the granting authority can be held responsible for them. Moreover, it was partially supported by Gruppo Nazionale per la Fisica Matematica in Italy.

\end{document}